%% file: main.tex
\documentclass[
12pt,
prx,
aps,
reprint,
a4paper,
longbibliograph,
nofootinbib,
nobibnotes,
superscriptaddress
]{revtex4-2}
\usepackage[utf8]{inputenc}
\usepackage[T1]{fontenc}
\usepackage{hyperref}
\usepackage{amsmath,amssymb,amsthm,amsfonts,mathtools}
\usepackage[dvipsnames]{xcolor}
\hypersetup{
    colorlinks = true,
    citecolor=PineGreen,
    linkcolor=MidnightBlue,
    urlcolor=PineGreen
}

\usepackage{tikz}
\usetikzlibrary{positioning,arrows, calc}
\usetikzlibrary{decorations.pathmorphing}
\usetikzlibrary{decorations.markings}
\tikzset{
 overlap/.style args = {#1/#2}{minimum height=#1,
               minimum width=#2,
               rounded corners=2.2mm, fill=black!50, opacity=0.6, sloped},
               party/.style = {draw, circle, fill=black, minimum width=1mm},
obi/.style={black!50, decorate, decoration={
            snake,
            segment length=2mm,
            amplitude=0.2mm}
            }
}
\tikzset{shifted path/.style args={from #1 to #2 by #3}{insert path={
let \p1=($(#1.east)-(#1.center)$),
\p2=($(#2.east)-(#2.center)$),\p3=($(#1.center)-(#2.center)$),
\n1={veclen(\x1,\y1)},\n2={veclen(\x2,\y2)},\n3={atan2(\y3,\x3)} in
(#1.{\n3+180+asin(#3/\n1)}) to (#2.{\n3-asin(#3/\n2)})
}}}

\theoremstyle{definition}
\newtheorem{theorem}{Theorem}

\newtheorem{proposition}{Proposition}

\newtheorem{lemma}[theorem]{Lemma}
\newtheorem{definition}[theorem]{Definition}
\newtheorem{prop}{Property}

\DeclareMathOperator\tr{Tr}
\newcommand\mean[1]{\left\langle #1 \right\rangle}

\def\id{\mathbb{I}}%

\def\conv{\mathrm{Conv}}
\def\stab{\operatorname{STAB}}

\def\be{\operatorname{BETA}}

\begin{document}

\title{Simultaneous variances of Pauli strings, weighted independence numbers,\protect\\ and a new kind of perfection of graphs}

\author{Zhen-Peng Xu}\email{zhen-peng.xu@ahu.edu.cn}
\affiliation{School of Physics, Anhui University, Hefei 230601, China}
\author{Jie Wang}\email{wangjie212@amss.ac.cn}
\affiliation{State Key Laboratory of Mathematical Sciences, Academy of Mathematics and Systems Science, Chinese Academy of Sciences, Beijing, China}
\author{Qi Ye}
\affiliation{Institute for Interdisciplinary Information Sciences, Tsinghua University, Beijing 100084, China}

\author{Gereon Ko{\ss}mann}
\affiliation{Institute for Quantum Information, RWTH Aachen University, 52074 Aachen, Germany}

\author{Ren\'{e} Schwonnek}\email{rene.schwonnek@itp.uni-hannover.de}
\affiliation{Institut f\"{u}r Theoretische Physik, Leibniz Universit\"{a}t Hannover, 30167 Hannover, Germany}

\author{Andreas Winter}\email{andreas.winter@uni-koeln.de}
\affiliation{Department Mathematik/Informatik---Abteilung Informatik,\protect\\ Universit\"at zu K\"oln, Albertus-Magnus-Platz, 50923 K\"oln, Germany}
\affiliation{ICREA {\&} Grup d'Informaci\'{o} Qu\`{a}ntica, Departament de F\'{\i}sica,\protect\\ Universitat Aut\`{o}noma de Barcelona, 08193 Bellaterra (Barcelona), Spain}
\affiliation{Institute for Advanced Study, Technische Universit\"at M\"unchen,\protect\\ Lichtenbergstra{\ss}e 2a, 85748 Garching, Germany}

\begin{abstract}
A set of Pauli stings is well characterized by the graph that encodes its commutatitivity structure, i.e., by its frustration graph.  
This graph  provides a natural interface between graph theory and quantum information, which we explore in this work.
We investigate all aspects of this interface for a special class of graphs that bears tight connections between the groundstate structures of a spin systems and topological structure of a graph. We call this class $\hbar$-perfect, as it extends the class of  perfect and $h$-perfect graphs.

Having an $\hbar$-perfect graph opens up several applications: we find efficient schemes for entanglement detection, a connection to the complexity of shadow tomography, tight uncertainty relations and a construction for computing good lower on bounds ground state energies.  Conversely this also induces quantum algorithms for computing the independence number. Albeit those algorithms do not immediately promise an advantage in runtime, we show that an approximate Hamilton encoding of the independence number can be achieved with an amount of qubits that typically scales logarithmically in the number of vertices. 
We also we also determine the  behavior of $\hbar$-perfectness under basic graph operations and evaluate their prevalence among all graphs.
\end{abstract}

\date{14 November 2025}

\maketitle

\section{Introduction}
The beauty of graphs is that they often arise in the most unexpected situations as just the right structure to describe an intriguing problem. It is thus no surprise that their appearances in quantum information science are also plentiful. This includes, for example, graph states in quantum entanglement~\cite{hein2006entanglement,audenaert2005entanglement}, orthogonality graphs in quantum contextuality~\cite{cabello2014graph}, graphs in surface codes in quantum error correction \cite{Sarkar2024} and many more. 
In the present work, we investigate a special class of graphs that arise as anti-commutativity (``frustration'') graphs of spin systems~\cite{sachdev1993gapless,chapman2023unified}. These graphs, which we will call $\hbar$-perfect, have the beauty of tightly connecting a long list of physical properties of a spin system, including uncertainty relations, the ground state structure of Hamiltonians, and the complexity of shadow tomography, to its topological structure.  

A central concept in the investigation of graphs are graph properties invariant under isomorphism, and more specifically characteristic parameters, such as the chromatic number, the clique number, or the independence number~\cite{Diestel2025GraphT}. 
For others, see the encyclopedic list in \cite{ISGCI}. 
{These graph parameters allow us to sort graphs into categories and classify them by relations between different parameters, which finally facilitates the applications of graph theory.}
A common example for this procedure are perfect graphs~\cite{berge1961farbung}. 
They are those for which the chromatic number and the clique number coincide for any induced subgraphs. Considering such subclasses of graphs often turns out to be fruitful, as they usually come with their own heap of methods and theorems that are available and valid only here rather than in full generality.
This often also includes properties of a physical system described by the graph.
For example, we have that the orthogonality graph of measurement directions is perfect if and only if  there is no quantum contextuality~\cite{cabello2014graph}.

We follow a similar procedure here. We consider the weighted independence number $\alpha(G,w)$ and the recently introduced weighted beta number $\beta(G,w)$ \cite{xu2023bounding}, asking for graphs where these numbers are equal for all weights $w$: due to their similarity with perfect and $h$-perfect graphs~\cite{fonlupt1982transformations}, we will refer to such graphs as $\hbar$-perfect. 
Whereas $\alpha(G,w)$ encodes information on the topology of a graph, $\beta(G,w)$ arises as an invariant of its possible representations on spin systems.

The scope of this work is to understand basic properties of $\hbar$-perfect graphs, provide tools for handling them analytically and numerically, and shine a light on several spots in quantum information theory where they seem to give just the right structure for simplifying generally difficult problems. However, we can only make a beginning, but hope to interest the reader in the numerous questions our work raises.

\subsection{Background and motivation}
{The independence number $\alpha(G)$ is a pivotal parameter in graph theory with profound implications across theoretical and applied domains. {Remarkably, as a basic character of a graph, it determines not only the zero-error capacity of a classical communication channel~\cite{Shannon1956TheZE,Alon2006TheSC}, but also a clear separation between classical models and quantum theory in Bell scenarios~\cite{cabello2014graph}.}
Theoretically, its equivalence to the clique number of the complement graph yields the fundamental lower bound $\alpha(\bar{G})$ for the chromatic number $\chi(G)$~\cite{Diestel2025GraphT}, a core concept behind results like the four color map theorem. 
Besides, perfect graphs $G$ can be fully characterized by weighted independence numbers $\alpha(G,w)$ together with weighted Lov\'asz numbers $\vartheta(G,w)$ as those satisfying $\alpha(G,w)=\vartheta(G,w)$ for any non-negative weight vector $w$~\cite{lovasz2019Graphs}.
In application, Claude Shannon established $\alpha(G)$ as the key parameter determining the zero-error capacity of a graph in communication theory~\cite{Shannon1956TheZE}. Computationally, the problem of finding $\alpha(G)$, known as the maximum independent set problem, is not only NP-hard itself but serves as a root of intractability~\cite{Karp1972ReducibilityAC}. Many other NP-hard problems, including minimum vertex cover and maximum clique, are polynomially reducible to it, underscoring its central role in classifying computational complexity~\cite{Cormen2001IntroductionTA}.} 

Very recently, the beta number, as a characteristic parameter of graphs, and its relation to the independence number have been exhibited in a series of works~\cite{de2023uncertainty,hastings2022optimizing,xu2023bounding,moran2024Uncertainty}.
For a given set of Pauli observables ${S_1, \dots, S_n}$ acting on multiple qubits, i.e. a spin system, the graph $G$ that collects their commutation/anti-commutation relations is often also called the anti-commutation graph~\cite{gokhale2019minimizing} or the frustration graph~\cite{chapman2020characterization,chapman2023unified}. 
In Ref.~\cite{xu2023bounding}, we formally introduced the number $\beta(G,w)$ as the smallest generalized radius of an ellipsoid whose principal axes relate to each other according to a positive vector $w$, containing the set of attainable expectation value tuples $(\langle S_1 \rangle_\rho,\dots,\langle S_n \rangle_\rho)$ for arbitrary states $\rho$, see Figure~\ref{fig:jnrvsbeta}. There is, however, also a purely algebraic definition \eqref{eq:beta}.
A geometric intuition behind this number can be gained from the task of generalizing the picture of a Bloch sphere, which represents the possible expectation values triples of three Pauli operators $X$, $Y$, and $Z$, to larger number of Pauli strings and concomitant higher dimensions.
As it turns out, $\hbar$-perfect graphs are exactly those graphs for which this picture generalizes in a meaningful manner; see Section~\ref{sec:def}.

Without an explicit name, this number has been considered before in the contexts of 
uncertainty relations \cite{de2023uncertainty} and   algebraic relaxations to the Hamiltonian ground state problem \cite{hastings2022optimizing} in the Sachdev-Ye-Kitaev (SYK) model~\cite{sachdev1993gapless}. 
Moreover, a general connection between $\beta(G,w)$ and the ground state problem of a many-body Hamiltonian can be drawn \cite{xu2023bounding,hastings2022optimizing}.
Very recently, a quantity closely related to the beta numbers has turned out to be  the key parameter in the sample complexity of shadow tomography for Pauli strings,  implying more possible applications~\cite{king2024triplyefficientshadowtomography,chen2024Optimal}.

The notion of $\hbar$-perfect graphs hence represents a fertile link between different fields as illustrated in Fig.~\ref{fig:structure}: it connects central problems in many-body physics, sample complexity in shadow tomography, detection of entanglement structure, and quantum uncertainty relations with weighted independence numbers of a graph, which are well understood and broadly investigated. 
This link opens up the possibility of a transfer of methods, tools, and results. On the one hand, we can employ tools for calculating the independence number in order to get ground state energies and other interesting many-body quantities. This will be especially useful when dealing with perfect graphs, which are a notable subclass of $\hbar$-perfect graphs.
On the other hand, we can use classical methods developed in many-body physics, or even algorithms running on quantum computers, in order to now tackle a graph problem.  

A central question of this work is to determine for which graphs the link can be exploited, specifically to characterize the set of $\hbar$-perfect graphs and explore their applications. A conjecture presented in \cite{de2023uncertainty} suggests that all graphs should be $\hbar$-perfect. However, this conjecture is refuted, with the smallest counterexample of $\hbar$-imperfectness provided by the anti-heptagon~\cite{xu2023bounding}.
\begin{figure}
    \centering
    \includegraphics[width=0.85\linewidth]{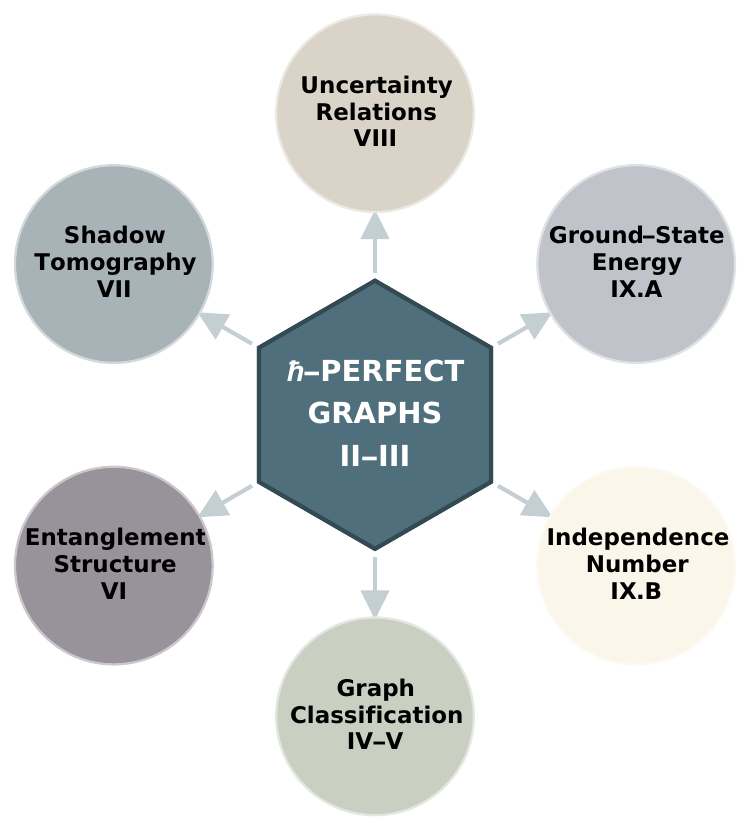}
    \caption{The graph-theoretic framework centered on $\hbar$-perfect graphs and its various applications distributed in each section.}
    \label{fig:structure}
\end{figure}

\subsection{Outline and main results}
The fact that not all graphs are  $\hbar$-perfect motivates a closer look at those which are. 
We start by introducing the exact definition of $\hbar$-perfect graphs. 
As a key technical contribution, 
we present a clear geometric picture capturing the interplay of the beta numbers $\beta(G,w)$ and independence numbers $\alpha(G,w)$ in Proposition~\ref{ob:beta} and Theorem~\ref{ob:stabbe}.  

We then turn to the behavior of $\hbar$-perfectness under typical graph operations. A list of six operations that preserve this property is collected in section~\ref{sec:props}. This is accompanied by the basic section~\ref{thm:perfect}, where we observe that all perfect and $h$-perfect graphs are also $\hbar$-perfect. When combined in a constructive manner, all these properties can serve as criteria for determining $\hbar$-perfectness of a given graph $G$. 

Unfortunately, these criteria are not sufficient in all cases. To address this problem, we introduce practical numerical methods for testing $\hbar$-perfectness. At their core stands the problem of computing $\beta(G,w)$.  
We present a toolbox containing three numerical methods for this.

(i) A complete semidefinite programming hierarchy bounding $\beta(G,w)$ from above, based on state polynomial optimization (see \cite{Klep2023StatePP}) together with a simplified version which extends the Lov\'asz-type hierarchy in \cite{moran2024Uncertainty}. 

(ii) A mean-field type approximation that gives a converging outer hierarchy of eigenvalue problems towards the beta number (Theorem~\ref{thm:deFinetti_approximation}). Moreover, using recent developments in the symmetrization of SDP's, we show that each level of the hierarchy can be computed in polynomial time in the level of the hierarchy.

(iii) A see-saw method method that complements the other methods by inner bounds. Unlike typical such approaches, here we are able to warm-start our see-saw with the outer approximations of previous methods. This leads to a reliable gap estimate. 

We then apply these methods to analyze all graphs with $n \leq 9$ vertices for $\hbar$-perfectness. The results reveal that all graphs with $n \leq 6$ vertices are $\hbar$-perfect, with only one $\hbar$-imperfect graph for $n = 7$ (up to isomorphism), and a small number of $\hbar$-imperfect graphs for $n = 8$ and $n = 9$ (more details in Section~\ref{sec:fre}).
Together with the spirit of the conjecture of \cite{de2023uncertainty} these findings may suggest that $\hbar$-perfect graphs are typical. We consider this question in more detail in  section~\ref{sec:fre}. As it turns out, the probability that a random graph on $n$ vertices is $\hbar$-perfect tends to zero as $n$ goes to infinity. A concrete upper bound on the scaling of this probability, exponential in $n^2$, i.e.~as a power of the total number of graphs, is formulated in Theorem~\ref{thm:haufigkeit}. 

Knowing that a particular graph of interest is  $\hbar$-perfect opens up substantial simplification in several applications, which we summarize in the following.
\begin{enumerate}
    \item \emph{Entanglement detection---}We develop a method for detecting entanglement in section~\ref{sec:entanglement} that, as we show, goes beyond the capabilities of linear witnesses, as it can detect more entangled states with less measurement resources. Building on this, we show how to  transform each graph state into an effective criterion for assessing multipartite entanglement structures, and estimate measures of entanglement from the violation of the criterion.
    \item \emph{--Sample complexity of shadow tomography--} We carry out the explicit formulation of the sample complexity parameter of shadow tomography in section~\ref{sec:shadow} for Pauli strings when the underlying frustration graph is $\hbar$-perfect, and provide lower and upper bounds in the general case. 
    \item \emph{Uncertainty relations---}For measurements arising from a set of Pauli-strings yielding a $\hbar$-perfect graph, we show that the computation of variance-based quantum uncertainty relations in section~\ref{sec:uncertainty_relations} turns into a linear program over the stable set polytope (see also~\cite{schwonnek2018uncertainty,kaniewski2019maximal}).  
    \item \emph{Ground state energy estimates---}Consider a graph for which the computation of $\alpha(G,w)$ is feasible, $\hbar$-perfectness of this graph now allows us to compute properties of spin-systems. In detail, we show in section~\ref{sec:app} how the ground state energy of a variant of a $XZ$-Hamiltonian can be approximated in this way. 

    \item \emph{Qubit encoding of the independence number---}As $n$ qubits have $4^n$ different Pauli strings, graphs with up to $n \sim 4^n$ vertices can be represented.
    Mapping a graph $G$ into the commutativity structure of a set of Pauli strings can  hence be a quite efficient quantum encoding. This  becomes  evident when compared to the typical one qubit to vertex encoding \cite{Farhi2014AQA}. 
    Nevertheless, a dense encoding always bears the challenge of potentially making relevant quantities of an object inaccessible. For example, we do not know how to formulate typical benchmark problems like max-cut on this level.
    A central result of this work is the insight that the independence number is an exception to this. In Section \ref{ssec:encoding} we show how to formulate the computation of $\alpha(G,w)$ as an approximate ground state problem of a single Hamiltonian. In an NISQ friendly setting we only need $O(\log(n)^2/ \epsilon^2) $ qubits for this. Given full stack quantum computing resources, this bound can even be improved further. The according computation leads to a quantum inspired classical algorithm that can compute the exact value of $\alpha(G,w)$ for $\hbar$-perfect graphs in a subexponential time $O(2^{c \sqrt{n}})$. In these instances we can hence improve on the best known general classical runtime scaling of $O(1.1996^n)$ exponentially.

    \item \emph{Generalized beta number---}In Proposition~\ref{prop:convergence_generalized_beta}, we consider a generalization $\beta(G,w,k)$, arising from approximating the joint numerical range by algebraic surfaces of order $k$. We show that the computation of this number will retrieve the independence number $\alpha(G,w)$ for the limit $k$ going to infinity for all graphs. This gives rise to qubit efficient methods for estimating independence numbers for $\hbar$-imperfect graphs. 
\end{enumerate}

The rest of the article is organized as follows. Basic definitions and concepts are stated in Section~\ref{sec:def}. We explore the properties of $\hbar$-perfect graphs in Section~\ref{sec:props}, develop numerical tools in Section~\ref{sec:numerical}, and investigate the statistical prevalence of $\hbar$-perfect graphs in Section ~\ref{sec:fre}. In Section~\ref{sec:entanglement} we explore quantum entanglement. In Section~\ref{sec:shadow}, we solve the sample complexity parameter in shadow tomography. In Section~\ref{sec:app}, we consider further implications of $\hbar$-perfect graphs. Including the task of estimating ground state energies from the clique structure of a graph, and encoding the independence number of a graph into a quantum system.  
We finally display our conclusions and outlooks in Section~\ref{sec:conclusion}.

\section{Definitions}
\label{sec:def}
In the following, we first motivate $\hbar$-perfect graphs from realizations as anti-commutation graphs for Pauli strings and then extract an abstract definition. In the following, we denote by $\mathcal{S}(\mathcal{H})$ the set of density matrices and by $\langle A \rangle_\rho$ the expectation value of the observable $A$ w.r.t. $\rho \in \mathcal{S}(\mathcal{H})$. A Pauli string $S_i$ of length $m$ is a tensor product of Pauli matrices $X, Y, Z$ or $\id$ with $m$ factors. It is straightforward to check that each pair of Pauli strings either commutes or anti-commutes. 
For a given set of Pauli strings $\mathcal{S}=\{S_1,\dots,S_n\}$, its frustration graph $G$ is defined as a graph with $n$ vertices that has an edge between vertices $i$ and $j$ if the strings $S_i$ and $S_j$ anti-commute and has no edge otherwise~\cite{chapman2020characterization}. Corresponding to a concrete realization of a frustration graph $G$ with Pauli strings $\mathcal{S}$, we make the following definition.

\begin{definition}[Weighted beta number~\cite{xu2023bounding}] 
For a given frustration graph $G$ w.r.t. $\{S_i\}$ and a vector $w \in \mathbb{R}^n_+$, we define the weighted beta number as 
	\begin{equation}\label{eq:beta}
        \beta(G,w) \coloneqq  \sup_{\rho \in \mathcal{S}(\mathcal{H})} \sum\nolimits_{i} w_i\langle S_i \rangle_\rho^2.
    \end{equation}
\end{definition}
From the definition, it is not immediately clear that $\beta(G,w)$ is indeed a universal object of the graph $G$ and not dependent on the concrete realization $\mathcal{S}$. 
However, in \cite{xu2023bounding}  it is shown that any two sets of Pauli strings $\mathcal{S}$ and $\mathcal{S}^\prime$ sharing the same frustration graph $G$, lead to the same value of the right hand in Eq.~\eqref{eq:beta}. This justifies calling the weighted beta number a graph parameter. Besides, a realization is said to be basic if the Pauli strings in the realization are of shortest length.

Given the perspective of optimizing the sum of variances in \eqref{eq:beta}, we could alternatively consider the weighted beta number as an optimization w.r.t. the following set
    \begin{equation}
	    {\cal Q}(\{S_i\}) \coloneqq \big\{ \big(\mean{S_1}_{\rho}^2, \dots, \mean{S_n}_{\rho}^2\big) \mid \rho\in \text{States} \big\}.
   \end{equation}
The set is termed the simultaneous (or joint) variance range of $\{S_1,\ldots,S_n\}$ here and the weighted beta number is the support function of this set with respect to the parameter $w \in \mathbb{R}^n_+$.

Compared to the fact that $\beta(G,w)$ is a graph parameter,
${\cal Q}(\{S_i\})$ depends on the concrete realization of the frustration graph.
For example, ${\cal S}_1 = \{XX,YY,Z Z\}$ and ${\cal S}_2= \{XX\id, YY\id, Z Z Z\}$ share the same frustration graph with three isolated vertices.
However, the point $(1,1,0)$ is in ${\cal Q}({\cal S}_2)$ but not in ${\cal Q}({\cal S}_1)$, as illustrated in Fig.~\ref{fig:jnrvsbeta}(d). 
But, since support functions only depend on the bipolar of the underlying set, we conclude directly that the bipolar of ${\cal Q}(\{S_i\})$ is independent of the concrete realization, which we call $\be(G)$ and prove a constructive description of it in the next proposition.

\begin{figure}[htpb]
	\centering
	(a)\includegraphics[width=0.2\textwidth]{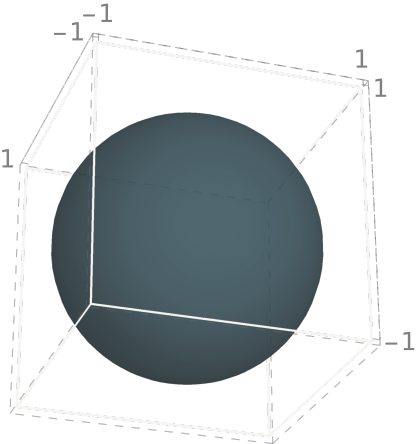}
	(b)\includegraphics[width=0.2\textwidth]{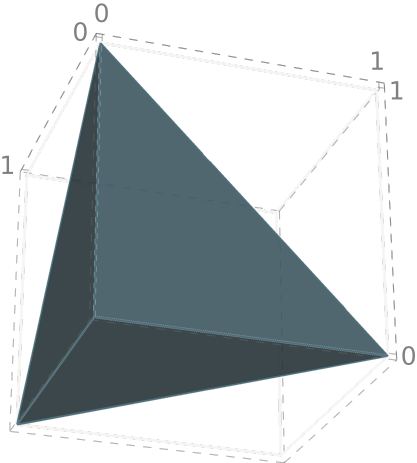}\\
	\hspace{2em} $(\langle X\rangle,\langle Y\rangle, \langle Z\rangle)$\hspace{6em} 
	$(\langle X\rangle^2,\langle Y\rangle^2, \langle Z\rangle^2)$\\ \vspace{1em}
	(c)\includegraphics[width=0.2\textwidth]{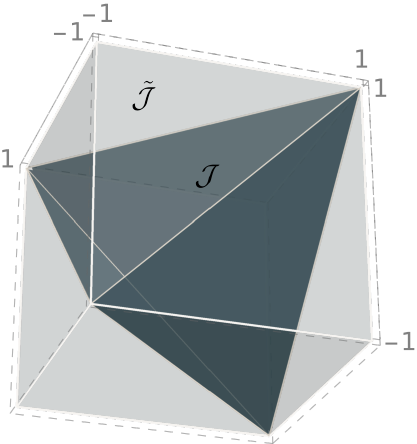}
	(d)\includegraphics[width=0.2\textwidth]{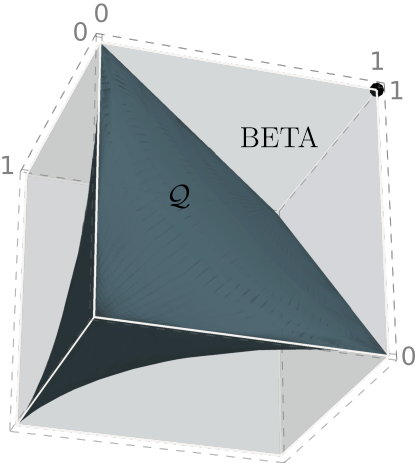}\\
	\hspace{1.5em} $(\langle XX\rangle,\langle YY\rangle, \langle ZZ\rangle)$\hspace{4em} 
	$(\langle XX\rangle^2,\langle YY\rangle^2, \langle ZZ\rangle^2)$\\ \vspace{0.5em}
	\caption{The joint numerical range ${\cal J}$, the convex hull $\tilde{\cal J}$ of the flips of ${\cal J}$ along each axis, its square ${\cal Q}$ and beta body $\be$ for Pauli strings $\{X, Y, Z\}$ and $\{XX, YY, Z Z\}$, where the beta body $\be$ is the convex hull of the corner generated by ${\cal Q}$. For $\{X, Y, Z\}$,  $\tilde{\cal J} $ coincides with ${\cal J}$ and $\be $ coincides with ${\cal Q}$. For $\{XX, YY, Z Z\}$, $\tilde{\cal J}$ (c, in gray, the cube) is strictly larger than ${\cal J}$ (c, in blue), and $\be$ (d, in gray, the cube) is strictly larger than ${\cal Q}$ (d, in blue, not convex, not containing the point $(1,1,0)$). Such a $\be$ coincides with the ${\cal Q}$ of $\{XX\id, YY\id, ZZZ\}$. In each of these two cases the frustration graph is $\hbar$-perfect, and $\tilde{J}$ is fully characterized by the intersections of the ellipsoids defined by the facets of $\be$.}
	\label{fig:jnrvsbeta}
\end{figure}

\begin{proposition}
\label{ob:beta}
For a frustration graph $G$, any realization with Pauli strings ${\cal S}$ leads to the same set
\begin{equation}
	\be(G) \coloneqq \conv(\downarrow\hspace{-0.3em} {\cal Q}({\cal S})) \cap \mathbb{R}_+^n,
\end{equation}
where we define $\downarrow\hspace{-0.3em}T \coloneqq \{x \in \mathbb{R}^n \,|\,\exists y \in T\ \mathrm{s.t.}\ y - x \in \mathbb{R}_+^n\}$ for a set $T \subseteq \mathbb{R}^n$. This is a so-called convex corner. 
Consequently, we have
\begin{align}
        \beta(G,w) = \max_{v\in \be(G)} \sum_i w_i v_i, \quad \text{for all} \  w\in \mathbb{R}_+^n
\end{align}
\end{proposition}	
\begin{proof}
See Appendix~\ref{sec:proofs1}.
\end{proof}
At the core of the definition of $\hbar$-perfect graphs is the comparison of the weighted beta number with the weighted independence number of a graph, which we define next. For a graph with a vertex set $V$, a subset $V' \subseteq V$ is called independent or stable if no two vertices in $V'$ are adjacent. The independence number $\alpha(G)$ denotes the maximum size of the independent set. For a given $w$, the weighted independence number is given by 
\begin{align}
    \alpha(G,w)\coloneqq\max_{V'\subseteq V \text{ independent}}   \sum\nolimits_{i\in V'} w_i.
\end{align}
Here we can also take the alternative perspective and regard $\alpha(G,w)$ as the support function of the polytope  
\begin{align}
    \stab(G)=\conv\bigl( \{(v_1,\dots,v_n)| \substack{ v_iv_j=0  \text{ if } i\sim_G j ,\\ v_i\in\{0,1\} } \} \bigr),
\end{align}
which is commonly referred to as the stable set polytope~\cite{lovasz2019Graphs}. 

The fact that $\alpha(G,w) \le \beta(G,w)$ for any $w \in \mathbb{R}_+^n$ implies the following observation.
\begin{theorem}
\label{ob:stabbe}
    Given a frustration graph $G$, we have the inclusion
	$\stab(G) \subseteq \be(G)$.
\end{theorem}
\begin{proof}
    See Appendix~\ref{sec:proofs1}.
\end{proof}
With the above result in our hands, we can define the main object of investigation by 
\begin{definition}[$\hbar$-Perfectness]
\label{basicdef}
    A graph $G$ is called $\hbar$-perfect if 
    $\stab(G) = \be(G)$, or equivalently, for any $w\in \mathbb{R}_+^n$,
       $\beta(G,w) = \alpha(G,w)$.
\end{definition}
An  object closely related to $\cal{Q}$ is the joint numerical range 
\begin{equation}
{\cal J}(\{S_i\}) =  \big\{ \big(\mean{S_1}_{\rho}, \dots, \mean{S_n}_{\rho} \big) \mid \rho\in \text{States} \big\},
\end{equation}
i.e., the set that contains all tuples of expectation values that could be measured on a state $\rho$ using the available Pauli strings. It effectively encodes all the information about the quantum mechanical state space that could be observed in a measurement of a set of Pauli stings. Note that, for example, the expectation value of a Hamiltonian which is a linear combination of the given set of Pauli strings can be regarded as a linear functional acting on this object. Hence, all the information of ground state energies is encoded in the geometry of this set. 
Flipping the signs (e.g. by declaring up to be down and down to be up) in the outcome of a Pauli string measurement is a natural symmetry of the underlying physics. On the level of the joint numerical range, this corresponds to a reflection along the corresponding axis. 
Denote by $\tilde{\cal J}(\{S_i\})$ the convex hull of the flips of ${\cal J}(\{S_i\})$ along all axes.
Then the weighted beta number as the maximum of a linear function on $\be(G)$, i.e., $\sum_i \langle S_i\rangle^2 \le \beta(G,w)$, defines an ellipsoid containing $\tilde{\cal J}({\cal S})$ and ${\cal J}({\cal S})$. The intersection of all such ellipsoids leads to an approximation of $\tilde{\cal J}({\cal S})$. 
For the set of Pauli strings $\cal S = \{X, Y, Z\}$, the geometric object $\tilde{\cal J}({\cal S})$ coincides with ${\cal J}({\cal S})$, which is a unit ball. The surface is known as the Bloch sphere, as illustrated in Fig.~\ref{fig:jnrvsbeta}(a), which corresponds to the only non-trivial linear constraint of $\be(G)$ that $\langle X\rangle^2+\langle Y\rangle^2+\langle Z\rangle^2 \le 1$. For any ${\cal S}$ with $\hbar$-perfect frustration graph $G$, $\tilde{\cal J}({\cal S})$ is determined by the finite set of ellipsoids defined by the facets of $\be(G)$.

\section{\texorpdfstring{Properties of $\mathbf{\hbar}$-perfect graphs}{}}
\label{sec:props}
In the following, we collect some basic properties of $\hbar$-perfect graphs. 
Properties 1-6 concern the behavior of $\hbar$-perfect graphs under typical graph operations. Then the last property formulated as Theorem~\ref{thm:perfect} sets a relation to other classes of graphs. These properties can be used in a constructive manner for testing $\hbar$-perfectness. For example, these properties resolve the $\hbar$-perfectness of all but two of $853$ non-isomorphic graphs with $7$ vertices. One of the remaining graphs is the $\hbar$-imperfect graph anti-heptagon, and another one is proven to be $\hbar$-perfect in another manner in Appendix~\ref{ssec:hper2}.

To start, the union (see Fig.~\ref{fig:prod}) 
of $\hbar$-perfect graphs gives again an $\hbar$-perfect graph. Conversely, $\hbar$-perfectness is also inherited by induced subgraphs. We have:\footnote{All the proofs of the properties in this section are provided in Appendix~\ref{sec:proofs1}.}
\begin{prop}[Fully connected union]
\label{thm:connected_union}
    If $G_1, G_2$ are $\hbar$-perfect, and $G$ is the fully connected disjoint union of $G_1$ and $G_2$, then $G$ is also $\hbar$-perfect.
\end{prop}
This property holds due to the corresponding properties of the weighted beta number and the independence number.
More explicitly, it can be proven that $\beta(G,w)$ is the maximum of $\beta(G_1,w_1)$ and $\beta(G_2,w_2)$, and $\alpha(G,w)$ is the maximum of $\alpha(G_1,w_1)$ and $\alpha(G_2,w_2)$, where $w_1$ and $w_2$ are the corresponding parts of $w$ for $G_1$ and $G_2$.
Hence, the $\hbar$-perfectness of $G_1$ and $G_2$ leads to the $\hbar$-perfectness of $G$.
\begin{prop}[Fully disconnected union]
\label{thm:disconnected_union}
    If $G_1, G_2$ are $\hbar$-perfect and $G$ is the fully disconnected disjoint union of $G_1$ and $G_2$, then $G$ is also $\hbar$-perfect.
\end{prop}
This property holds for a similar reason that $\beta(G,w)$ is the sum of $\beta(G_1,w_1)$ and $\beta(G_2,w_2)$, and $\alpha(G,w)$ is the sum of $\alpha(G_1,w_1)$ and $\alpha(G_2,w_2)$ in this case.
\begin{prop}[Induced subgraphs]
\label{thm:subgraph}
    For an $\hbar$-perfect graph $G$, if $G'$ is an induced subgraph of $G$ with the vertex set $V'$, then $G'$ is also $\hbar$-perfect.
\end{prop}
The underlying intuition is that $\stab(G')$ and $\be(G')$ can be viewed as the 
intersection of $\stab(G)$ and $\be(G)$ 
with the lower-dimensional space specified by the vertex set $V'$. Consequently, $\stab(G)=\be(G)$ implies that $\stab(G')=\be(G')$.

There is more than one way to take the product of two graphs. As it turns out, the lexicographic product will preserve $\hbar$-perfectness, i.e. we have

\begin{prop}[Lexicographic product] \label{tmh:product}
  For two given $\hbar$-perfect graphs $G_1$ and $G_2$, denote by $G$ the lexicographic product~\cite{sabidussi1961} of $G_1$ and $G_2$, then $G$ is also $\hbar$-perfect.
\end{prop}
This property is related to the weighted version of the fact that $\beta(G)$ is the product of $\beta(G_1)$ and $\beta(G_2)$, $\alpha(G)$ is the product of $\alpha(G_1)$ and $\alpha(G_2)$.
 \begin{figure}[t]
     \centering
     (a)\includegraphics[width=0.3\linewidth]{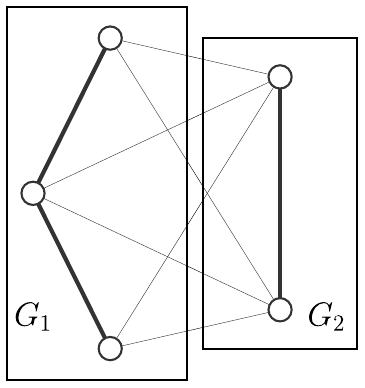} \hspace{3em} (b) \includegraphics[width=0.3\linewidth]{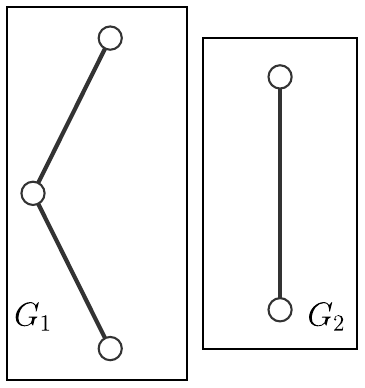}\hspace{1em}
     \\ (c) \includegraphics[width=0.5\linewidth]{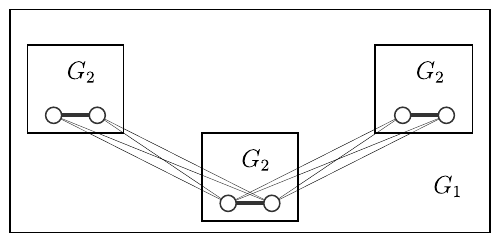}
     \caption{(a) The fully connected  and (b) the fully unconnected unions of  $\hbar$-perfect graphs are again $\hbar$-perfect. The same accounts for (c) the lexicographic product. }
     \label{fig:prod}
 \end{figure}

Instead of manipulating whole graphs, we can also consider single-vertex operations. 
For a given graph $G$ and one of its vertex $v$, a new vertex $v'$ is said to be a \textit{copy} of $v$ if $v'$ shares the same neighbors as $v$ and $v'$ is not a neighbor of $v$. If $v'$ is also a neighbor of $v$, $v'$ is said to be a \textit{split} of $v$.
In both copying and splitting operations, we use two vertices to replace one.
In the case of copying, the weights of those two vertices can accumulate for the beta number and the independence number.
As for the case of splitting, the corresponding weights compete, and the maximal one takes effect. The beta number and the independence number exhibit identical behavior under those two operations.
\begin{figure}[b]
   \centering
(a)\includegraphics[width=0.08\textwidth]{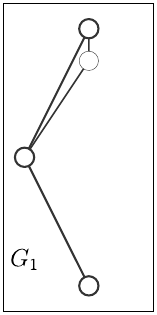}\hspace{3em}(b) \includegraphics[width=0.08\textwidth]{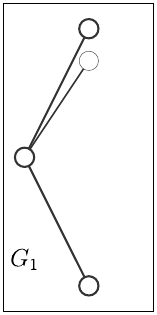}
    \caption{(a) splitting and (b) copying of vertices of an $\hbar$-perfect graph creates a new graph that is still $\hbar$-perfect.}
   \label{fig:scp}
\end{figure}
 
\begin{prop}[Copying of vertices]
\label{thm:copy}
    For an $\hbar$-perfect graph $G$, if the graph $G'$ is obtained from $G$ by copying one vertex, then $G'$ is also $\hbar$-perfect.
\end{prop}

\begin{prop}[Splitting of vertices]
\label{thm:splitting}
    For an $\hbar$-perfect graph $G$, if the graph $G'$ is obtained from $G$ by splitting one vertex, then $G'$ is also $\hbar$-perfect.
\end{prop}

As mentioned in the Introduction, a graph is simply said to be perfect if its chromatic number and its clique number also coincide for all induced subgraphs. 
A further common concept of `perfectness' is related to the facet structure of the stable set polytope. In detail, a graph is called $h$-perfect~\cite{chvatal1975certain,fonlupt1982transformations} if all the facets of $\stab(G)$ can be grouped into the following three classes:
    \begin{enumerate}
        \item $x_i \ge 0$, for $i$ in the vertex set of $G$;
        \item $\sum_{i\in K} x_i \le 1$, for $K$ a clique in $G$;
        \item $\sum_{i\in C} x_i \le a$, for $C$ an odd-cycle in $G$ with $2a+1$ vertices with $a$ being an integer.
     \end{enumerate}
If and only if the third class vanishes also, the graph is perfect~\cite{chvatal1975certain}.
Since the beta number equals the independence number for all these three classes, e.g., both are $1$ for the second class, we come to our last property in the following theorem. 
\begin{theorem}[Perfect and $h$-perfect graphs]
\label{thm:perfect}
    Any perfect graph and any $h$-perfect graph is $\hbar$-perfect. That is,
    \begin{align}
      G \text{ perfect } \Rightarrow \ G \text{ $h$-perfect } \Rightarrow G\ \text{$\hbar$-perfect}.
    \end{align}
\end{theorem}
The smallest three $h$-imperfect graphs are shown in Fig.~\ref{fig:himperfect}, all of which are $\hbar$-perfect as guaranteed by the aforementioned properties. The one in (a) is $\hbar$-perfect because it can be generated from the cycle with $5$ vertices by splitting one vertex.
The one in (b) is $\hbar$-perfect because it is an induced subgraph of another graph $G_{15}$ whose $\hbar$-perfectness is proven in Appendix~\ref{ssec:hper2}.
The one in (c) is $\hbar$-perfect since it is the fully connected union of one point and the cycle with $5$ vertices.
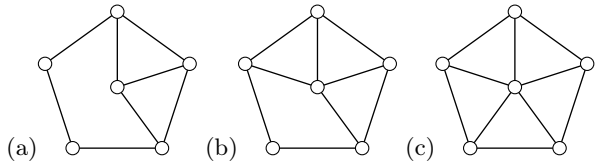
\begin{figure}[htpb]
	\centering
(a)\begin{tikzpicture}[scale=1]
  \foreach \i in {1,...,5} {
    \pgfmathsetmacro{\angle}{72*\i-54} 
    \coordinate (p\i) at (\angle:1cm);
  }
  \coordinate (p6) at (0,0);
  \draw[line width=0.5] (p1) -- (p2) -- (p3) -- (p4) -- (p5) -- (p1) -- (p6) -- (p2);
  \draw[line width=0.5] (p6) -- (p5);
  \foreach \i in {1,...,6} {
    \filldraw[fill=white, draw=black] (p\i) circle (2.5pt);
  }
\end{tikzpicture}
(b)\begin{tikzpicture}[scale=1]
  \foreach \i in {1,...,5} {\pgfmathsetmacro{\angle}{72*\i-54}
  \coordinate (p\i) at (\angle:1cm);}
  \coordinate (p6) at (0,0); 
  \draw[line width=0.5] (p1) -- (p2) -- (p3) -- (p4) -- (p5) -- (p1) -- (p6) -- (p2);
  \draw[line width=0.5] (p3) -- (p6) -- (p5);
  \foreach \i in {1,...,6} {
    \filldraw[fill=white, draw=black] (p\i) circle (2.5pt);
  }
\end{tikzpicture}
(c)\begin{tikzpicture}[scale=1]
  \foreach \i in {1,...,5} {
    \pgfmathsetmacro{\angle}{72*\i-54}
    \coordinate (p\i) at (\angle:1cm);
  }
  \coordinate (p6) at (0,0); 
  \draw[line width=0.5] (p1) -- (p2) -- (p3) -- (p4) -- (p5) -- (p1) -- (p6) -- (p2);
  \draw[line width=0.5] (p3) -- (p6) -- (p5);
  \draw[line width=0.5] (p4) -- (p6);
  \foreach \i in {1,...,6} {
    \filldraw[fill=white, draw=black] (p\i) circle (2.5pt);
  }
\end{tikzpicture}
	\caption{The smallest three $h$-imperfect  graphs.}
	\label{fig:himperfect}
\end{figure}

\section{\texorpdfstring{Checking for $\mathbf{\hbar}$-perfectness and techniques for computing  $\mathbf{\beta(G,w)}$}{}}
\label{sec:numerical}
The properties presented in the last section can be used as basic operations for determining $\hbar$-perfectness attempting to  decompose  a large graph into smaller, easier ones. In fortunate cases such a procedure will end with a decomposition of a large graph into a basic set of $\hbar$-perfect graphs. 
Examples of those are 
cliques and cycles. Nevertheless, those two classes are not big  enough for deciding 
$\hbar$-perfectness in practice. In order to enlarge this class, we will now introduce some tools for deciding $\hbar$-perfectness on graphs of moderate sizes. 

At first glance,  $\hbar$-perfectness according to definition~\ref{basicdef} requires us to compare $\alpha(G,w)$ and $\beta(G,w)$ for all non-negative weight vectors $w$. Without a closed form of these expressions at hand, this task is, however, impossible in practice. 
Nevertheless we are fortunate; it suffices to only check  equality of those two graph parameters on a finite set of  weights. 
Instead of comparing the functions $\alpha(G,w)$ and $\beta(G,w)$, we can equivalently try to decide if the sets $\be(G)$ and $ \stab(G)$ are equal to each other. 
Recall that $\stab(G)$ is a polytope and also the inclusion relation $\stab(G)\subseteq \be(G)$. Since both sets are convex, it is easy to see that only the $w$'s corresponding to the facets of $\stab(G)$ have to be checked. In detail, we have that a polytope in a finite-dimensional space is fully described by  its facets. If there is any point outside of $\stab(G)$ that is in $\be(G)$,  it can be detected by at least one of the inequalities defined by the facets. This can be exactly done by comparing  $\alpha(G,w)$ and $\beta(G,w)$ corresponding to the face-normal of this facet.
Note that this strategy is also used to prove Property~\ref{thm:perfect}.

The central problem we are left with is, therefore, the pointwise computation of  $\beta(G,w)$ as a function of the weights. 
In the following, we will present  three different algorithmic approaches for this. A practical solution for our decision problem will, in the end, consist of a clever combination of all those methods. 

The first method is an outer approximation in the form of an SDP hierarchy. Even though this hierarchy has no guarantee to stop at a finite level, it usually  performs quite well in practice and is sometimes even capable of producing tight bounds at very low levels. 
Our second method is a de Finetti type approximation used to lift the non-linearity in  $\beta(G,w)$  to a sequence of linear problems under symmetry on an extended space. We combine this with some symmetry reduction techniques and end up with an approximation routine that only relies on the computation of the largest eigenvalue of a sparse matrix. Here we can give clear performance guarantees that can be formulated as achieving $1/N$ accuracy in $poly(N)$ runtime.

These methods can be complemented by simple see-saw iterations that start from an initial state $\psi_0$ and successively improve the performance in order to find the optimum in the definition of $\beta(G,w)$. This see-saw method produces inner bounds on $\beta(G,w)$, thus in contrast to the other two methods that produce outer bounds. In practice, the see-saw method often quickly produce good approximations. In some instances, it might, however, get stuck in local minima. This behavior strongly depends on the initial state $\psi_0$. We can, however, combine it nicely with the other methods by producing warm starts, i.e. educated guesses on $\psi_0$ from outer approximations.

\subsection{\texorpdfstring{Computation of  $\mathbf{\beta(G,w)}$  via SDP hierarchies}{}}
\label{sec4a}

For a given graph $G$ and a given non-negative weight vector $w$, there is a complete hierarchy of SDP relaxations for bounding $\beta(G,w)$ from above, which is based on state polynomial optimization~\cite{moran2024Uncertainty,klep2024state}. 
The main idea is as follows.
We assign the $i$-th vertex the letter $x_i$ with the constraints that $x_i^2 = 1$ and $x_i x_j = - x_j x_i$ ($x_ix_j=x_jx_i$) whenever the vertices $i$ and $j$ are adjacent (non-adjacent) in $G$.
Then we consider the words $u=x_{i_1}\cdots x_{i_t}$ together with their pseudo-expectations $\langle u\rangle$ which are assumed to be real numbers.
The involution of $u$ is defined as $u^* = x_{i_t}\cdots x_{i_1}$.
A state monomial is of the form $u = u_0\langle u_1\rangle \cdots \langle u_s\rangle$, where $u_0, \ldots, u_s$ are words. The degree of $u$ is the sum of lengths of the $u_i$'s. For example, the degree of $x_1\langle x_1\rangle$ is $2$.
Now for a positive integer $r$, we define the $r$-th moment matrix indexed by state monomials up to degree $r$ through
$[M_r]_{u,v} = \langle u^*v\rangle$.
Then the $r$-th level of the hierarchy is given by the following SDP:
\begin{equation}\label{eq:hierarchy}
\begin{aligned}
	\lambda_r(G,w) \coloneqq & \sup \sum_{i=1}^n w_i\langle x_i\rangle^2\\
\rm{s.t.}\ & M_r\succeq0,\\
&[M_r]_{1,1}=1,\\
&[M_r]_{u,v}=[M_r]_{a,b}, \text{ if }\langle u^*v\rangle=\langle a^*b\rangle,
\end{aligned}
\end{equation}
where the conditions in the last line root in the normalization of $x_i$ and their commutation and anticommutation relations encoded in $G$. Each $\lambda_r(G,w)$ provides an upper bound on $\beta(G,w)$ and it was proved in \cite{klep2024state} that $\lim_{r\to\infty}\lambda_r(G,w)=\beta(G,w)$. Unfortunately, this complete hierarchy has a very high computational complexity and does not scale well. A simplified hierarchy with finite levels was presented in \cite{moran2024Uncertainty} by selecting a subset of state monomials to index moment matrices, which, however, does not result in exact values of $\beta(G,w)$ even for some graphs with $7$ vertices.
To better balance computational costs and the tightness of SDP relaxations, we propose a new strategy to select state monomials, leading to an enhanced hierarchy with infinite levels for bounding $\beta(G,w)$ from above. See Appendix~\ref{ssec:spo} for more details.

\subsection{\texorpdfstring{Computation of $\mathbf{\beta(G,w)}$ via a mean-field approximation}{}}

For a given graph $G$, denote by $\mathcal{S}$ a representation of $G$ with the shortest length $\ell$ of Pauli strings. Then the dimension of the elements $S_i$ is $d = 2^\ell$.
Due to the convexity of $\sum_i w_i  \langle S_i\rangle^2_{\rho}$ in $\rho$, the maximum of this quantity can always be achieved with $\rho$ being a pure state.
Hence, we can assume $\rho$ to be from the set of pure states in the following consideration.
With $H_S = \sum_i w_i S_i\otimes S_i$, we have
\begin{equation}\label{eq:derivation_deFinetti_argument}
\begin{aligned}
\beta(G,w)	= &\max_{\rho}\sum\nolimits_i w_i \langle S_i\rangle^2_{\rho} \\
	=&\max_{\rho}\sum\nolimits_i w_i \langle S_i\otimes S_i\rangle_{\rho\otimes\rho} \\
	=&\max_{\rho}\big\langle \sum\nolimits_i w_i S_i\otimes S_i\big\rangle_{\rho\otimes\rho} \\
	=&\max_{\rho} \tr[(\rho\otimes\rho) H_S]\\
\end{aligned}
\end{equation}

The optimization problem in \eqref{eq:derivation_deFinetti_argument} has a solution in the form of an SDP hierarchy. To this end, we define the following subspace
\begin{align}
  \operatorname{Sym}^m(\mathcal{H}) 
    \coloneqq \operatorname{span} \left\{ \psi^{\otimes m} \ \vert \ \psi \in \mathcal{H} \right\} \subset \mathcal{H}^{\otimes m}.
\end{align}
The operators acting on  $\operatorname{Sym}^m(\mathcal{H})$ are called Bose-symmetric, and their set is SDP representable \cite{Harrow2013_ChurchSymmetricSubspace}. In particular, each pure product state has an $m$-extension in $\operatorname{Sym}^m(\mathcal{H})$, so that 
\begin{equation}\begin{split}
\label{eq:definition_optimization_kextension}
  c^{(m)} \coloneqq \max &\tr[ \sigma_1^{m+2}  (H_S\otimes \id^{(m)}) ] \\
    &\text{s.t. }  \sigma_1^{m+2}\geq 0,\, \tr\sigma_1^{m+2}=1,\\
    &\phantom{==:}
     \sigma_1^{m+2} \in \operatorname{Sym}^{m+2}(\mathcal{H})
\end{split}\end{equation}
is an upper SDP approximation of \eqref{eq:derivation_deFinetti_argument}. Moreover, this can be converted into an eigenvalue problem for each level $m$ of the hierarchy. 

\begin{theorem}
\label{thm:deFinetti_approximation}
For integers $m$, the optimization problem in \eqref{eq:definition_optimization_kextension} is a
hierarchy of upper approximations converging to $\beta(G,w)$ from above, with the error bound 
\begin{equation}
  0 \leq c^{(m)} - \beta(G,w) \leq O(nd/m),
\end{equation}
where $n$ is the number of vertices in $G$. Moreover, each level of the hierarchy can be initialized in $\operatorname{poly}(m)$ time.
\end{theorem}
\begin{proof}
Due to the fact that we optimize in the Bose symmetric subspace, we can use a finite de Finetti theorem customized for this.
From the finite de Finetti theorem~\cite[Thm.~II.8]{Christandl2007}, it is guaranteed that there exists a separable bipartite state $\tau$ such that the trace distance $T(\tau,\tr_m \sigma) \le 4d/(m+2)$ .
This implies that 
\begin{equation}
\begin{aligned}\label{eq:mfdiff}
  |\lambda_{\max}(H_S^{(m)}) -& \beta(G,w)| \\
	&\le \max_{\tau\in {\rm SEP}} |\tr((\tr_m \sigma - \tau)H_S)| \\
	&\le \max_{\tau\in {\rm SEP}} \lambda_{\max}(H_S) T(\tau, \tr_m \sigma) \\
	&\le 4 \|w\|_1 d/(m+2),
\end{aligned}
\end{equation}
where $\|w\|_1 = \sum_i w_1 = O(n)$, the last inequality is due to the fact that the elements in $\{S_i\otimes S_i\}$ commute with each other and each of them has eigenvalues $\pm 1$. Moreover, as we essentially aim to construct the first block of irreducible representations of the symmetric group $S_{m+2}$, we can use the techniques developed in \cite{Litjens2016} to calculate the matrix entries of the reduced SDP in the basis of the Bose symmetric subspace in $\operatorname{poly}(m)$ time. 
\end{proof}
Hence, this constitutes a convergent hierarchy, where the calculation for each level is an SDP.
For the $m$-th level, the dimension of the Bose symmetric subspace, i.e., the rank of $P_{\rm Sym}^{(m+2)}$ is $D = \binom{m+d+1}{m+2} = \binom{m+d+1}{d-1}$. Thus, $H_S^{(m)}$ represented in the Bose symmetric subspace is of dimension $D$, which is linear in the size $n$ of the graph,  polynomial in the level $m$, and polynomial in $d=2^\ell$ or exponential in length $\ell$ of the representation.
Then the value of $\lambda_{\max}(H_S^{(m)})$ can be obtained in $O(D^3)$ by full diagonalization.
To have the difference between $\lambda_{\max}(H_S^{(m)})$ and $\beta(G,w)$ smaller than $\epsilon > 0$, it suffices to set $m = O(nd/\epsilon)$. Consequently, $D = O((nd/\epsilon)^d)$ and the runtime of the whole algorithm is $O((nd/\epsilon)^{3d})$.
Since this hierarchy is polynomial in $n$ and superexponential in $d$, it is suitable for large graphs represented by short Pauli strings of length $\ell$.
Nevertheless, the number of nontrivial Pauli strings with length $\ell$ is $4^\ell - 1$, which is the maximal size of the corresponding frustration graph.
In the situation where we have small graphs represented by long Pauli strings, we switch to the hierarchy based on state polynomial optimization.

\subsection{\texorpdfstring{Computation of $\mathbf{\beta(G,w)}$ via a see-saw algorithm}{}}

There is also a comparatively lightweight see-saw method~\cite{xu2023bounding} for computing lower bounds on $\beta(G,w)$. 
For a given graph $G$ and one representation $\{S_i\}$ with the shortest length,
\begin{equation}
\begin{aligned}
    \beta(G,w) & = \max_{\rho} \sum\nolimits_i w_i \langle S_i\rangle^2_{\rho}\\
    &= \max_{\rho,\,\|b\|=1} \left(\sum\nolimits_i b_i\sqrt{w_i}\langle S_i\rangle_{\rho} \right)^2.
\end{aligned}
\end{equation}
Then for a given norm vector $b$, the optimization reduces to the calculation of the maximal singular value of $\sum_i b_i\sqrt{w_i} S_i$. For a given state $\rho$, the optimal norm vector $b$ is just the normalization of $(\langle S_1\rangle_\rho, \ldots, \langle S_n\rangle_\rho)$. This leads to a see-saw method by randomly choosing an initial state and iteratively optimizing over the norm vector $b$ and the state $\rho$. Since each step in the iteration results in a value no less than the one in the previous step, the sequence of those values converges to a limit that is a lower bound of $\beta(G,w)$.
Yet, the tightness of the lower bounds given by the see-saw method highly depends on the quality of the initially chosen quantum state.
We could then combine the see-saw method for lower bounds with the two hierarchies for upper bounds in the following way:
\begin{enumerate}
	\item Solve some level of one of the hierarchies to obtain an upper bound on $\beta(G,w)$;
	\item From the solution of the corresponding SDP, extract a legal quantum state close to the optimal state;
	\item Run the see-saw method with this approximately optimal state as the initial point, yielding a lower bound on $\beta(G,w)$.
\end{enumerate}
If the upper bound matches the lower bound (up to some numerical precision), then we obtain the exact value of $\beta(G,w)$.
For step 2, the optimal solution of the SDP at the $m$-th level of the hierarchy via mean-field approximation is a symmetric $(m+2)$-partite state $\sigma$. Then we can take the eigenstate of $\tr_{m+1}\sigma$ corresponding to the maximal eigenvalue as the initial state for the see-saw method. More details on the retrieval of a state for step 2 with the hierarchy based on state polynomial optimization are provided in Appendix~\ref{ssec:seesaw_state}.

\section{\texorpdfstring{Frequency of $\mathbf{\hbar}$-perfectness}{}}
\label{sec:fre}
Having introduced the properties and numerical tools, an immediate question to ask is: How common are $\hbar$-perfect graphs? 

As it turns out, small graphs can be tested for $\hbar$-perfectness and $\hbar$-imperfectness quite efficiently. 
We tested all graphs on up to 9 vertices (up to isomorphism) for their membership in different classes. 
To determine whether an $h$-imperfect graph is $\hbar$-perfect, we have first made use of the aforementioned properties. For the undetermined ones, we have employed the numerical methods in the previous section for the estimation. 
Then the $\hbar$-perfectness and $\hbar$-imperfectness of all graphs with $8$ vertices are determined up to a precision of $1\times 10^{-5}$. Among all connected graphs with $9$ vertices, only $78$ of them are left undetermined. 
The result of this is summarized in Table~\ref{tab:kinds_ghs}. More details are provided in Appendix~\ref{ssec:typical}.
\begin{table}[h!]
    \centering
        \begin{tabular}{r|l|l|l|l|l|l|l}
    \hline\hline
        $n$ & $3$ & $4$ & $5$ & $6$ & $7$ & $8$ & $9$ \\ \hline
        connected & $2$ & $6$ & $21$ & $112$ & $\bf{853}$ & $11117$ & $261080$ \\ \hline
        $\hbar$-perfect & $2$ & $6$ & $21$ & $\bf{112}$ & $852$ & $11099$ & $259583^*$  \\ \hline
        $h$-perfect & $2$ & $6$ & $\bf{21}$ & $109$ & $780$ & $8689$ & $146375$ \\ \hline
        perfect~\cite{oeis} & $2$ & $\bf{6}$ & $20$ & $105$ & $724$ & $7805$ & $126777$  \\ \hline
	claw-free~\cite{oeis} & $\bf{2}$ & $5$ & $14$ & $50$ & $191$ & $881$ & $4494$ \\ \hline\hline
    \end{tabular}
    \caption{The numbers of different classes of non-isomorphic graphs, where $n$ is the number of vertices. In the case of $9$ vertices, the number of $\hbar$-perfect graphs is in the range $[259583, 259661]$ since there are $78$ undetermined graphs. Claw-free~\cite{faudree1997} graphs are used in the study of solvable models~\cite{chapman2023unified}. In comparison, the ground state energy of models corresponding to $\hbar$-perfect graphs can be effectively estimates as discussed in Section~\ref{sec:app}.
    }
    \label{tab:kinds_ghs}
\end{table}

The numbers presented in Table~\ref{tab:kinds_ghs} suggest that $\hbar$-perfectness is indeed a typical feature for small graphs. 
However, the following observation suggests that the percentage of $\hbar$-perfect graphs decreases super-exponentially when the number of vertices is large enough, by noticing that the number of all graphs with $n$ vertices is $2^{n(n-1)/2}$. 
\begin{theorem}
\label{thm:haufigkeit}
	The number of $\hbar$-perfect graphs on $n\gg 1$ vertices is at most $2^{c n(n-1)/2}$ for some constant $c<1$.
\end{theorem}
The exact proof is given in Appendix~\ref{ssec:asym}. The main idea is that there are forbidden subgraphs like $\bar{C}_7$ for $\hbar$-perfect graphs, and the chance of avoiding all such subgraphs decays exponentially in the number of edges, $N=\frac{n(n-1)}{2}$.
This is a common feature in extremal graph theory problems. 
The result extends to any graph property defined by the absence of a set of forbidden induced subgraphs. For such properties, the number of graphs satisfying the property is at most $2^{\tilde{c} N}$ for some $\tilde{c}<1$.
It is plausible that the optimal constant $\tilde{c}$ in the exponent is a characteristic of the specific graph property. For the related concepts of perfectness, $h$-perfectness, and $\hbar$-perfectness, we can speculate that the constants form a strictly increasing sequence:
\begin{equation}
\frac{1}{2} \le c_{\text{perfect}} < c_{h\text{-perfect}} < c_{\hbar\text{-perfect}} < 1,
\end{equation}
where the first inequality is an application of the fact that the number of graphs on $n$ vertices without any graph from a given set ${\cal L}$ as a subgraph (not necessarily induced) is $2^{(1-1/p)N + o(N)}$, where $p$ is the smallest chromatic number of graphs in ${\cal L}$~\cite{Erds1986TheAN,Balogh2004TheNO}.
The value of $p$ is $3$ for ${\cal L}$ the set of imperfect graphs, since all graphs with the chromatic number $2$ are bipartite graphs and consequently perfect ones.
The value of $p$ is also $3$ for ${\cal L}$ the set of $h$-imperfect graphs, since the first two $h$-imperfect graphs in Fig.~\ref{fig:himperfect} have a chromatic number of $3$. Hence, the lower bound of $c_{h-\text{perfect}}$ is also $1/2$, without any improvement over the one for $c_{\text{perfect}}$.
However, the result on all graphs with no more than $9$ vertices shows that all the $\hbar$-imperfect graphs here are with chromatic number $4$. If this is indeed the case for any $\hbar$-perfect graphs, then it leads to the lower bound of $3/4 \le c_{\hbar-\text{perfect}}$.
This speculation suggests that each successive property is ``less common'' in the random graph model.
Once the speculation holds, it implies that there are super-exponentially more $\hbar$-perfect graphs than the other two kinds of perfect graphs.

Remarkably, there are at most $n!$ possible ways to label a graph on $n$ vertices such that the resulting graphs are isomorphic to each other. Since $\lim_{n\to \infty} n! 2^{-\epsilon n^2} = 0$ for any $\epsilon>0$, the number of non-isomorphic  graphs on $n\gg 1$ vertices in the considered class is also at the scale of $2^{cN+o(N)}$ for $c$ being either $c_{\rm perfect}, c_{h\text{-perfect}}$ or $c_{\hbar\text{-perfect}}$. Since in particular the number of all isomorphism classes of graphs on $n$ vertices is at the scale $2^{N+o(N)}$, the percentage of each class is at the scale $2^{-(1-c)N+o(N)}$, which decreases super-exponentially when $n$ is large enough.

In the following, we apply the framework of $\be(G)$ and $\hbar$-perfect graphs in the detection and estimation of quantum entanglement, shadow tomography, and outline further fruitful implications in some basic applications.

\section{Entanglement detection and estimation}
\label{sec:entanglement}
One of the central tasks in the foundations of quantum information is how to detect quantum entanglement efficiently, where the main approach is entanglement witness~\cite{horodecki1996necessary,terhal2000bell,huber2013structure}, leading to entanglement estimation also~\cite{brandao2005quantifying,guhne2007estimating,sun2024bounding}.
One dimension is to detect more entangled states with the same amount of measurement resources. Another dimension is to use less measurement resources for the same detection ability.
In this section, we contribute to both of those two dimensions by employing $\hbar$-perfect graphs in the detection of high-dimensional entanglement, and turning each graph state into a tool for multipartite entanglement structures. Besides, since each observable is in the form of Pauli strings, the measurement can be implemented as a composition of the ones on single qubit.

\subsection{High-dimensional entanglement}
The fact that $\sum_i w_i \langle S_i\rangle^2 \le \beta(G,w)$ implies a tight nonlinear witness for bipartite entanglement, i.e., for any separable state $\rho$,
\begin{equation}\label{eq:nlwitness}
	\sum\nolimits_i w_i |\langle S_i\otimes S'_i\rangle_{\rho}| \le  \beta(G,w),
\end{equation}
where ${\cal S} = \{S_i\}$ and ${\cal S}' = \{S'_i\}$ are two realizations of $G$ as a frustration graph.
This inequality holds due to the convexity of the absolute value, the Cauchy-Schwarz inequality, and the fact that separable states are linear combinations of product states. 
We note that a similar but less tight inequality has been obtained~\cite{de2023uncertainty}, where the right hand is the weighted Lov\'asz number~\cite{lovasz1979shannon} $\vartheta(G,w) \ge \beta(G,w)$ and the left hand does not take an absolute value. As we can see in the last example in this section, those small differences make a distinctive improvement on the detection ability of entanglement.

By taking another perspective, define
\begin{equation}
\begin{aligned}
	&p({\cal S}, {\cal S}', \rho) \coloneqq (|\langle S_1\otimes S'_1\rangle_{\rho}|, \ldots, |\langle S_n\otimes S'_n\rangle_{\rho}|),\\
	&{\cal W}({\cal S}, {\cal S}') \coloneqq \conv(\downarrow\hspace{-0.3em}\{p({\cal S}, {\cal S}', \rho) \,|\, \rho \text{ is separable}\}).
\end{aligned}
\end{equation}
As it turns out,
\begin{equation}\label{eq:entb}
{\cal W}({\cal S}, {\cal S}') = \be(G).
\end{equation}
On the one hand,  Eq.~\eqref{eq:nlwitness} implies that 
${\cal W}({\cal S}, {\cal S}') \subseteq \be(G)$,
as $\be(G)$ is the minimal convex set containing all the points $p$ with non-negative elements satisfying the inequality $\sum_i w_i p_i \le \beta(G,w)$ by definition.
On the other hand, denote by $\bar{\cal S}$ a basic representation of $G$, then there are local unitary transformations to convert any observable pair $S_i\otimes S'_i$ in $({\cal S}, {\cal S}')$ to $(\bar{S}_i\otimes D_i)\otimes(\bar{S}_i\otimes D'_i)$, where $\bar{S}_i \in {\cal S}$, $D_i$ and $D'_i$ are diagonal matrices with diagonal elements as $\pm 1$.
Consequently, ${\cal Q}(\bar{\cal S}) \subseteq {\cal W}({\cal S}, {\cal S}')$ since all the points in ${\cal Q}(\bar{\cal S})$ can be realized by some separable state $\rho$ with the tuple $({\cal S}, {\cal S}')$. Proposition~\ref{ob:beta} directly implies that $\be(G) \subseteq {\cal W}({\cal S}, {\cal S}')$.

Hence, verifying that the point $p({\cal S}, {\cal S}', \rho)$ is outside  $\be(G)$ implies the entanglement of $\rho$.
Since the exact characterization of $\be(G)$ is difficult especially when the local dimension of $\rho$ is high, our numerical methods provide a hierarchy of outer approximations of $\be(G)$ as proven in Appendix~\ref{ssec:hierprop} and lead to effective criteria of entanglement based on SDP. To be more precise, we use the point $p({\cal S}, {\cal S}', \rho)$ of the given state $\rho$ as the weight vector $w$ in Eq.~\eqref{eq:hierarchy}. Once the optimal value is strictly smaller than the square of the norm of $p({\cal S}, {\cal S}',\rho)$, it means that this point is outside the set of feasible points allowed by the SDP in Eq.~\eqref{eq:hierarchy}. Consequently, this point is outside $\be(G)$ and the state $\rho$ should be entangled.
\begin{figure}
    \centering
    \includegraphics[width=0.875\linewidth]{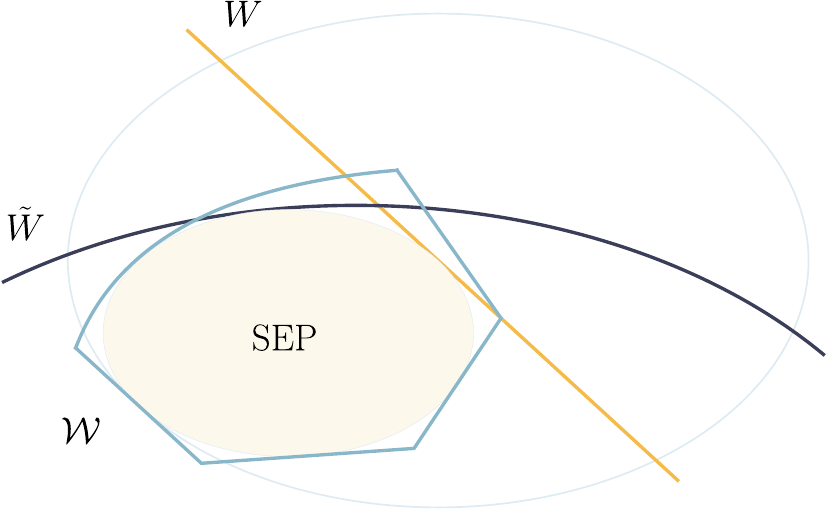}
    \caption{Entanglement detection via beta body versus linear and nonlinear witnesses. A linear witness $W$ defines a hyperplane in the space of all states, separating some entangled states from the
set ${\rm SEP}$ of separable states. Similarly, a nonlinear witness $\tilde{W}$ defines a nonlinear hypersurface. The beta body defines a closed convex set ${\cal W}$, whose surface is in general a combination of hypersurfaces and hyperplanes. For the underlying graph to be $\hbar$-perfect, there are only hyperplanes.}
    \label{fig:placeholder}
\end{figure}

The procedure can be simplified for $G$ an $\hbar$-perfect graph, since then $\be(G)=\stab(G)$ is a polytope. 
In particular, when $G$ is $h$-perfect or even perfect, $\stab(G)$ is well-characterized, see Section~\ref{sec:props}.
The detection of entanglement reduces to the separation of a point from a polytope, for which there are more economic algorithms, such as the Gilbert algorithm. In addition, we can use the symmetry of the graph $G$ to further reduce the complexity of the algorithm as in Ref.~\cite{xu2023graph}.
As a simple example, we take ${\cal S} = {\cal S}' = \{X, Y, Z\}$. The corresponding frustration graph $G$ is the triangle graph which is perfect and hence $\hbar$-perfect. Then $\be(G) = \stab(G) = \{(x,y,z)|x+y+z\le 1\} \cap \mathbb{R}^3_+$.
If $\rho$ is the density matrix of any of the four Bell states $|\Psi^\pm\rangle = (|00\rangle \pm |11\rangle)/\sqrt{2}$ and $|\Phi^\pm\rangle = (|01\rangle \pm |10\rangle)/\sqrt{2}$, then $p({\cal S}, {\cal S}',\rho) = (1,1,1)$ which is outside $\be(G)$. Furthermore, for $\rho$ to be any Bell diagonal state, i.e., the mixture of the four Bell states, it is entangled if and only if the corresponding point $p({\cal S}, {\cal S}',\rho)$ is out of $\be(G)$. More details are provided in Appendix~\ref{ssec:entanglement}. 
In comparison, any linear witness can detect at most one of them, since the maximal mixture of any two is separable.

The advantage of this new method over the linear witness becomes more evident as the local dimension of the state increases.
When the local dimension is $d=2^k$, the $4^k$ maximally entangled Bell states can be detected simultaneously by taking ${\cal S}={\cal S}'$ as all $k$-partite Pauli strings except the identity. For a dimension $d$ that is not a power of $2$, one strategy is to embed the system into a higher-dimensional space; another is to use lower-dimensional subspaces as a cover of the original. We take $d=3$ as an example. Denote by $X^{(i)}, Y^{(i)}, Z^{(i)}$ the operators $X, Y, Z$ in the two-dimensional subspace without the $i$-th dimension. Then as detailed in Appendix~\ref{ssec:entanglement}, it holds for any separable state that
\begin{equation}
\label{eq:dim3}
    \sum_{i,j=1}^3 \big[ |\langle X^{(i)} X^{(j)}\rangle| + |\langle Y^{(i)} Y^{(j)}\rangle| + |\langle Z^{(i)} Z^{(j)}\rangle|\big] \le 4.
\end{equation}
Direct calculation shows that this inequality is violated by the following state with any $p\neq 2/5$, 
\begin{equation}
    \rho(p) = (1-p) |\Psi_3\rangle\langle \Psi_3| + p |\Phi'\rangle\langle \Phi'|,
\end{equation}
where $|\Psi_3\rangle = (\sum_{i=1}^3 |ii\rangle)/\sqrt{3}$ and $|\Phi'\rangle = (|12\rangle + |21\rangle)/\sqrt{2}$. Furthermore, the criterion in Eq.~\eqref{eq:dim3} detects unfaithful entanglement~\cite{Weilenmann2019EntanglementDB} also.

Another advantage of this new method is to recognize critical measurements for the detection of entanglement, exemplified by the bipartite state $\rho(v)$ with local dimensions of $4$ from Ref.~\cite{de2023uncertainty}, where the entanglement cannot be detected by the PPT criterion. More explicitly,
\begin{equation}\label{eq:rhov}
    \rho(v) = v\sum_{i=1}^6 |\psi_{S_i}\rangle\langle \psi_{S_i}|/6 + (1-v)\id/16,
\end{equation}
where $v\in [0,1]$, $|\psi_i\rangle = \sum_{j=0}^3 |j\rangle \otimes (S_i |j\rangle)/2$ and $\{S_i\}$ are the $6$ Pauli strings $\id Y, XX, YZ, ZX, ZY$ and $ZZ$. 
It is found in Ref.~\cite{de2023uncertainty} that $\rho(v)$ is PPT entangled if and only if $v\in (3/5,1]$ by employing all the $15$ nontrivial bipartite Pauli strings as $4$-dimensional measurements per party. 

As discussed in Appendix~\ref{sec:proofs}, the frustration graph $G_{15}$ of those $15$ Pauli strings is $\hbar$-perfect, and all nontrivial facets of $\be(G_{15}) = \stab(G_{15})$ have the norm vector with either $5$ or $10$ elements being $1$ and the others being $0$. Then we only need to consider the corresponding criteria in Eq.~\eqref{eq:nlwitness}. As it turns out, the criterion in Eq.~\eqref{eq:nlwitness} with only the $5$ pairwise anticommutative measurements in ${\cal S}_5 =\{\id X, \id Z, XY, YY, ZY\}$ per party leads to the same conclusion about the entanglement of $\rho(v)$, which saves $66.6\%$ measurement resources in comparison with the original criterion in Ref.~\cite{de2023uncertainty}.
By testing all the other combinations of $6$ bipartite Pauli strings with the same frustration graph as $\{S_i\}$, there are $120$ such PPT entangled states whose entanglement can be detected by the criterion in Eq.~\eqref{eq:nlwitness} with measurements in ${\cal S}_5$. However, all the other $119$ states except the one in Eq.~\eqref{eq:rhov} cannot be detected by the original criterion in Ref.~\cite{de2023uncertainty}.
This example exactly illustrates how the new method can detect more entanglement with less measurement resources.
\subsection{Multipartite entanglement structure}
We continue to consider multipartite entanglement.
Graph states form an important class of states with high entanglement~\cite{hein2006entanglement}. For a given graph $G$, the graph state $|G\rangle$ is defined as the common eigenstate with eigenvalue $1$ of the commuting operators $\{X_i\otimes Z_{N(i)}\}_{i\in V(G)}$, where $V(G)$ is the vertex set of $G$, i.e., the number of parties, $N(i)$ is the set of vertices connecting to $i$ and $Z_{N(i)} = \otimes_{j\in N(i)} Z_j$. Those operators generate the stabilizer group of the graph state.
Graph states and their stabilizers have already been employed to construct linear witnesses in the detection of entanglement structures~\cite{jungnitsch2011entanglement}, especially in the way with minimal measurement resources~\cite{toth2005detecting,zhou2019detecting}.
GHZ states are equivalent to graph states with complete graph up to local unitaries.
However, any linear witness cannot detect the genuine multipartite entanglement of any two GHZ states, since their maximal combination is not genuinely multipartite entangled.
We extend our technique to the multipartite scenario with the ability to detect all the graph states forming a complete basis, leading to an adaptive strategy for the detection of different entanglement structures with few measurement resources.

To illustrate the main idea, we take the $8$ tripartite GHZ states $(|0ij\rangle\pm |1\bar{i}\bar{j}\rangle)/\sqrt{2}$ as an example, where $i,j \in \{0,1\}$, $\bar{i}=1-i$ and $\bar{j}=1-j$.
The stabilizer group excluding the identity contains (up to a minus sign)
\begin{equation}\label{eq:stabilizer}
	Z Z\id, Z\id Z, \id Z Z, XXX, YYX, YXY, XYY,
\end{equation}
denoted by $\{O_i\}_{i=1}^7$ in order.
For the bipartition $1|23$ of the three parities, the operators $O_i$ are divided into $S_i$ and $S'_i$. More explicitly,
\begin{align}\label{eq:splitted}
	S_i:\ & Z, Z, \id, X, Y, Y, X,\nonumber\\
	S'_i:\ & Z\id,\id Z,Z Z, XX, YX, XY, YY.
\end{align}
Notice that ${\cal S} =\{S_i\}$ and ${\cal S}' = \{S'_i\}$ share the same frustration graph denoted by $G$.
Then for any tripartite state $\rho$  separable for this bipartition, we have $p({\cal S}, {\cal S}', \rho) \in \be(G)$. Equivalently,
for any non-negative coefficient $w$,  we have
$\sum_i w_i |\langle O_i\rangle_{\rho}|	= \sum_i w_i |\langle S_i\otimes S'_i\rangle_{\rho}| \le \beta(G,w)$ with $\beta(G,w) = w_3+\max\{w_1+w_2,w_4+w_7,w_5+w_6\} \coloneqq t_1$.
Similarly for another two bipartitions, we obtain upper bounds for $\sum_i w_i |\langle O_i\rangle_{\rho}|$ in the form $w_1+\max\{w_2+w_3,w_4+w_5,w_6+w_7\} \eqqcolon t_2$ and $w_2+\max\{w_1+w_3,w_4+w_6,w_5+w_7\}\eqqcolon t_3$.
The convexity of the absolute value and the fact that any biseparable state is a mixture of the separable states for some bipartition implies that
\begin{equation}\label{eq:gmecriterion}
	\sum\nolimits_i w_i |\langle O_i\rangle_{\sigma}| \le \max\{t_1, t_2, t_3\}
\end{equation}
for any biseparable state $\sigma$.  
The reformulation of Eq.~\eqref{eq:gmecriterion} leads to a polytope defined by linear inequalities.
Thus, the detection of a point outside of this polytope, or equivalently, the violation of the inequality in Eq.~\eqref{eq:gmecriterion} indicates the genuine tripartite entanglement.
Setting all elements of $w$ to $1$, we have $\sum_i |\langle O_i\rangle_{\sigma}| \le 3$.
For $\rho$ equal to any of those $8$ GHZ states,  $\sum_i |\langle O_i\rangle_{\rho}| = 7$.
In fact, for any state $\rho$ whose fidelity with any of those $8$ GHZ states is greater than $1/2$, its genuine tripartite entanglement can be detected by $\sum_i |\langle O_i\rangle_{\rho}| > 3$. Furthermore, this criterion is tight for $\rho$ being any GHZ diagonal state, since it is genuinely tripartite entangled if and only if $\sum_i |\langle O_i\rangle_\rho| \le 3$.

The same data as in Eq.~\eqref{eq:gmecriterion} can be used for the detection of other types of entanglement structure.
Similarly, the upper bound in Eq.~\eqref{eq:gmecriterion} is $\min \{t_1, t_2, t_3\}$ for $\rho$ any fully separable state, since there is entanglement whenever the state is not separable for one partition. By taking the weight vector $w=(0,1,1,2,0,0,0)$, the upper bounds for GHZ states, biseparable states, and fully separable states are $4$, $3$ and $2$, respectively, leading to a distinction of different entanglement structures.
In comparison, the witness constructed in Ref.~\cite{zhou2019detecting} can only tell whether there is genuine entanglement or not. 
More details are provided in Appendix~\ref{ssec:entanglement}.

In the general case, any graph state $|G\rangle$ leads to a similar criterion of genuine multipartite entanglement as in Eq.~\eqref{eq:gmecriterion}. More explicitly,
\begin{equation}
	\sum\nolimits_i w_i |\langle O_i\rangle_{\sigma}| \le \max_{\rm bp} \beta(G_{\rm bp},w),
\end{equation}
where `bp' stands for bipartition and $G_{\rm bp}$ is the corresponding frustration graph of the operators in one part, e.g., $\{S_i\}$ in Eq.~\eqref{eq:splitted}.
Equivalently,
\begin{equation}\label{eq:gmecriterion2}
\{|\langle O_i\rangle_{\sigma}| \}_{i} \in  \conv\left(\bigcup_{\rm bp}{\be(G_{\rm bp})}\right).
\end{equation}
To overcome the difficulty in the characterization of $\be(G_{\rm bp})$, we can carefully choose a subset of $\{O_i\}$ such that $G_{\rm bp}$ is $\hbar$-perfect for any bipartition. This is always possible since all graphs with no more than $6$ vertices are $\hbar$-perfect. Then the right-hand side of Eq.~\eqref{eq:gmecriterion2} becomes the union of polytopes, which is much easier to characterize by using the vertex representation.

As the size of the stabilizer group increases exponentially with the number of parties, the detection with all the operators in the stabilizer group becomes heavy in practice. To reduce the cost of measurement resources, we design an adaptive strategy by adding measurements sequentially for different degree of entanglement. We illustrate this strategy still with the triparite GHZ states and the stabilizers in Eq.~\eqref{eq:stabilizer}. Notice that the mean values of $ZZ\id, Z\id Z, \id ZZ$ can be obtained from the measurement $ZZZ$ by taking the marginals. Hence, we treat them as one measurement. Then with $k\ge 2$ measurements from Eq.~\eqref{eq:stabilizer} in order, we have $\sum_{i=1}^{k+2} |\langle O_i\rangle_{\sigma}| \le 3$ for any biseparable state $\sigma$. In comparison, any of the $8$ GHZ states results in a violation with value $k+2$. This implies that  $k$ measurements are enough for the detection of entanglement when the coefficient $p > 3/(k+2)$ in the mixed state $p\, {\rm GHZ} + (1-p) \id/8$, e.g., $k=2$ for $p>3/4$. Furthermore, by considering the weighted version as in Eq.~\eqref{eq:gmecriterion}, the previous discussion with the weight vector $w=(0,1,1,2,0,0,0)$ reveals that only two measurements are sufficient for the detection of different entanglement structures. Thus, an adaptive strategy can also be constructed for this task against different degrees of entanglement.

\subsection{Entanglement estimation}
Any linear witness of entanglement can be used for entanglement estimation~\cite{sun2024bounding}. Under the assumption that an effective lower bound of the distance between a given point and the set $\be(G)$ or $\conv(\cup_{\rm bp}\be(G_{\rm bp}))$ can be obtained,
 we also convert the entanglement criteria in Eqs.~\eqref{eq:entb} and \eqref{eq:gmecriterion2} into a method of entanglement estimation.
In the case that all $G_{\rm bp}$ are all $\hbar$-perfect graphs, we can employ Gilbert algorithm to estimate the distance. In the general case, relaxations based on state polynomial optimization might be an option, which we leave for future research.
Another option is to employ maximal $\hbar$-perfect subgraphs of $G_{\rm bp}$, each of which leads to some linear constraints and their combination results in a polytope as a relaxation of $\be(G_{\rm bp})$.
Here we adopt the taxicab distance to estimate the genuine tripartite entanglement measured by the trace distance, i.e., ${\cal E}_T$, with the operators $\{O_i\}_{i=1}^7$ in Eq.~\eqref{eq:stabilizer}. Another example based on the Hilbert-Schmidt distance is provided in Appendix~\ref{ssec:entanglement}.

For the target state $\rho$, denote by $d_T(\rho)$ the distance from the point $\{|\langle O_i\rangle_{\rho}|\}$ to the set $\conv(\cup_{\rm bp}\be(G_{\rm bp}))$, which is a polytope. By definition, for any biseparable state $\sigma$, we have
\allowdisplaybreaks
\begin{align}
  d_T(\rho) &\le \sum\nolimits_i ||\langle O_i\rangle_{\rho}| - |\langle O_i\rangle_{\sigma}|| \nonumber\\
	    &\le \sum\nolimits_i |\langle O_i\rangle_{\rho} - \langle O_i\rangle_{\sigma}|  \nonumber\\
	    &= \max_{s_i = \pm 1}\sum\nolimits_i s_i \tr[O_i(\rho - \sigma)] \nonumber\\
	    &= \max_{s_i = \pm 1}\tr\big[\sum\nolimits_i s_i O_i(\rho - \sigma)\big] \nonumber\\
	    &\le \max_{s_i = \pm 1} \lambda_{\rm gap} \big( \sum\nolimits_i s_i O_i \big) T(\rho, \sigma) \nonumber\\
	    &= 8\, T(\rho,\sigma),
\end{align}
where $T(\rho,\sigma)$ is the trace distance between states $\rho$ and $\sigma$, $\lambda_{\rm gap}$ is the difference between the maximal and minimal eigenvalues.
The inequality in the second last line holds due to $\tr(\rho-\sigma)=0$, then there exists $\epsilon$ such that the maximal singular value of $\sum_i s_i O_i + \epsilon \id$ is $\lambda_{\rm gap}(\sum_i s_i O_i)/2$.  The equality in the last line comes from the fact that $\lambda_{\rm gap}=8$ for any combination of $s_i=\pm 1$ in this special case. In the general case, the calculation of $\lambda_{\rm gap}$ reduces to an integer program since $O_i$ has eigenvalues $\pm 1$ and they commute with each other.
By minimizing over the biseparable state $\sigma$, we obtain the estimation
\begin{equation}
	{\cal E}_T(\rho) \ge d_T(\rho)/8.
\end{equation}
This lower bound is tight, which can be verified by the fact that ${\cal E}_T(\rho) = 1/2$ for $\rho$ any of the $8$ GHZ states and
$d_T(\rho) \ge |\sum_i |\langle O_i\rangle_{\rho}| - \sum_i |\langle O_i\rangle_{\sigma}|| \ge 4$ for any $\sigma$ that is biseparable. 
In comparison, no linear witness can detect or be used to estimate the entanglement of the two GHZ states $(|000\rangle \pm |111\rangle)/\sqrt{2}$ at the same time.

Our method sheds new light on the detection and estimation of high-dimensional and multipartite entanglement by providing an easy but fruitful way to construct entanglement criteria, which are more effective than linear entanglement witness in the detection and estimation of entanglement.

\section{Sample complexity of shadow tomography}\label{sec:shadow}
The uncertainty relation based on variances is related to the sample complexity in the shadow tomography of Pauli strings~\cite{king2024triplyefficientshadowtomography,chen2024Optimal}. More exactly,
the sample complexity for shadow tomography of $n$ Pauli strings $\{S_i\}$ without quantum memory is lower bounded by $\Omega(1/[\epsilon^2 \delta(\{S_i\})])$ and upper bounded by $O(\log n/[\epsilon^2 \delta(\{S_i\})])$, where $\epsilon$ is the difference between the real value and the estimation, and
\begin{equation}
	\delta(\{S_i\}) = \min_{w\in {\cal D}} \max_{|\psi\rangle} \sum\nolimits_i w_i \langle \psi|S_i|\psi\rangle^2
\end{equation}
with ${\cal D}$ being the set of probability distributions, that is, $\sum_i w_i = 1$ and $w_i\ge 0$.
Hence, by definition,
\begin{equation}\label{eq:delta}
	\delta(\{S_i\}) = \min_{w\in {\cal D}} \beta(G,w),
\end{equation}
where $G$ is the frustration graph of $\{S_i\}$.
The direct implication is that the sample complexity depends only on the frustration graph $G$, but not the exact set of Pauli strings.

The parameter $\delta(\{S_i\})$ is generally hard to calculate. 
Based on the hierarchy in Eq.~\eqref{eq:hierarchy} to estimate $\beta(G,w)$, we propose an SDP hierarchy to approximate $\delta(\{S_i\})$ at the $r$-th level by $1/\omega_r$, where
\begin{align}\label{eq:deltaapp}
	\omega_r \coloneqq &\max \sum\nolimits_i w_i\nonumber\\
	\mathrm{s.t.}\ & \lambda_r(G,w) \le 1,\nonumber\\
	     & w\ge 0,
\end{align}
and $\lambda_r(G,w)$ is the optimal value of the $r$-th level in the original complete hierarchy~\cite{moran2024Uncertainty} to calculate $\beta(G,w)$.
The conditions in Eq.~\eqref{eq:deltaapp} imply that
the set of feasible points of $w$  is the polar set of feasible points in the SDP for $\lambda_r(G,w)$ projected into the dimensions related to the monomials in the form $x_i\langle x_i\rangle$ and intersected with $\mathbb{R}^n_+$, which again is characterized by conditions resulting in an SDP. More details are provided in Appendix~\ref{ssec:hierprop}.
As $\lambda_r(G,w)$ converges to $\beta(G,w)$, we have that $1/\omega_r$ converges to $\delta(\{S_i\})$.
In the case where $G$ is $\hbar$-perfect, the calculation of $\delta(\{S_i\})$ can be simplified due to the following observation.
Denote $\alpha^*$ and $\vartheta$ the fractional packing number~\cite{schrijver1979fractional} and the Lov\'{a}sz number of a graph.
\begin{theorem}\label{ob:deltabs} 
	For a given set of Pauli strings $\{S_i\}$ with $G$ being its frustration graph, the sample complexity index
	 $ \delta(\{S_i\}) \in [1/\alpha^*(\bar{G}), 1/\vartheta(\bar{G})]$,
 and the lower bound is achieved whenever $G$ is $\hbar$-perfect. 
\end{theorem}
This result also follows from the consideration of polar sets, whose exact proof is given in Appendix~\ref{ssec:tomography}.
Except that the difficulty in the calculation of the parameter $\delta$ reduces for $\hbar$-perfect graphs,  the estimation range can also be carried out with linear and semidefinite programs for the general case.
Thus, for graph $G$ to be $\hbar$-perfect, the sample complexity is lower bounded by $\Omega(\alpha^*(\bar{G})/\epsilon^2)$ and upper bounded by $O(\alpha^*(\bar{G})\log n/\epsilon^2)$. The general lower bound is
$\Omega(\vartheta(\bar{G})/\epsilon^2)$, which could be much worse since for every $\varepsilon > 0$ there exists graph $G$ such $\alpha^*(\bar{G})/\vartheta(\bar{G}) = O(n^{1-\varepsilon})$~\cite{linz2024lsystemslovasznumber}\footnote{Notice that $\alpha^*(\bar{G}) = \chi_f(G) \ge n/\alpha(G)$. As shown in ~\cite[proof of Thm 1.6]{linz2024lsystemslovasznumber}, the generalized
Johnson graph introduced there can lead to $n/(\alpha(G) \vartheta(\bar{G})) = O(n^{1-\varepsilon})$ for any $\varepsilon > 0$. }. 
This result builds a channel between graph theory and sample complexity. As an example, for a random graph $G(n,p)$ with $p$ being the probability of each edge~\cite{coja2005lovasz}, its complement graph is $G(n,1-p)$. Since $\vartheta(G(n,1-p)) = \sqrt{n/(1-p)}$, we conclude that the general lower bound is $\Omega( \sqrt{n/(1-p)} /\epsilon^2)$. 

In addition, $\beta(G,w)$ is a convex function of $w\ge 0$.
This property can be used to simplify the calculation of $\delta(\{S_i\})$ in the general case even if $G$ is $\hbar$-imperfect.
More explicitly,
by employing the symmetry of the graph, we only need to consider symmetric probability distributions in Eq.~\eqref{eq:delta}. 
If $\delta(\{S_i\})$ can be achieved by a probability distribution $w$, then $\delta(\{S_i\}) = \beta(G,w) = \beta(G,w_P)$ where $w_P$ is permuted from $w$ by any permutation $P$ in the automorphism group ${\cal P}$ of graph $G$. Hence, 
\begin{equation}
	\delta(\{S_i\}) = \sum\nolimits_{P\in {\cal P}}  \beta(G,w_P)/ |{\cal P}| \ge \beta(G,\bar{w}),
\end{equation}
where $\bar{w} = \sum_{P\in {\cal P}} w_P/|{\cal P}|$ and the last inequality is due to the convexity of $\beta(G,w)$.
Especially, when $G$ is vertex-transitive, any vertex can be changed to another one by a permutation in the automorphism group. This implies that $\delta(\{S_i\}) = \beta(G)/n$, which can be achieved with $w$ being the uniform distribution. With this, we find that the ratio of the complexity of anti-cycle $\bar{C}_n$ over the one of cycle $C_n$ is at least $\Omega(n/\log n)$.

Furthermore, for any graph $G$, it has been proven that $\delta = 1/\chi_f(G)$ when a strategy with Clifford measurements is implemented~\cite{chen2024Optimal}. Another main consequence of Observation~\ref{ob:deltabs} is that such a strategy is in fact optimal for $\hbar$-perfect graphs since $\chi_f(G) = \alpha^*(\bar{G})$. See Appendix~\ref{ssec:tomography} for more details.

\section{\texorpdfstring{Uncertainty relations for $\mathbf{\hbar}$-perfect graphs}{}}
\label{sec:uncertainty_relations}
Uncertainty relations play a fundamental role as a quantitative tool for bounding possible quantum behaviors.  
We can interpret a Pauli-string $S_i$ as Hermitian observable with eigenvalues $\pm 1$ and projectors $P_i^\pm = (1\pm S_i)/2$.  
For a set of Pauli strings, a lower bound on the (weighted) sum of the variances of the corresponding measurements is an uncertainty relation \cite{huang2012variance,schwonnek2018uncertainty}. Estimating these bounds is not only a popular rabbit hole on its own~\cite{de2023uncertainty,hastings2022optimizing,xu2023bounding}, but also the starting point for several applications that typically use uncertainty as a fingerprint of quantumness~\cite{oppenheim2010uncertainty,catani2022nonclassical}. For Pauli-string observables, there is a close connection to the weighted beta number. A quick computation \cite{de2023uncertainty} reveals the weighted uncertainty relation
\begin{align}\label{eq:ucr}
    \sum w_i \Delta^2_\rho S_i\ge  \sum\nolimits_i w_i -  \beta(G,w),
\end{align}
where $\{S_i\}$ is a set of Pauli strings. The lower bound is tight in the sense that there is always a state $\rho$ to achieve the lower bound for any given weight vector $w$.
Denote $\bar{\uparrow}\,T = \{(x_1,\ldots,x_n)\mid \exists t\in T\ \mathrm{s.t.}\ t_i \le x_i \le 1\}$.
From another perspective,
the definition of $\be(G)$ and the fact that $\Delta^2_\rho S_i = 1 - \langle S_i\rangle^2$ imply that
the convex hull $\conv(\bar{\uparrow}\,\{(\Delta_\rho^2 S_1, \ldots \Delta_\rho^2 S_n)\}_\rho)$ is  a transformation of $\be(G)$, that is, $\{(1-x_1,\ldots,1-x_n) \mid x\in \be(G)\}$, which is a polytope when $G$ is $\hbar$-perfect.

Beyond direct uncertainty relations, a primordial question that gets a qualitative answer by a relation like \eqref{eq:ucr} is \cite{schwonnek2018uncertainty}: How small can an uncertainty, say $\Delta^2_\rho S_i$,  be, given that all other uncertainties $\Delta^2_\rho S_{j\neq i}$ are not larger than some constants $c_j$.
Even with relation \eqref{eq:ucr} and the characterization with $\be(G)$ in hand, this optimization task is generally very difficult, since $\beta(G,w)$ is typically highly non-linear.
In the general case, we can still construct an SDP hierarchy for the numerical estimation of $\Delta_\rho^2 S_i$, as a variant of the one estimating the beta number in Eq.~\eqref{eq:hierarchy}:
\begin{equation}\label{eq:ucrsdp}
\begin{aligned}
 \inf\,\,\, &1-\langle x_i\rangle^2\\
\rm{s.t.}\ & M_r\succeq0,\\
&[M_r]_{1,1}=1, \quad 1-\langle x_j\rangle^2 \le c_j, \,\forall j\ne i,\\
&[M_r]_{u,v}=[M_r]_{a,b}, \text{ if }\langle u^*v\rangle=\langle a^*b\rangle,
\end{aligned}
\end{equation}
where the notations are the same as in Section~\ref{sec4a}, $\langle x_i\rangle^2$ is an element in the matrix $M_r$, and $1-\langle x_i\rangle^2$ corresponds to $\Delta^2_\rho S_i$.

Nevertheless, for $\hbar$-perfect graphs, this task can be reduced to a linear program since the set of all possible $\{\Delta^2_\rho S_i\}_\rho$ is a transformation of $\stab(G)$ together with the extra linear constraints imposed by $\Delta^2_\rho S_i \le c_i$ in this situation. That is,
\begin{equation}\label{eq:ucrlp}
\begin{aligned}
	\inf\,\,\, &x_i\\
	\mathrm{s.t.}\ &x_j \ge c_j, \,\forall j\neq i,\\
		    &x_j \in [0,1], \,\forall j,\\
		    &\sum\nolimits_i w_i (1-x_i) \le \alpha(G,w), \,\forall w\in W_f,
\end{aligned}
\end{equation}
where $W_f$ is the set of norm vectors corresponding to facets of $\stab(G)$.
In comparison with the SDP in Eq.~\eqref{eq:ucrsdp} whose size increases dramatically with the level $r$, the linear program in Eq.~\eqref{eq:ucrlp} is of fixed size (i.e., the number of vertices) and is thus much lighter.

\section{Ground state energies and independence numbers}
\label{sec:app}
In this section we will further explore the bridge created by the property $\beta(G,w)=\alpha(G,w)$. 
Passing this bridge in one direction allows for the approximation of the eigenvalues of a spin system via $\beta(G,w)$ which can now be done by classical tools for graphs.
Passing the bride in the other, allows for the formulation of $\alpha(G,w)$ as an optimization problem on spin systems, which for example can be tackled with a quantum computer.
We will see that both directions can be fruitful in certain situations.

\subsection{Computation of ground state energies from graph structures}
For real-valued coefficients $a_i$, let $H=\sum_i a_i S_i$ be a Hamiltonian. By an appropriate choice of $S_i$, a huge amount of Hamiltonians of interest can be expressed in this form. Examples range from literally every Hamiltonian considered in quantum computing~\cite{lloyd_terhal_2016,ciani2019hamiltonian} to most models in many-body physics~\cite{chapman2020characterization,chapman2023unified,wang2024certifying}. 
Also before the advent of quantum computing, it was and is a central task of several fields to  solve the dynamics of these models, to find the ground state energy of $H$, or to make  at least a good estimate. In most instances this task is, however, difficult and  attracts numerous analytical and numerical efforts.

The use of graph tools for the study of spin systems has a rich history, to which we want to add for our special class for graphs. 
A recent result with a related class of graphs shows that a system is solvable when the frustration graph of $H$ is claw-free~\cite{chapman2023unified}. However, the percentage of claw-free graphs decreases quickly as the number of vertices increases, as listed in Table~\ref{tab:kinds_ghs}, which means that most models cannot be solved in the corresponding approach~\cite{chapman2023unified}. As shown in Section~\ref{sec:props}, $\hbar$-perfect graphs contain this class of graphs. Even though the fraction of these graphs in the set of all graphs will vanish as the number of vertices goes to infinity, there are at least many more $ \hbar$-perfect graphs than claw free graphs. 

In general, we dare to claim efficient solvability of the corresponding models as there are some subtile assumptions on the efficiency of graph computations.

Given an $\hbar$-perfect graph and the access to the function 
$\beta(G,w)$ allows to compute outer bounds on the spectrum of $H$. Note that this is the  more challenging direction, as it is not covered by typical variation methods that approach the state space from the inside.  
For any state $\rho$, the Cauchy-Schwarz inequality implies that
$\langle H\rangle^2_\rho = (\sum_i a_i \langle S_i\rangle_\rho)^2 \le \left(\sum_i a_i^2/w_i \right)\left(\sum_i w_i \langle S_i\rangle^2_\rho\right) \le \left(\sum_i a_i^2/w_i\right) \beta(G,w)$. By minimizing over all positive weight vectors $w$ and taking the square root, we get
\begin{equation}
\begin{aligned}
    \label{eq:approx}
   \inf_\rho \langle H \rangle_\rho  &\geq -\Big(\inf_w  \Big( \sum\nolimits_i a_i^2/w_i\Big)\beta(G,w)\Big)^{1/2}\\
				     &\geq -\inf_{w: \beta(G,w)=1}  \Big( \sum\nolimits_i a_i^2/w_i\Big)^{1/2}.
\end{aligned}
\end{equation}
In principle, we could construct a convex SDP hierarchy for computing the above low bound in a similar spirit to Eq.~\eqref{eq:deltaapp}:
\begin{align}\label{gs-sdp}
\inf\,\,\,&\sum_i a_i^2/w_i\nonumber\\
    \mathrm{s.t.}\ & \lambda_r(G,w) \le 1,\nonumber\\
	     & w\ge 0.
\end{align}
Again, $\hbar$-perfectness simplifies the optimization over $\be(G)$ to the optimization over polytope $\stab(G)$. 
That is, the square of the ground state energy is no more than the optimal value of the following program:
\begin{align}
    \label{eq:approx3}
    \inf_w \,\,\,  &\sum\nolimits_i a_i^2/w_i\nonumber\\
    \mathrm{s.t.}\ & \sum\nolimits_{i\in I} w_i \le 1,\, \forall I \in {\cal I},\nonumber\\
    & w \ge 0,
\end{align}
where ${\cal I}$ contains all the independence sets of $G$. Since the cost function is convex for $w$, Eq.~\eqref{eq:approx3} defines a convex program with linear constraints which can be solved efficiently given efficient access to $\mathcal{I}$.

As a simple example, we take the series of Hamiltonian $H_n$ whose frustration graph is $C_{2n+1}$, that is,
\begin{equation}
	H_n = X_1 + Z_1 + Y_n + \sum\nolimits_{i=1}^{n-1} (X_iX_{i+1} + Z_i Z_{i+1}).
\end{equation}
Due to the symmetry of $C_{2n+1}$ and the convexity of the aim function in Eq.~\eqref{eq:approx3},
the optimal solution in Eq.~\eqref{eq:approx3} provides the estimation that $\langle H_n\rangle \ge -\sqrt{n(2n+1)}$, as $\alpha(C_{2n+1}) = n$.
For $n = 2, \ldots, 7$, we calculate the ground state energy numerically and compare it with the estimation.
The gaps between them for $n=2, \ldots, 7$ in order are
$0.084594, 0.13099, 0.098687, 0.043419, 0.036549, 0.012361$, which decreases in general.

In the aforementioned case, the optimal estimation is achieved by the weight vector with all elements equal to each other.
In the general case, a weight vector other than the uniformly distributed ones can indeed result in a tighter estimation.
The idea is to increase the components of $w$ without increasing the weighted independence number. 
We take the example
\begin{equation}
	\tilde{H}_2 = X_1 + Z_1 + Y_2 + \left( X_1X_2 + Y_1Y_2 + Z_1Z_2 \right).
\end{equation}
Let $G$ be the frustration graph of the Pauli strings in $\tilde{H}_2$ in order.
Then,
\begin{align}
	\langle \tilde{H}_2\rangle^2 \le \inf_{w\ge 0} \sum_{i=1}^6 1/w_i \beta(G,w) \le 6 \beta(G) = 18.
    \end{align}
However, by taking $w=(2,2,1,1,1,1)$, then $\beta(G,w) = \beta(G)=3$ while $\sum_{i=1}^6 1 /w_i = 5$. This provides a tighter bound $\langle \tilde{H}_2\rangle^2 \le 15$. In comparison, the exact ground state energy is $-3.722935$, while the estimated lower bound is $-\sqrt{15} \approx -3.872983$, which is better than the one $-\sqrt{18} \approx -4.242641$ with $w=(1,1,1,1,1,1)$. 
\subsection{Memory friendly encodings of the independence number}\label{ssec:encoding}
In this subsection, we explore the perspectives of quantum algorithms for the estimation of the independence number of a graph. 
It is known that for  general instances the computation of independence numbers is an NP-hard problem. So, it is unlikely to expect any substantial advantage on quantum computers without further structural assumptions. 
Nevertheless,  $\hbar$-perfectness gives rise to further structures that we can employ in a meaningful way. Here, the main advantages of using quantum computers can be memory efficiency. 
Typically, the encoding of a graph into a quantum circuit or Hamiltonian builds on using one qubit per vertex~\cite{Farhi2014AQA,Zhou2018QuantumAO,Skolik2022EquivariantQC}. For large graphs, say more than $10^4$ vertices, a potential quantum advantage becomes relevant but also challenging, since devices this with this number of qubits do not exist.
Here, the encoding into a set of Pauli strings implies a potentially fruitful alternative as this encoding can be much (exponentially) denser. 

Recall that for  $\ell$ qubits, there are $4^\ell$ Pauli strings. Picking any subset $S$ of these strings will encode a graph. 
Hence, it is easy to find graphs with qubit encoding at the scale of $O(4^\ell)$. Indeed, clever counting\footnote{Let $S$ be a set of length-$\ell$ Pauli strings. Let $S^c$ be its complement in the set of all Pauli strings of length $\ell$. At least one of these two sets has exponentially many elements, Since $S$ and $S^c$ have the same probability to pick in a uniform sample, the average number of elements will be exponential.} reveals that a randomly and uniformly picked subset of these $4^\ell$ strings will, on average, encode an exponential amount of vertices. 
Beyond this statistical statement, it makes sense to ask whether a given graph $G$ will admit such a dense encoding. Let  $n$ be the number of vertices of a graph $G$.  Such graphs can be identified with the help of the following theorem, whose proof is given in Appendix~\ref{sec:proofs2}.

\begin{theorem}
\label{thm:minlen1}
Let $G$ be a frustration graph with adjacency matrix $A$.  Then the length of Pauli strings in any realization of $G$ is no less than ${\operatorname{rank}_{\mathbb{F}_2}(A)}/2$  with $\operatorname{rank}_{\mathbb{F}_2}(A)$ being the rank of $A$ over the field $\mathbb{F}_2$, which is a tight bound for every $G$.
\end{theorem}
 We note that recently, the concept frustration graphs with weighted edges was proposed for the strings of Weyl-Heisenberg matrices~\cite{Makuta2025FrustrationGF}, where results similar to Theorem~\ref{thm:minlen1} were proven for special kinds of frustration graphs with weighted edges.
 
In the following, we will assume a regime of $n=c\, 4^\ell$ with $c\approx 1$ for the Pauli string encoding of $G$. 
Based on Theorem \ref{thm:minlen1}, a possible  method for  constructing examples of  graphs with exponentially dense encoding is to  start from a small adjacency matrix $A_0$ with $\operatorname{rank}_{\mathbb{F}_2}(A_0) = 2l$ and then use linear combinations to expand it to an adjacency matrix $A$ of size $n = c\, 4^\ell$. 

Albeit a dense encoding into Pauli strings can be efficient to implement on a quantum device, it is not clear at all that useful quantities of a graph can be easily accessed. For the computation of $\alpha(G)$ or $\beta(G)$, we are confronted with the fact that these quantities do not arise from a linear optimization, i.e. as the ground state of a Hamiltonian. For variational methods like QAOA or other types of quantum machine learning, this does not impose an inevitable issue, as they can in principle handle any type of well behaved cost function. Nevertheless, as many case studies show,  their practical performance can be quite poor. 
It is therefore useful to find Hamiltonian formulations of $\beta(G)$, and thus $\alpha(G)$ under $h$-perfectness. 
We will do this in the following. Our ansatz is based on the de Finetti hierarchy already presented in Section~\ref{sec:numerical}.

\subsubsection{Approximate encoding in a single Hamiltonian}
A full adaption of the de Finetti hierarchy, presented as a classical numerical method in Sections~\ref{sec:numerical} and \ref{ssec:exactalpha}, to a quantum encoding carries some challenges. The projection onto the Bose symmetric subspace on a multi-qubit system is potentially hard to implement using a basic gate set, as it requires post-selection and globally entangling operations. 

Here we investigate the alternative route of an incoherent symmetrization obtained by mixing over permutations of subsystems.
Recall the two-copy Hamiltonian given  in \eqref{eq:derivation_deFinetti_argument}, i.e. $H_S = \sum_{i=1}^n w_i S_i\otimes S_i$.
Extending $H_s$ by $m$ copies of $\ell$ qubits and making it permutation invariant by averaging over all permutations of these tensorfactors gives the Hamiltonian
\begin{align}
\tilde{H}_S^{(m)} \;=\; \frac{1}{2\binom{m+2}{2}}
\sum_{a\neq b=1}^{m+2}\;
\sum_{i=1}^{n} w_i\,
S_i^{(a)} \otimes S_i^{(b)} \,,
\end{align}
where we used the notation $S_i^{(a)}$ to denote the measurement of the string $S_i$ on the $a$-th tensorfactor. Again, de Finetti guarantees us that the maximal eigenvalue of this Hamiltonian will be attained by a product state and approximates  $\beta(G)$. We make this precise as follows.

\begin{theorem}
    Let a graph with $n$ vertices be encoded in $l=1/2 \log_2(n/c)$ qubits. Assume weights normalized by $\Vert w \Vert_1=n$. Then, the maximal eigenvalue of $\tilde{H}_S^{(m)}$ approximates $\beta(G,w)$ by 
    \begin{align}
        \beta(G,w) \leq \lambda_{\max}\left( \tilde{H}_S^{(m)} \right) \leq \beta(G,w) + n\sqrt{ \frac{\log(n)-\log(c)}{m}}.
    \end{align}
 That is, we obtain a linear Hamiltonian encoding of the number $\beta(G,w)$ with an additive regularized error 
    \begin{align}
        \varepsilon_{reg} =\left|\lambda_{\max}(\tilde{H}_S^{(m)})/n -\beta(G,w)/n\right|
    \end{align}
    by using in total $L$ qubits with 
    \begin{align}
        L=(m+2)l\sim  O\left(\frac{\log(n/c)^2 }{\varepsilon_{reg}^2}\right).
    \end{align}
\end{theorem}

The proof of this theorem is based on the use of LOCC de Finetti theorems; see  \cite{li2015quantum,brandao2013quantum}, and can be found in the Appendix~\ref{ssec:alpha}.

\subsubsection{Exact computation via the de Finetti hierarchy}\label{ssec:exactalpha}
Given that the projection onto Bose subspaces can be  implemented, e.g. on an error corrected quantum device or in the spirit of a quantum inspired classical method, we can fully adapt the numerical method from Section \ref{sec:numerical} in order to compute  
the weighted independence number $\alpha(G,w)$ for $\hbar$-perfect graphs. Here, the Hamiltonian $H_S^{(m)}$, which acts on the Bose subspace, gives an encoding. 
Assume the weights are normalized by $\|w\|_1 = n$, i.e., by the number of vertices. Then Eq.\eqref{eq:mfdiff} implies that
\begin{equation}
	|\lambda_{\max}(H_S^{(m)}) - \alpha(G,w)| \le  4nd/(m+2),
\end{equation}
where $d=2^\ell$ is the dimension of the Pauli strings, $m$ is the level of the hierarchy. 
We know that $\alpha(G)$ is an integer. Hence, we can obtain this number exactly by reducing the absolute de Finetti error below the threshold of $1/2$ and rounding. In order to leave room for inaccuracies arising from the computation of the ground state energy, we will take the somewhat arbitrary value  $\varepsilon_{abs}\leq1/4$. 
That is, we need to guarantee $4nd/(m+2) < 1/4$ in order to get an estimate of $\lambda_{\max}(H_S^{(m)})$ with an error no more than $1/4$. Hence, taking $m = 16nd - 1$ copies of $\ell$ in the de Finetti hierarchy is enough.
The analysis in Section~\ref{sec:numerical} implies that the runtime for the estimation of $\lambda_{\max}(H_S^{(m)})$ up to an error no more than $1/4$ is $O(2^{3d\log(4nd)})$, which is the time complexity for
the exact solution of $\alpha(G,w)$ with this quantum-inspired classical algorithm. This algorithm then has an advantage for large $\hbar$-perfect graphs encoded in relatively short Pauli strings, i.e. $d = o(n)$. In comparison, the classical algorithms for a general graph with $n$ vertices takes exponential time, e.g., $O(2^{n/3})$ in Ref.~\cite{Tarjan1976FindingAM} and $O(1.1996^{n})$ in Ref.~\cite{Xiao2013ExactAF}. Especially in the case where $n = c\,4^\ell$, we have $d= \sqrt{n/c}$, and the time complexity $O(2^{3d\log(4nd)})$ becomes $O(2^{3/2\sqrt{n/c}[3\log(n) + \log(16/c)]})$, which has a clear advantage in runtime for large $n$.

Notice that there are also quantum algorithms for computing the maximum eigenvalue, e.g., the variational quantum eigenvalue solver~\cite{Peruzzo2013AVE} and ground state preparation~\cite{Lin2020NearoptimalGS} to solve the ground state energy problem on a quantum computer. 
Since $H_S^{(m)}$ is in the Bose-symmetric subspace of dimension $D = {m+d+1 \choose d-1}$ which is less than $[(16n+1)\sqrt{d+2}]^{(d-1)}$, it can be encoded via $ L = \log D = O( (d-1) \log((16n+1)\sqrt{d+2}) )$ qubits. This implies again a memory advantage over the typical encoding methods~\cite{Farhi2014AQA,Zhou2018QuantumAO,Skolik2022EquivariantQC} when $d=o(n)$.

\subsection{\texorpdfstring{Encoding of the independence number for $\mathbf{\hbar}$-imperfect graphs}{}}
To tackle $\hbar$-imperfect graphs, we introduce a generalization of the beta number, that is,
\begin{equation}
	\beta(G,w,k) = \max_{\rho} \sum\nolimits_i w_i |\langle S_i\rangle^k_{\rho}|.
\end{equation}
By definition, $\beta(G,w) = \beta(G,w,2)$.
\begin{proposition}
\label{prop:convergence_generalized_beta}
Given a graph $G$, as $k$ goes to infinity, $\beta(G,w,k)$ converges to the weighted independence number $\alpha(G,w)$.
More precisely, for any nonnegative weight vector $w$,
\begin{equation}
	\lim_{k\to \infty} \beta(G,w,k) = \alpha(G,w).
\end{equation}    
\end{proposition}
The convergence of $\beta(G,w,k)$ is based on the fact that the limit of $|\langle S_i\rangle^k_{\rho}|$ is either $0$ or $1$ as $k$ tends to infinity. The full proof is given in Appendix~\ref{ssec:gbn}.

For a given graph $G$, assume that $\{S_i\}$ is a basic representation of $G$ with Pauli strings of length $\ell$ and $d=2^\ell$.
By extending the hierarchy to $\beta(G,w,k)$ with $k$ an even number, the finite de Finetti theorem~\cite[Thm.~II.8]{Christandl2007} similarly ensures that
\begin{equation}
	\big|\beta(G,w,k) - \lambda_{\max}(H_{S,k}^{(m)})\big| \le n (2dk/(m+k)),
\end{equation}
where weights are normalized by $\|w\|_1 = n$, $H_{S,k}^{(m)} = P_{\rm Sym}^{(m+k)} (H_{S,k}\otimes \id^{\otimes m}) P_{\rm Sym}^{(m+k)}$, and $H_{S,k} = \sum_i S_i^{\otimes k}$.
This provides a hierarchy of classical SDP to calculate $\alpha(G,w)$ for $\hbar$-imperfect graphs, which potentially has advantages for large graphs encoded in relatively short Pauli strings.

On the other hand, we can approximate $\beta(G,w,k)$ and consequently $\alpha(G,w)$ on a quantum computer. Denote
\begin{equation}
    \tilde{H}_{S,k}^{(m)} = \frac{1}{|{\rm Per}(m\!+\!k)|} \sum_{\sigma\in {\rm Per}(m\!+\!k)} \sum_{i=1}^n w_i \bigotimes_{j=1}^k S_i^{(\sigma(j))},
\end{equation}
where ${\rm Per}(m+k)$ is the permutation group on $m+k$ elements. 
Then for any graph $G$ with $n=c\,4^\ell$ vertices encoded in Pauli strings of length $\ell$, the finite de Finetti theorem for fully one-way LOCC~\cite[Thm. 1]{li2015quantum} similarly ensures that
we obtain a linear Hamiltonian encoding of the number $\beta(G,w,k)$ with an additive regularized error $\varepsilon_{reg}$
    by using in total $L$ qubits with 
    \begin{align}
        L=(m+k)l\sim  O\left(\frac{[(k-1)\log(n/c)]^2 }{\varepsilon_{reg}^2}\right).
    \end{align}
\section{Outlook and Conclusion}\label{sec:conclusion}
For some of the graphs considered here, beneficial aspects and surprising properties have been reported before. However, a structured investigation was missing. This work provides a first step towards this. 
We formalize the class of $\hbar$-perfect graphs and study their properties, implications, and applications. 

A central question arising in this context is to decide the $\hbar$-perfectness of a given graph. Using the properties developed, this task can now be tackled in a constructive manner. 
Another aspect is the role of $\hbar$-perfectness in the analysis of problems related to optimization over Pauli strings. Here, having an $\hbar$-perfect frustration graph bears huge simplifications for convex optimization problems formulated in terms of $\beta(G,w)$, since they can be cast as linear programs on the stable set polytope. 
That being said, one should not leave aside that finding this polytope is a problem that can indeed have hard instances when tackled classically. 
Here, the bridge set by $\hbar$-perfect graphs gives another interesting perspective by opening up the use of quantum optimization methods. 

There are many questions that remain for future work:

$\bullet$ The definition of $\hbar$-perfect graphs actually does not depend on the concrete form of Pauli strings, but only depends on the algebraic relations that they satisfy. It would be very interesting to generalize the definition of $\hbar$-perfect graphs to other relevant operators in quantum physics, which would also extend the applications of $\hbar$-perfect graphs.

$\bullet$ The frustration graph with weighted edges encoding the commutation relations of generalized high-dimensional Pauli strings has been recently introduced in Ref.~\cite{Makuta2025FrustrationGF}, can the framework developed in this work be generalized for the high-dimensional Pauli strings?

$\bullet$ The characterization of $\hbar$-perfect graphs is not yet completed. The list of properties provided in this work can be very likely extended. One special task is to determine the $\hbar$-perfectness of $G_7$ in Appendix~\ref{ssec:hper2} by using properties only.

$\bullet$ Though we have made some advances in approximating weighted beta numbers, the design of more efficient algorithms will be essential for determining $\hbar$-perfect graphs numerically, especially when it comes to larger graphs. 
By this we can also extend the study on the typicality of $\hbar$-perfectness, to larger graphs.

$\bullet$ Are there any connections between weighted beta numbers and other quantum graph parameters, e.g., the ones studied in \cite{gribling2018bounds}?

$\bullet$ Is there another graph parameter $\alpha'(G)$ between the independence number $\alpha(G)$ and the beta number $\beta(G)$, which is relatively easy to compute, e.g., it can be computed in polynomial time? Then we can extend the sandwich theorem $\alpha(G)\le \vartheta(G) \le \theta(G)$ to $\alpha(G)\le \alpha'(G) \le \beta(G) \le \vartheta(G) \le \theta(G)$, where $\theta(G)$ is the clique covering number~\cite{Diestel2025GraphT}. 

$\bullet$ Can the independence number for $\hbar$-imperfect graphs be solved in bounded-error quantum polynomial time?

$\bullet$ If we randomly take ${\rm poly}(m)$ Pauli strings of length $m$, how would the sample complexity parameter $\delta$ change with the length $m$?

$\bullet$ Is the hierarchy of SDP relaxations for approximating weighted beta numbers proposed in this work complete? That is, does the sequence of upper bounds produced by the hierarchy converge to the weighted beta number? If not, can we make it complete by adding more but not whole state monomials to the monomial basis? Can the structures (say, sparsity or symmetry) of the graph be exploited to reduce the size of SDP relaxations arising from the hierarchy? 

$\bullet$ All $\hbar$-imperfect graphs with up to $9$ vertices have a chromatic number at least $4$, and it is natural to conjecture that it holds in general. In comparison, all imperfect graphs have a chromatic number at least $3$. If the conjecture is true, then the results~\cite{Balogh2004TheNO} in graph theory imply a tighter estimation of the ratio of $\hbar$-perfect graphs among all graphs.

$\bullet$ The examples provided in this work mainly serve the purpose of showcasing the use of $\hbar$-perfectness in possible applications. Using the tools developed in this work to tackle actual problem instances of relevant size is now open for future work, 
where the extra structures of the considered Pauli strings, e.g. $p$-local Pauli strings~\cite{anschuetz2024bounds}, would be helpful.

\section*{Acknowledgments}
The authors thank Dong-Ling Deng, Felix Huber, and Sixia Yu for inspiring discussions and comments.
Z.P.X. is supported by {National Natural Science Foundation of China} (Grant No.~12305007), 
Anhui Provincial Natural Science Foundation (Grant No.~2308085QA29, No.~2508085Y003), Anhui Province
Science and Technology Innovation Project (No.~202423r06050004). 
J.W. is funded by National Key R\&D Program of China under grant No.~2023YFA1009401 and Natural Science Foundation of China under grant No.~12571333.
GK acknowledges support from the Excellence Cluster - Matter and Light for Quantum Computing (ML4Q). GK acknowledges funding by the European Research Council (ERC Grant Agreement No. 948139).
A.W. is supported by the Spanish MICIN (project PID2022-141283NB-I00) with the support of FEDER funds, by the Spanish MICIN with funding from European Union NextGenerationEU (PRTR-C17.I1) and the Generalitat de Catalunya, by the Spanish MTDFP through the QUANTUM ENIA project: Quantum Spain, funded by the European Union NextGenerationEU within the framework of the ``Digital Spain 2026 Agenda'', by the Alexander von Humboldt Foundation, and by the Institute for Advanced Study of the Technical University Munich. 

\appendix
\onecolumngrid
\pagebreak
\include{appendix.tex}

\bibliography{ref.bib}

\end{document}

%% file: appendix.tex
\begin{center}
    {\Large \bf{Appendix}}\\ \vspace{1em}
\end{center}
The Appendix is structured as follows: In Appendix~\ref{sec:proofs}, we gather the proofs related to the properties of $\be(G)$ and $\hbar$-perfect graphs in Appendix~\ref{sec:proofs1}, the $\hbar$-perfectness of $G_7$ and $G_{15}$ in Appendix~\ref{ssec:hper2}, the relation between the length of representation and the adjacency matrix in Appendix~\ref{sec:proofs2}, and the convergence of generalized beta numbers in Appendix~\ref{ssec:gbn}. In Appendix~\ref{sec:numtools}, we provide the details of the new hierarchy based on state polynomial optimization in Appendix~\ref{ssec:spo}, its combination with a see-saw method in Appendix~\ref{ssec:seesaw_state}, the numerical results on all graphs up to $9$ vertices in Appendix~\ref{ssec:typical}, the further application of this hierarchy in shadow tomography in Appendix~\ref{ssec:hierprop}, and the asymptotic prevalence of $\hbar$-perfectness in Appendix~\ref{ssec:asym}. In Appendix C, we supplement the details of the applications. Especially, Appendix~\ref{ssec:entanglement} discusses entanglement detection and estimation, Appendix~\ref{ssec:tomography} analyzes the sample complexity of shadow tomography, Appendix~\ref{ssec:ground} focuses on ground state energy problems, and Appendix~\ref{ssec:alpha} explores the quantum encoding of the independence number.

\section{Proofs related to beta number and $\hbar$-perfectness}
\label{sec:proofs}

\subsection{\texorpdfstring{Proposition 1, Theorem 1 and Theorem 2, Properties 1 to 6}{}}\label{sec:proofs1}
\setcounter{theorem}{0}
\setcounter{proposition}{0}
\begin{proposition}\label{thm:beta}
For given $G$ as a frustration graph, any realization ${\cal S}$ leads to the same set  defined as
	\begin{equation}
		\be(G) \coloneqq \conv(\downarrow\hspace{-0.3em} {\cal Q}({\cal S})) \cap \mathbb{R}_+^n,
	\end{equation}
    where $\downarrow\hspace{-0.3em} T \coloneqq \{x|\exists y \in T\ s.t.\ x_i\le y_i\}$ for a given set $T$.
	Besides, $\beta(G,w) = \max_{v\in \be(G)} \sum_i w_i v_i$, $\forall w\in \mathbb{R}_+^n$.
\end{proposition}
\begin{proof}
    For a given graph $G$ and any two realizations ${\cal S}$ and ${\cal S}'$, let 
    \begin{equation}
        T = \conv(\downarrow\hspace{-0.3em} {\cal Q}({\cal S})) \cap \mathbb{R}_+^n,\ 
        T' = \conv(\downarrow\hspace{-0.3em} {\cal Q}({\cal S}')) \cap \mathbb{R}_+^n.
    \end{equation}
    By definition, both $T$ and $T'$ are convex sets. If there is a point $v \in T'\setminus T$, then there exists a hyperplane to separate $v$ from $T$ in the way that
    \begin{equation}\label{seq:tv}
        \max_{u\in T}\sum_i w_i u_i  < \sum_i w_i v_i,
    \end{equation}
    where $w$ is the norm vector of the corresponding hyperplane.

    Notice that all the elements  $v$ are nonnegative. Hence,
    \begin{equation}\label{seq:ww1}
        \sum_i w_i v_i \le \sum_i w^+_i v_i,
    \end{equation}
    where $w^+_i = \tau_i w_i$ with $\tau_i = \max\{0,{\rm sign}(w_i)\}$.

    Similarly,
    \begin{align}
        \max_{u\in T}\sum_i w_i u_i &\le \max_{u\in T}\sum_i w^+_i u_i\nonumber\\
        &= \max_{u\in T}\sum_{i} w_i \tau_{i} u_i\nonumber\\
        &= \max_{u\in T}\sum_{i} w_i (\tau_{i} u_i)\nonumber\\
        &\le \max_{u'\in T}\sum_{i} w_i u'_i,
    \end{align}
    where the last inequality follows from the fact that $u' = (u'_1,\ldots) \in T$ whenever $u \in T$ since $T$ is a downset. 
    This implies that
    \begin{equation}\label{seq:ww2}
        \max_{u\in T}\sum_i w_i u_i = \max_{u\in T}\sum_i w^+_i u_i = \beta(G,w^+).
    \end{equation}
    By combining Eqs.(\ref{seq:tv},\ref{seq:ww1},\ref{seq:ww2}), we conclude that
    \begin{equation}\label{eq:contradiction}
        \beta(G,w^+) < \sum_i w^+_i v_i,
    \end{equation}
    which is a contradiction by noticing that $ \beta(G,w^+) = \max_{v'\in T'} \max_i w^+_i v'_i \ge \sum_i w^+_i v_i$.

    Hence, $T'\subseteq T$, and for the same reason $T\subseteq T'$, which means that $T=T'$ and leads to the quantity $\be(G)$ independent of the realization.
\end{proof}
\begin{theorem}
	For any graph $G$, $\stab(G) \subseteq \be(G)$.
\end{theorem}
\begin{proof}
    Notice that $\be(G)$ and $\stab(G)$ are both convex sets. Besides, $\stab(G) = \downarrow\hspace{-0.3em} \stab(G)$. By replacing $T$ by $\be(G)$ and $T'$ by $\stab(G)$ in the proof of Theorem~\ref{thm:beta}, we arrive at a contradiction as in Eq.~\eqref{eq:contradiction} since $\sum_i w_i^+ v_i \le \alpha(G,w^+) \le \beta(G,w^+)$.
\end{proof}

\setcounter{prop}{0}
\begin{prop}[Fully connected union]
\label{thm:connected_union2}
If $G_1, G_2$ are $\hbar$-perfect and $G$ is the fully connected disjoint union of $G_1$ and $G_2$, then $G$ is also $\hbar$-perfect.
\end{prop}

\begin{proof}[Proof of Property \ref{thm:connected_union2}] 
Denote by $\{A_{i}\}$ any realization of $G$, and $V_1, V_2$ the vertex sets of $G_1$ and $G_2$. 
Notice that for any nonnegative weight vector $w$ we have
\begin{align}
      \sum_{i\in V_k} w_i \mean{A_i}^2 &= \max_{\rho}\mean{\sum_{i\in V_k} \sqrt{w_{i}}  x_{i} A_{i}}_{\rho}^2 \nonumber\\
      &\le \max_{\rho} \lambda_k^2 \mean{B_k}_{\rho}^2,
\end{align}
where $B_k = \sum_{i\in V_k} \sqrt{w_{i}}  x_{i} A_{i}/ \lambda_k$, $(x_i)_{i\in V_k}$ is the normalized vector of $\left(\sqrt{w_{i}} \mean{A_{i}}\right)_{i\in V_k}$, and $\lambda_k$ is the maximal eigenvalue of $\sum_{i\in V_k} \sqrt{w_{i}}x_{i}A_{i}$.

By definition, $\{B_1, B_2\} = 0$. Hence,
\begin{align}
      \beta(G,w) &= \sum_{i\in V_1} w_i \mean{A_i}^2 + \sum_{i\in V_2} w_i \mean{A_i}^2 \nonumber\\
      &\le \max_{\rho} \left[\lambda_1^2 \mean{B_1}_{\rho}^2 + \lambda_2^2 \mean{B_2}_{\rho}^2\right]\nonumber\\
      &\le \max_k \lambda_k^2.
\end{align}
However,
\begin{equation}
      \max_k \lambda_k^2 \le \max_k \alpha(G_k,(w_i)_{i\in V_k}) = \alpha(G,w),
\end{equation}
where the first inequality is from the facts that
\begin{equation}
    \lambda_k^2 \le \sum_{i\in V_k} w_i \mean{A_i}^2 \le \beta(G_k,(w_i)_{i\in V_k}),
\end{equation}
and $G_1, G_2$ are $\hbar$-perfect. The last equality is from the assumption that $G$ is the fully connected disjoint union of $G_1$ and $G_2$.
\end{proof}

\begin{prop}[Fully disconnected union]
\label{thm:disconnected_union2}
If $G_1, G_2$ are $\hbar$-perfect and $G$ is the fully disconnected disjoint union of $G_1$ and $G_2$, then $G$ is also $\hbar$-perfect.
\end{prop}
\begin{proof}[Proof of Property \ref{thm:disconnected_union2}]
The proof is similar to the one of Property~\ref{thm:connected_union2}. The difference is that in this case, $ [B_1,B_2] =0$ which leads to 
    \begin{equation}
        \beta(G,w) = \sum_k \lambda_k^2 \le \sum_k \alpha(G_k, (w_i)_{i\in V_k}) = \alpha(G,w).
    \end{equation} 
\end{proof}

\begin{prop}[Induced subgraphs]
\label{thm:subgraph2}
    For a $\hbar$-perfect graph $G$, if graph $G'$ is an induced subgraph of $G$ with the vertex set $V'$, then $G'$ is also $\hbar$-perfect.
\end{prop}
\begin{proof}[Proof of Property \ref{thm:subgraph2}] 
    For any nonnegative weight $w'$ of $G'$, denote by $w$ the weight of $G$, where $w_i = w'_i$ if $i \in V'$, otherwise $w_i = 0$. By definition, $w$ is a nonnegative weight of $G$. Consequently,
    \begin{align}
        \beta(G',w') = \beta(G,w) = \alpha(G,w) = \alpha(G',w'),
    \end{align}
    which implies that $G'$ is also $\hbar$-perfect.
\end{proof}

\begin{prop}[Lexicographic product] \label{thm:product2}
  For two given $\hbar$-perfect graphs $G_1$ and $G_2$, denote by $G = G_1[G_2]$ the lexicographic product~\cite{sabidussi1961} of $G_1$ and $G_2$, then $G$ is also $\hbar$-perfect.
\end{prop}
\begin{proof}[Proof of Property \ref{thm:product2}]
  Let $\{A_{ij}\}$ be any realization of $G$, where $A_{ij}$ represents the vertex in $G$ which corresponds to $i$ in $G_1$ and $j$ in $G_2$. 
  For any nonnegative weight $w = (w_{ij})_{ij}$, let $\bar{A}_i = \sum_j \sqrt{w_{ij}}x_{ij}A_{ij}/\lambda_i$, where $(x_{ij})_{j}$ is the normalized vector of $\left(\sqrt{w_{ij}} \mean{A_{ij}}\right)_{j}$, and $\lambda_i$ is the maximal eigenvalue of $\sum_j \sqrt{w_{ij}}x_{ij}A_{ij}$.

  By definition,
  \begin{align}\label{eq:product_upper_bound}
    \beta(G,w) &= \max_{\rho} \sum_{i} \sum_j w_{ij}\mean{A_{ij}}_{\rho}^2\nonumber\\
             &= \max_{\rho}\sum_i \mean{\sum_j \sqrt{w_{ij}}  x_{ij} A_{ij}}_{\rho}^2\nonumber\\
             &= \max_{\rho}\sum_i \lambda_i^2 \mean{\bar{A}_i}_{\rho}^2\nonumber\\
             &\le \max_{\rho}\sum_i \beta(G_2, (w_{ij})_j) \mean{\bar{A}_i}_{\rho}^2\nonumber\\
             &\le \beta(G_1, (\beta(G_2, (w_{ij})_j))_i )\nonumber\\
             &= \alpha(G_1, (\alpha(G_2, (w_{ij})_j))_i )\nonumber\\
             &= \alpha(G,w),
  \end{align}
  where the first inequality in the forth line is from the following fact: 
  \begin{equation}
    \lambda_i^2 = \max_{\rho}\mean{\sum_j \sqrt{w_{ij}} x_{ij} A_{ij}}_{\rho}^2 \le \beta(G_2, (w_{ij})_j).
  \end{equation}
  The second inequality in the fifth line relies on the fact~\cite{xu2023bounding} that $\{\bar{A}_i\} = \{ U (B_i\otimes D_i) U^\dagger\}$ for some unitary $U$, where $B_i$ is the shortest realization of $G_1$ and $D_i$ is diagonal with the absolute values of diagonal elements no more than $1$.
  In the second last equality, we have made use of the fact that both $G_1$ and $G_2$ are $\hbar$-perfect.
  However, $\beta(G,w)$ is always not less than $\alpha(G,w)$. Hence, $\beta(G,w) = \alpha(G,w)$.
\end{proof}

\begin{prop}[Splitting of vertices]
\label{thm:splitting2}
For a $\hbar$-perfect graph $G$, if the graph $G'$ is obtained from $G$ by splitting one vertex, then $G'$ is also $\hbar$-perfect.
\end{prop}
\begin{proof} {Proof of Property \ref{thm:splitting2}.}
Notice that $G'$ is an induced subgraph of the lexicographic product of $G$ and the complete graph $K_2$ with two vertices. Hence, the $\hbar$-perfectness of $G'$ follows from Properties~\ref{thm:subgraph2} and \ref{thm:product2}.
\end{proof}

\begin{prop}[Copying of vertices]
\label{thm:copy2}
For a $\hbar$-perfect graph $G$, if the graph $G'$ is obtained from $G$ by copying one vertex, then $G'$ is also $\hbar$-perfect.
\end{prop}

\begin{proof} {Proof of Property \ref{thm:copy2}.}
    Without loss of generality, we label the vertex which has been copied and its copy with $n$ and $n+1$, respectively, where $n$ is the size of $G$. Furthermore, we take the realization $\{A_i\}_{i=1}^n \cup \{A_n\}$ of $G'$, where $\{A_i\}_{i=1}^n$ is one of $G$. For any nonnegative weight vector $w$ for $G'$, we have
    \begin{equation}
        \max_{\rho} \sum_{i=1}^{n+1} w_i \mean{A_i}^2_{\rho} = \max_{\rho} \sum_{i=1}^n \tilde{w}_i \mean{A_i}^2_{\rho},
    \end{equation}
    where $\tilde{w}_i = w_i$ for $i=1, \ldots, n-1$ and $\tilde{w}_n = w_n + w_{n+1}$.
    This implies that
    \begin{align}
        \beta(G,w)  = \beta(G,\tilde{w}) = \alpha(G,\tilde{w}) = \alpha(G,w).
    \end{align}
\end{proof}

\begin{theorem}[Perfect and $h$-Perfect graphs]\label{thm:perfect2}
    Any perfect graph and any $h$-perfect graph is $\hbar$-perfect. That is,
    \begin{align}
      \text{perfect }  \Rightarrow \text{$h$-perfect }  \Rightarrow \text{$\hbar$-perfect }.
    \end{align}
    In addition, odd cycles are all $h$-perfect graphs. 
\end{theorem}
\begin{proof}
As stated in the main text, for a given graph $G$, if all the facets of $\stab(G)$ can be grouped into the following three classes:
    \begin{enumerate}
        \item $x_i \ge 0$, for $i$ in the vertex set of $G$;
        \item $\sum_{i\in K} x_i \le 1$, for $K$ a clique in $G$;
        \item $\sum_{i\in C} x_i \le a$, for $C$ an odd-cycle in $G$ with $2a+1$ vertices with $a$ being an integer,
     \end{enumerate}
    then it is $h$-perfect~\cite{chvatal1975certain,fonlupt1982transformations}. 
If and only if the third class vanishes as well, the graph is perfect~\cite{chvatal1975certain}. Thus, it is clear that perfect graphs are always $h$-perfect.

    For a given $h$-perfect graph $G$, to prove its $\hbar$-perfectness is equivalent to prove that any point $x$ in $\be(G)$ satisfies all the constraints corresponding to facets of $\stab(G)$. 
    By definition, $\beta(G,w) \ge 0$ for any nonnegative weight vector $w$, hence $x_i \ge 0$ holds always.
    Since $\beta(K) = 1$ for any $K$ a clique in $G$, we know that $\sum_{i\in K} x_i \le 1$. Similarly, $\beta(C) = a$ for any odd-cycle $C$ in $G$ with $2a+1$ vertices~\cite{xu2023bounding}, which implies that $\sum_i x_i \le a$.
\end{proof}

\subsection{$\hbar$-perfectness of $G_7$ and $G_{15}$}
\label{ssec:hper2}
Among all graphs on $7$ vertices there are two whose $\hbar$-perfectness cannot be determined using the properties in the previous section. One of them is $\bar{C}_7$, which is known to be $\hbar$-imperfect. Another one is shown in Fig.~\ref{fig:g7}, which is named as $G_7$ here and we prove that it is $\hbar$-perfect. Notice that all the graphs with no more than $6$ vertices are $\hbar$-perfect, $G_7$ has only one facet whose norm vector contains no zero, and this norm vector is full of $1$. To prove the $\hbar$-perfectness of $G_7$, it is sufficient to prove $\beta(G_7)=\alpha(G_7)=2$.
\begin{proof}[Proof of $\beta(G_7)=2$]
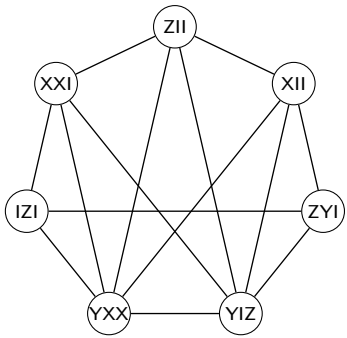
\begin{figure}[ht]
    \centering
    \begin{tikzpicture}[scale=0.8,every node/.style={draw, circle,minimum size=16, inner sep=-2, font=\sffamily\scriptsize}]
  \foreach \i/\lab in {0/ZII,1/XII,2/ZYI,3/YIZ,4/YXX,5/IZI,6/XXI} {
    \node (v\i) at ({90 - \i*360/7}:2.5cm) {\lab};
  }
  \foreach \a/\b in {0/1,1/2,2/3,3/4,4/5,5/6,6/0,0/3,0/4,1/3,1/4,6/3,6/4,2/5} {
    \draw[line width=0.5] (v\a) -- (v\b);
  }
\end{tikzpicture}
    \caption{ \label{fig:ZweiGraphenUndEineTabelle2}The graph $G_7$.}
    \label{fig:g7}
\end{figure}

A standard realization $(S_i)_{i=1}^7$ of $G_7$ in the order as in Fig.~\ref{fig:ZweiGraphenUndEineTabelle2} is
\begin{align}
    Z_1, X_1X_2, Z_2, Y_1X_2X_3, Y_1Z_3, Z_1Y_2, X_1.
\end{align}
For any real unit vector $\vec{w} = (w_i)_{i=1}^7$, we can find out the maximal singular value $\delta(\vec{w})$ of $\sum_i w_i S_i$ analytically as a function of $\vec{w}$. More explicitly,
\begin{align}
    \Delta(\vec{w}) \coloneqq\frac{[\delta(\vec{w})^2-1]^2}{4} = &w_1^2 (w_3^2+w_6^2)+w_3^2 (w_5^2+w_7^2) +w_6^2 (w_2^2+w_4^2)+w_2^2 w_7^2 +2 |w_3 w_6| \sqrt{w_5^2 (w_2^2+w_4^2)+w_4^2 w_7^2}.
\end{align}

Define 
\begin{equation}
    x_0 = \omega_1^2, x_1 = \omega_3^2, x_2 = \omega_6^2, y_1 = \omega_5^2+ \omega_7^2, y_2 = \omega_2^2+ \omega_4^2, z = \sqrt{y_1 y_2-\omega^2 \omega^4}.
\end{equation}
Then we have
\begin{equation}\label{eq:complicated}
   \Delta(\vec{w}) =  x_0 x_1+ x_0 x_2 + x_1 y_1 + x_2 y_2 + 2\sqrt{x_1 x_2}z+y_1y_2-z^2,
\end{equation}
where $x_0 + x_1 + x_2 + y_1 + y_2 = 1$ and $0 \leq z \leq \sqrt{y_1 y_2}$.

Since the right-hand side of Eq.~\eqref{eq:complicated} is a quadratic form in $z$, its maximum is achieved when $z = \min \{\sqrt{y_1 y_2}, \sqrt{x_1x_2} \}$.

We consider the first case that  $y_1 y_2 \leq x_1 x_2$. In this case, by setting $z = \sqrt{y_1y_2}$, we have
\begin{align}
     \Delta(\vec{w}) &= x_0 x_1+ x_0 x_2 + x_1 y_1 + x_2 y_2 + 2\sqrt{x_1 x_2 y_1 y_2}\nonumber\\
     &\le x_0 x_1+ x_0 x_2 + x_1 y_1 + x_2 y_2 + x_1 y_2 + x_2 y_1\nonumber\\
     &= (x_0 + y_1 + y_2)(x_1 + x_2)\nonumber\\
     &\le [(x_0 + y_1 + y_2) + (x_1 + x_2)]^2/4\nonumber\\
     &=1/4.
\end{align}
In the second case that  $y_1 y_2 \geq x_1 x_2$, by setting $z = \sqrt{x_1x_2}$ we have
\begin{align}
     \Delta(\vec{w}) &= x_0 x_1+ x_0 x_2 + x_1 y_1 + x_2 y_2 + x_1 x_2 + y_1 y_2.
\end{align}
For any given $x_1, x_2$ and $y_1$, the value of $x_0 + y_1$ is also fixed and $\Delta(\vec{w})$ is a linear function of $x_0$ and $y_1$. Hence, the maximum of $\Delta(\vec{w})$ is achieved either when $x_0 = 0$ or $y_1=x_1x_2/y_2$, where the second situation is covered by the previous case. For $x_0 = 0$, 
\begin{align}
     \Delta(\vec{w}) &= x_1 y_1 + x_2 y_2 + x_1 x_2 + y_1 y_2 = (x_1 + y_2) (x_2 + y_1)  \le 1/4.
\end{align}
In summary, we have proven that $\Delta(\vec{w}) \le 1/4$ and consequently $\delta(\vec{w})^2 \le 2$, which implies that $\beta(G_7) \le 2$. Since $\beta(G_7) \ge \alpha(G_7) = 2$, we conclude that $\beta(G_7) = 2$.
\end{proof}
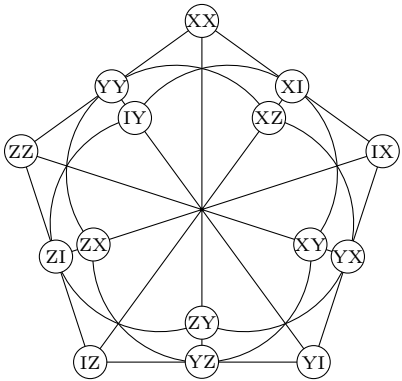
\begin{figure}[htpb]
	\centering
	\begin{tikzpicture}[scale=1.25]
\foreach \i in {1,...,5} {
    \pgfmathsetmacro{\angleA}{72*(\i-1) + 18}  
    \pgfmathsetmacro{\angleB}{72*(\i-1) - 18}  
    \coordinate (ps\i) at (\angleA:2); 
    \coordinate (ps2\i) at (\angleB:1.615);  
    \coordinate (ps3\i) at (\angleB:1.2);}
\draw (ps1) -- (ps2) -- (ps3) -- (ps4) -- (ps5) -- cycle;
\draw (ps1) -- (ps24) (ps2) -- (ps25) (ps3) -- (ps21) (ps4) -- (ps22) (ps5) -- (ps23);
\foreach \i in {1,...,5} {
    \pgfmathsetmacro{\rotangle}{72*(\i-1)}
    \begin{scope}[rotate=\rotangle]
        \draw (0, -0.47) + (180:1.15) arc (180:360:1.15);
    \end{scope}}
\def\labels{{"IX", "XX", "ZZ", "IZ", "YI", "YX", "XI", "YY", "ZI", "YZ", "XY", "XZ", "IY", "ZX", "ZY"}}
\foreach \i in {1,...,5} {
    \filldraw[fill=white, draw=black] (ps\i) circle (0.175);
    \pgfmathparse{\labels[\i-1]}
    \node[font=\scriptsize, inner sep=0.5pt] at (ps\i) {\pgfmathresult}; 
    \filldraw[fill=white, draw=black] (ps2\i) circle (0.175);
    \pgfmathparse{\labels[\i+4]}
    \node[font=\scriptsize, inner sep=0.5pt] at (ps2\i) {\pgfmathresult};
    \filldraw[fill=white, draw=black] (ps3\i) circle (0.175);
    \pgfmathparse{\labels[\i+9]}
    \node[font=\scriptsize, inner sep=0.5pt] at (ps3\i) {\pgfmathresult};}
\end{tikzpicture}
	\caption{The complementary graph of $G_{15}$ where all the vertices on the same line or half-circle are connected and the operators on the adjacent vertices commute with each other.}
	\label{fig:g15c}
\end{figure}
Denote by $G_{15}$ the frustration graph of all the $15$ Pauli strings with length $2$ except the identity, then $G_{15}$ is $\hbar$-perfect.
\begin{proof}[Proof of $\hbar$-perfectness of $G_{15}$]
	Due to the high symmetry in $G_{15}$ as illustrated in Fig.~\ref{fig:g15c}, all the facets of $G_{15}$ can be classified into 
	\begin{enumerate}
		\item $x_i \ge 0$,
		\item $\sum_{i\in C} x_i \le 1$, for all cliques $C$,
		\item $\sum_{i\in G'} x_i \le 2$, where $G'$ is any subgraph  isomorphic to the one without vertices labeled by $\{\id X, \id Z, XY, ZY, YY\}$. 
	\end{enumerate}
	For the first two classes, the corresponding beta number equals the independence number naturally. For the third class, the beta number is also $2$ as the independence number, which has been proven in Ref.~\cite{xu2023bounding}. Hence, $G_{15}$ is $\hbar$-perfect by definition.

\end{proof}

\subsection{Length of the realization of a frustration graph}\label{sec:proofs2}
\begin{theorem}\label{thm:minlen}
Let $G$ be a frustration graph with an adjacency matrix $A$.  Then the length of Pauli strings in any realization of $G$ is no less than ${\operatorname{rank}_{\mathbb{F}_2}(A)}/2$  with $\operatorname{rank}_{\mathbb{F}_2}(A)$ being the rank of $A$ over the field $\mathbb{F}_2$, which is achieved by the standard realization~\cite{xu2023bounding}.
\end{theorem}

We first notice that each realization $\{S_i\}_i$ with length $t$ of $G$ corresponds to a collection of complete tripatite graphs $T_1, \ldots, T_t$, in the way that $T_k$ is the frustration graph of $\{S_{i,k}\}_i$ with $S_{i,k}$ the $k$-th Pauli matrix of $S_i$.
In addition, $\{S_i, S_j\} = 0$ if and only if there is an odd number of $k$ such that $\{S_{i,k}, S_{j,k}\} = 0$. Consequently,
\begin{equation}\label{eq:decompA}
A = A(T_1) \oplus A(T_2) \oplus \cdots \oplus A(T_t),
\end{equation}
where $\oplus$ is modulo $2$, i.e., in the field $\mathbb{F}_2$.

Denote by $\mathrm{tp}_{\oplus}(G)$ the minimum number $t$ of complete tripartite graphs $T_1, \ldots, T_t$ satisfying Eq.~\eqref{eq:decompA}.
Clearly, $\mathrm{tp}_{\oplus}(G)$ is the shortest length of Pauli strings in any realization, which is thus the length of Pauli strings in the standard realization. Then we only need to prove that $\mathrm{tp}_{\oplus}(G)$ is given by half the $\mathbb{F}_2$-rank of the adjacency matrix $A$, for which we need the following facts.
\begin{lemma}[Ref.~\cite{Delsarte1975AlternatingBF} Lemma~10]\label{lm:standsymp}
	Any skew-symmetric matrix $A$ of order $n$ over the field $\mathbb{F}_q$ has even rank. If $\operatorname{rank}_{\mathbb{F}_q}(A) = r = 2k$, there exists a nonsingular matrix $L$ such that the transformed matrix $L^\top A L$ takes the canonical form $B$:
\begin{equation}
	L^\top A L = B = \bigoplus_{i=1}^k J, \quad \text{where } J = \begin{pmatrix} 0 & 1 \\ 1 & 0 \end{pmatrix}.
\end{equation}
\end{lemma}

\begin{lemma}\label{lm:triadj}
The matrix $M = u v^\top \oplus v u^\top$ over $\mathbb{F}_2$, for any $u, v \in \mathbb{F}_2^n$, is the adjacency matrix of a complete tripartite graph together with other isolated vertices.
\end{lemma}
\begin{proof}
Partitioning the vertex set $V$ based on the supports of $u$ and $v$ into $V_A = \{p: u_p=1, v_p=0\}$, $V_B = \{p: u_p=0, v_p=1\}$, $V_C = \{p: u_p=1, v_p=1\}$, and $V_D = \{p: u_p=0, v_p=0\}$, a case analysis shows that the entry $M_{pq} = u_p v_q \oplus v_p u_q$ is 1 if and only if $p$ and $q$ belong to different sets among $V_A, V_B, V_C$. This corresponds precisely to the adjacency matrix of the complete tripartite graph $K_{V_A, V_B, V_C}$ together with the isolated vertices in $V_D$.
\end{proof}

\begin{lemma}\label{lm:trirank}
For any complete tripartite graph $T$, it holds that ${\operatorname{rank}_{\mathbb{F}_2}(A(T)) \le 2}$.
\end{lemma}
\begin{proof}
Let $T$ be a complete tripartite graph with vertex sets $A, B, C $, whose adjacency matrix $A(T)$ over $\mathbb{F}_2$ is defined in the way that
 $(A(T))_{ij} = 1$ if $i$ and $j$ belong to different parts among $A, B, C $, and $(A(T))_{ij} = 0$ otherwise.

Let $n$ be the number of vertices. Define the indicator vectors $\mathbf{1}_A, \mathbf{1}_B, \mathbf{1}_C \in \mathbb{F}_2^n$ as the characteristic vectors of $A, B, C$, respectively.
Then, by definition,
\begin{equation}
A(T) = \mathbf{1}_A \mathbf{1}_B^\top \oplus \mathbf{1}_B \mathbf{1}_A^\top \oplus \mathbf{1}_A \mathbf{1}_C^\top \oplus \mathbf{1}_C \mathbf{1}_A^\top \oplus \mathbf{1}_B \mathbf{1}_C^\top \oplus \mathbf{1}_C \mathbf{1}_B^\top.
\end{equation}
Define $u = \mathbf{1}_A \oplus \mathbf{1}_C,  w = \mathbf{1}_B \oplus \mathbf{1}_C$, the direct calculation leads to
$A(T) = u w^\top \oplus w u^\top$. That is, $A(T)$ is the sum of two rank-1 matrices modulo $2$. Therefore,
$\operatorname{rank}_{\mathbb{F}_2}(A(T)) \le 2.$
\end{proof}

\begin{proof}
Let $r = \operatorname{rank}_{\mathbb{F}_2}(A) = 2k$ and $t = \mathrm{tp}_{\oplus}(G)$.
We first prove that $t\ge k$.
By definition, $A = \bigoplus_{i=1}^t A(T_i)$. Applying the rank subadditivity property and Lemma~\ref{lm:trirank}, we have
\begin{equation}
2k = \operatorname{rank}_{\mathbb{F}_2}(A) \le \sum_{i=1}^t \operatorname{rank}_{\mathbb{F}_2}(A(T_i)) \le 2t.
\end{equation}
That is, $k \le t$.

We turn to prove that $t\le k$. Since $A$ is an alternating matrix with $\operatorname{rank}_{\mathbb{F}_2}(A)=2k$, then Lemma~\ref{lm:standsymp} guarantees that $A$ is congruent to $B$ via a nonsingular matrix $L$, i.e., $A = L^{-\top} B L^{-1}$, where 
\begin{equation}
    B = \sum_{i=1}^{k} \left( e_{2i-1} e_{2i}^\top \oplus e_{2i} e_{2i-1}^\top \right),
\end{equation}
with $e_j$ the standard basis vectors.
Defining the vectors $u_i = L^{-\top} e_{2i-1}$ and $v_i = L^{-\top} e_{2i}$, we have
\begin{equation}
    A = \sum_{i=1}^{k} M_i, \quad \text{where } M_i = u_i v_i^\top \oplus v_i u_i^\top.
\end{equation}
By Lemma~\ref{lm:triadj}, each matrix $M_i$ is the adjacency matrix of a complete tripartite graph $T_i$. Therefore, $t\le k$ by the definition of $t$.
\end{proof}
Similarly, denote by $\mathrm{tp}(G)$ the minimum number $t$ of complete tripartite graphs $T_1, \ldots, T_t$ satisfying 
\begin{equation}
A = A(T_1) + A(T_2) + \cdots + A(T_t).
\end{equation}
And denote by $\mathrm{bp}(G)$ the minimum number $t$ of complete bipartite graphs $B_1, \ldots, B_t$ satisfying 
\begin{equation}
A = A(B_1) + A(B_2) + \cdots + A(B_t).
\end{equation}
Then, by definition, 
\begin{equation}
\mathrm{tp}_{\oplus}(G) \le \mathrm{tp}(G) \le \mathrm{bp}(G),
\end{equation}
where the last inequality is due to the fact that any complete bipartite graph can be viewed as a complete tripartite graph with one part being empty.
The number $\mathrm{bp}(G)$ is also known as the biclique partition number~\cite{graham1971addressing} of the graph $G$, which has applications in telephone switching circuitry~\cite{graham1971addressing} and its relation with other graph parameters has been established~\cite{lyu2023finding}. 

Here we have focused on arbitrary frustration graphs representing the anticommutation and commutation relations of Pauli strings. Recently, the concept of frustration graph with weighted edges was proposed for the strings of Weyl-Heisenberg matrices~\cite{Makuta2025FrustrationGF}. It is interesting to see that results similar to Theorem~\ref{thm:minlen} hold also for special kinds of frustration graphs with weighted edges.

\subsection{Convergence of generalized beta numbers}
\label{ssec:gbn}
\begin{proposition}
Given a graph $G$, as $k$ goes to infinity, $\beta(G,w,k)$ converges to the weighted independence number $\alpha(G,w)$.
More precisely, for any nonnegative weight vector $w$,
\begin{equation}
	\lim_{k\to \infty} \beta(G,w,k) = \alpha(G,w).
\end{equation}    
\end{proposition}
\begin{proof}
To begin, notice that $\beta(G,w,k)$ decreases as $k$ increases since $|\langle S_i\rangle|\le 1$, and $\beta(G,w,k) \ge 0$ by definition.
Thus, $\beta(G,w,k)$ converges, and denote by $\tilde{\alpha}(G,w)$ the limit. 

Besides, $\{S_i\}_{i\in I}$ is a set of commuting Pauli strings for an independent set $I$ of $G$, which implies that there is a state $\rho$ such that $|\langle S_i\rangle_\rho| = 1$ for $i\in I$ and the other expectation values . Consequently, $\beta(G,w,k)\ge \max_{I} \sum_{i\in I} w_i = \alpha(G,w)$ for any $k$, by the definition of the weighted independence number. Due to the arbitrariness of $k$, we have $\tilde{\alpha}(G,w) \ge \alpha(G,w)$.

If  $\tilde{\alpha}(G,w)$ is strictly larger than $\alpha(G,w)$, there exists $\epsilon > 0$ such that $\tilde{\alpha}(G,w) \ge \alpha(G,w) + \epsilon$.
For each $k$, denote by $\Lambda_k$ the set of states $\rho$ such that $\sum_i w_i |\langle S_i\rangle^k_\rho| \ge \alpha(G,w)+\epsilon$. By assumption, $\Lambda_k$ is non-empty, closed, and compact. As $|\langle S_i\rangle^k_\rho|$ decreases with $k$, $\Lambda_{k+1} \subseteq \Lambda_k$.
From Cantor's intersection theorem, we know that $\Lambda \coloneqq \cap_{k=1}^\infty \Lambda_k$ is not empty. 

Denote by $\rho$ a state in $\Lambda$, then $\sum_i w_i |\langle S_i\rangle^k_\rho| \ge \alpha(G,w)+\epsilon$ for each $k$ by assumption. Then we have 
\begin{equation}
\begin{aligned}
	\alpha(G,w)+\epsilon & \le
	\lim_{k\to \infty} \sum\nolimits_i w_i |\langle S_i\rangle^k_\rho|\\
			     &= \sum\nolimits_i w_i \big(\lim_{k\to\infty} |\langle S_i\rangle^k_\rho|\big) \\
			     &=\sum\nolimits_{i\in \bar{I}} w_i,
\end{aligned}
\end{equation}
where $\bar{I}$ is the set of all $i$ such that $|\langle S_i\rangle_\rho|=1$. For any $i\not\in \bar{I}$, we have $\lim_{k\to\infty} |\langle S_i\rangle^k_\rho| = 0$ due to $|\langle S_i\rangle_\rho| < 1$.
We claim that $\bar{I}$ must be an independence set, in which case $\sum_{i\in \bar{I}} w_i \le \alpha(G,w)$ results in a contradiction. If $\bar{I}$ is not an independence set, it contains two adjacent vertices $i_1$ and $i_2$. Then $|\langle S_{i_1}\rangle_\rho|$ and  $|\langle S_{i_2}\rangle_\rho|$ cannot be $1$ simultaneously, since the sum of squares of them is no more than $1$, which contradicts the definition of $\bar{I}$. This finishes the proof.
\end{proof}
\section{Numerical methods}\label{sec:numtools}
\subsection{A new hierarchy of SDP relaxations}\label{ssec:spo}
The calculation of the weighted beta number $\beta(G,w)$ can be formulated as a nonlinear optimization problem over the expectations $\langle\cdot\rangle_\rho$:
\begin{equation}\label{eq:spop}
\beta(G,w) \coloneqq \begin{cases}\sup_\rho &\sum_{i=1}^n w_i\langle S_i \rangle_\rho^2\\
\textrm{s.t.}&S_i^2=1,\quad\text{ for }i=1,\ldots,n,\\
&S_iS_j=-S_jS_i,\quad\text{ for }i\sim_{G} j,\\
&S_iS_j=S_jS_i,\quad\text{ for }i\nsim_{G} j,
\end{cases}
\end{equation}
where $i\sim_{G} j$ (reps. $i\nsim_{G} j$) means that the vertices $i$ and $j$ are connected (reps. not connected) in graph $G$.
The problem \eqref{eq:spop} is an instance of state polynomial optimization problems \cite{klep2024state}.
A complete hierarchy of moment-sum-of-Hermitian-squares relaxations was established in \cite{klep2024state} for solving state polynomial optimization problems. We begin by outlining its main idea. Consider letters $\{x_1,\ldots,x_n\}$ and words $u=x_{i_1}\cdots x_{i_t}$ built from these letters. The involution of a word is defined as $(x_{i_1}\cdots x_{i_t})^*=x_{i_t}\cdots x_{i_1}$.
We further consider pseudo-expectations of words $\langle u\rangle=\langle u^*\rangle\in\mathbb{R}$ which behave as scalars. A state monomial in words and pseudo-expectations is of the form $u=u_0\langle u_1\rangle\cdots\langle u_s\rangle$ where $u_0,\ldots,u_s$ are words in $\{x_1,\ldots,x_n\}$. 
When specifying to \eqref{eq:spop}, the constraint $S_i^2=1$ gives rise to $x_i^2=1$ for $i=1,\ldots,n$, and the constraint $S_iS_j=-S_jS_i$ (resp. $S_iS_j=S_jS_i$) for $i\sim_{G} j$ (resp. $i\nsim_{G} j$) gives rise to $x_ix_j=-x_jx_i$ (resp. $x_ix_j=x_jx_i$) for $i\sim_{G} j$ (resp. $i\nsim_{G} j$). Such relations allow us to reduce any word to the normal form $cx_{i_1}\cdots x_{i_t}$ with $i_1<\cdots<i_t$ and $c\in\{-1,1\}$. 

For a state monomial $u=u_0\langle u_1\rangle\cdots\langle u_s\rangle$, its degree is the sum of the lengths of the words $u_i$. Now, given a positive integer $r$, the $r$-th moment matrix $M_r$ is indexed by state monomials of degree at most $r$ and its entry $[M_r]_{u,v}\coloneqq\langle u^*v\rangle$ where $u^*v$ is assumed to be reduced to the normal form. Then, the $r$-th order moment relaxation for \eqref{eq:spop} is given by the following SDP:
\begin{equation}\label{mom}
\lambda_r(G,w)\coloneqq\begin{cases}
\sup &\sum_{i=1}^n w_i\langle x_i\rangle^2\\
\rm{s.t.}&M_r\succeq0,\\
&[M_r]_{1,1}=1,\\
&[M_r]_{u,v}=[M_r]_{a,b}, \text{ if }\langle u^*v\rangle=\langle a^*b\rangle,
\end{cases}
\end{equation}
where $1$ denotes the empty word. It was shown in \cite{klep2024state} that the sequence of upper bounds $(\lambda_r(G,w))_{r\ge1}$ converges monotonically to the optimum of the state polynomial optimization problem \eqref{eq:spop}, i.e., $$\lim_{r\to\infty}\lambda_r(G,w)=\beta(G,w).$$
However, the size of the moment relaxation grows rapidly with the size of the original problem as well as the relaxation order $r$, and so it soon becomes intractable. To circumvent this difficulty, a simplified Lov\'asz-type hierarchy was proposed in \cite{moran2024Uncertainty}, where one considers the reduced moment matrix which is indexed by state monomials of the forms
\begin{gather}
1,\,x_{i}\langle x_{i}\rangle,\,x_{i}x_{j}\langle x_{i}\rangle\langle x_{j}\rangle,\nonumber\\
\cdots\label{lovasz}\\
x_{1}\cdots x_{n}\langle x_{1}\rangle\cdots\langle x_{n}\rangle\nonumber.
\end{gather}
As observed in \cite{moran2024Uncertainty}, the simplified Lov\'asz-type hierarchy may not converge to $\beta(G,w)$ for graphs with $7$ vertices.

Here, to improve the bounds while maintaining scalability, we propose to extend the monomial basis in \eqref{lovasz} by including two-body pseudo-expectations $\langle x_{i}x_{j}\rangle$.
Specifically, we form the monomial basis by selecting state monomials of the forms
\begin{gather}
1,\,x_{i}\langle x_{i}\rangle,x_{i}x_{j}\langle x_{i}\rangle\langle x_{j}\rangle,x_{i}\langle x_{j}\rangle\langle x_{i}x_{j}\rangle,\nonumber\\
x_{i}x_{j}x_{k}\langle x_{i}\rangle\langle x_{j}\rangle\langle x_{k}\rangle,x_{i}x_{j}x_{k}\langle x_{i}x_{j}\rangle\langle x_{k}\rangle,x_{i}x_{k}\langle x_{i}x_{j}\rangle\langle x_{j}\rangle\langle x_{k}\rangle,x_{k}\langle x_{i}x_{j}\rangle\langle x_{i}\rangle\langle x_{j}\rangle\langle x_{k}\rangle,\label{new}\\
\cdots\nonumber
\end{gather}
Note that for $i\sim_{G} j$, it holds that $\langle x_{i}x_{j}\rangle=-\langle x_{j}x_{i}\rangle$ which implies $\langle x_{i}x_{j}\rangle=0$. Thus, we can exclude $\{\langle x_{i}x_{j}\rangle\}_{i\sim_{G} j}$ when constructing the monomial basis \eqref{new}. Moreover, if $\langle u^*v\rangle=-\langle v^*u\rangle$, then we set $\langle u^*v\rangle=\langle v^*u\rangle=0$ in the moment matrix.

Among graphs with $8$ and $9$ vertices, there are $100$ and $2963$ graphs whose $\hbar$-(im)perfectness requires numerical verification, respectively. 
It turns out that by virtue of the new hierarchy we can verify $\hbar$-perfectness for all those graphs with $8$ vertices and verify $\hbar$-perfectness for all but $78$ those graphs with $9$ vertices (up to a $<1\times10^{-5}$ gap between lower and upper bounds). In Table~\ref{tab:g9}, we show the top $12$ graphs that yield the largest gaps. The calculations were performed on a desktop with 128G RAM and {\tt Mosek 10.2} was employed as the underlying SDP solver. Our code for reproducing the results is available at \url{https://github.com/wangjie212/BetaNumber}.

\begin{table}[tbp]
\begin{center}
\def \scale {1.2}
\def \zoom {0.4}
\def \nth {9}
\def \an {360/\nth}
\def \bline {0.1cm}
\def \acolor {black}
\begin{tabular}{c@{\hskip 15pt}c@{\hskip 15pt}c@{\hskip 15pt}c@{\hskip15pt}c@{\hskip15pt}c}
 &\\
  &
 \begin{tikzpicture}[scale=\scale, baseline=\bline]
    \foreach \i in {1,2,...,\nth} {
\coordinate (A\i) at ({-cos(\i*\an)},{sin(\i*\an)}) {};
\node[party, fill=white, scale=\zoom] (A\i) at (A\i) {};
}
\draw (A1)--(A3) (A1)--(A5) (A1)--(A6) (A1)--(A7) (A1)--(A8) (A1)--(A9) (A2)--(A4) (A2)--(A5) (A2)--(A6) (A2)--(A7) (A2)--(A8) (A2)--(A9) (A3)--(A5) (A3)--(A7) (A3)--(A9) (A4)--(A6) (A4)--(A7) (A4)--(A8) (A4)--(A9) (A5)--(A7) (A5)--(A8) (A6)--(A8) (A6)--(A9) (A7)--(A9);
 \end{tikzpicture}
 &
\begin{tikzpicture}[scale=\scale, baseline=\bline]
\foreach \i in {1,2,...,8} {
\coordinate (A\i) at ({-cos(\i*\an)},{sin(\i*\an)}) {};
\node[party, fill=white, scale=\zoom] (A\i) at (A\i) {};
}
\coordinate (A9) at ({-cos(9*\an)},{sin(9*\an)}) {};
\node[draw, party, fill=black, minimum width=1mm, scale=\zoom] (A9) at (A9) {};
\draw (A1)--(A3) (A1)--(A5) (A1)--(A6) (A1)--(A7) (A1)--(A9) (A2)--(A4) (A2)--(A5) (A2)--(A7) (A2)--(A8) (A3)--(A5) (A3)--(A6) (A3)--(A8) (A3)--(A9) (A4)--(A6) (A4)--(A7) (A4)--(A8) (A4)--(A9) (A5)--(A7) (A5)--(A9) (A6)--(A8) (A6)--(A9) (A7)--(A9) (A8)--(A9);
 \end{tikzpicture}
 &
 \begin{tikzpicture}[scale=\scale, baseline=\bline]
    \foreach \i in {1,2,...,8} {
\coordinate (A\i) at ({-cos(\i*\an)},{sin(\i*\an)}) {};
\node[party, fill=white, scale=\zoom] (A\i) at (A\i) {};
}
\coordinate (A9) at ({-cos(9*\an)},{sin(9*\an)}) {};
\node[draw, circle, fill=black, minimum width=1mm, scale=\zoom] (A9) at (A9) {};
\draw (A1)--(A3) (A1)--(A5) (A1)--(A6) (A1)--(A7) (A1)--(A9) (A2)--(A4) (A2)--(A5) (A2)--(A7) (A2)--(A8) (A2)--(A9) (A3)--(A5) (A3)--(A6) (A3)--(A8) (A3)--(A9) (A4)--(A6) (A4)--(A7) (A4)--(A8) (A4)--(A9) (A5)--(A7) (A5)--(A8) (A5)--(A9) (A6)--(A8) (A6)--(A9) (A7)--(A9) (A8)--(A9);
 \end{tikzpicture}
&
 \begin{tikzpicture}[scale=\scale, baseline=\bline]
    \foreach \i in {1,2,...,8} {
\coordinate (A\i) at ({-cos(\i*\an)},{sin(\i*\an)}) {};
\node[party, fill=white, scale=\zoom] (A\i) at (A\i) {};
}
\coordinate (A9) at ({-cos(9*\an)},{sin(9*\an)}) {};
\node[draw, circle, fill=black, minimum width=1mm, scale=\zoom] (A9) at (A9) {};
\draw (A1)--(A3) (A1)--(A5) (A1)--(A6) (A1)--(A7) (A1)--(A9) (A2)--(A4) (A2)--(A5) (A2)--(A7) (A2)--(A8) (A2)--(A9) (A3)--(A5) (A3)--(A6) (A3)--(A8) (A3)--(A9) (A4)--(A6) (A4)--(A7) (A4)--(A8) (A4)--(A9) (A5)--(A7) (A5)--(A8) (A5)--(A9) (A6)--(A8) (A6)--(A9) (A7)--(A9) (A8)--(A9);
 \end{tikzpicture}
\\
 & \\
 $\omega$&$(1, 1, 1, 1, 1, 1, 1, 1, 1)$&$(1, 1, 1, 1, 1, 1, 1, 1, 2)$&$(1, 1, 1, 1, 1, 1, 1, 1, 2)$&$(1, 1, 1, 1, 1, 1, 1, 1, 2)$\\
 $\gamma_{\mathrm{ub}}$ & 2.00048 & 2.00227 & 2.00175&2.00118\\
 see-saw& 2.00000 & 2.00000 & 2.00000&2.00000\\
 $\alpha$ & 2 & 2 & 2&2\\
 &\\
  &
 \begin{tikzpicture}[scale=\scale, baseline=\bline]
    \foreach \i in {1,2,...,8} {
\coordinate (A\i) at ({-cos(\i*\an)},{sin(\i*\an)}) {};
\node[party, fill=white, scale=\zoom] (A\i) at (A\i) {};
}
\coordinate (A9) at ({-cos(9*\an)},{sin(9*\an)}) {};
\node[draw, circle, fill=black, minimum width=1mm, scale=\zoom] (A9) at (A9) {};
\draw (A1)--(A3) (A1)--(A5) (A1)--(A6) (A1)--(A7) (A1)--(A8) (A1)--(A9) (A2)--(A4) (A2)--(A5) (A2)--(A7) (A2)--(A8) (A2)--(A9) (A3)--(A5) (A3)--(A6) (A3)--(A8) (A3)--(A9) (A4)--(A6) (A4)--(A7) (A4)--(A8) (A4)--(A9) (A5)--(A7) (A5)--(A8) (A5)--(A9) (A6)--(A7) (A6)--(A8) (A6)--(A9) (A7)--(A9) (A8)--(A9);
 \end{tikzpicture}
 &
\begin{tikzpicture}[scale=\scale, baseline=\bline]
    \foreach \i in {1,2,...,\nth} {
\coordinate (A\i) at ({-cos(\i*\an)},{sin(\i*\an)}) {};
\node[party, fill=white, scale=\zoom] (A\i) at (A\i) {};
}
\draw (A1)--(A3) (A1)--(A5) (A1)--(A6) (A1)--(A7) (A1)--(A8) (A1)--(A9) (A2)--(A4) (A2)--(A5) (A2)--(A6) (A2)--(A7) (A2)--(A8) (A2)--(A9) (A3)--(A5) (A3)--(A7) (A3)--(A8) (A3)--(A9) (A4)--(A6) (A4)--(A7) (A4)--(A8) (A4)--(A9) (A5)--(A6) (A5)--(A7) (A6)--(A8) (A7)--(A9) (A8)--(A9);
 \end{tikzpicture}
 &
\begin{tikzpicture}[scale=\scale, baseline=\bline]
    \foreach \i in {1,2,...,\nth} {
\coordinate (A\i) at ({-cos(\i*\an)},{sin(\i*\an)}) {};
\node[party, fill=white, scale=\zoom] (A\i) at (A\i) {};
}
\draw (A1)--(A3) (A1)--(A5) (A1)--(A6) (A1)--(A7) (A1)--(A8) (A1)--(A9) (A2)--(A4) (A2)--(A5) (A2)--(A7) (A2)--(A8) (A2)--(A9) (A3)--(A5) (A3)--(A6) (A3)--(A7) (A3)--(A9) (A4)--(A6) (A4)--(A7) (A4)--(A8) (A4)--(A9) (A5)--(A6) (A5)--(A8) (A5)--(A9) (A6)--(A8) (A7)--(A9) (A8)--(A9);
 \end{tikzpicture}
&
\begin{tikzpicture}[scale=\scale, baseline=\bline]
    \foreach \i in {1,2,...,8} {
\coordinate (A\i) at ({-cos(\i*\an)},{sin(\i*\an)}) {};
\node[party, fill=white, scale=\zoom] (A\i) at (A\i) {};
}
\coordinate (A9) at ({-cos(9*\an)},{sin(9*\an)}) {};
\node[draw, circle, fill=black, minimum width=1mm, scale=\zoom] (A9) at (A9) {};
\draw (A1)--(A3) (A1)--(A4) (A1)--(A6) (A1)--(A7) (A1)--(A9) (A2)--(A4) (A2)--(A5) (A2)--(A6) (A2)--(A7) (A2)--(A8) (A2)--(A9) (A3)--(A5) (A3)--(A6) (A3)--(A8) (A3)--(A9) (A4)--(A6) (A4)--(A7) (A4)--(A8) (A4)--(A9) (A5)--(A7) (A5)--(A8) (A5)--(A9) (A6)--(A8) (A6)--(A9) (A7)--(A9) (A8)--(A9);
 \end{tikzpicture}
 \\
 & \\
 $\omega$&$(1, 1, 1, 1, 1, 1, 1, 1, 2)$&$(1, 1, 1, 1, 1, 1, 1, 1, 1)$&$(1, 1, 1, 1, 1, 1, 1, 1, 1)$&$(1, 1, 1, 1, 1, 1, 1, 1, 2)$\\
 $\gamma_{\mathrm{ub}}$ & 2.00056 & 2.00125 & 2.00058&2.00227\\
 see-saw& 2.00000 & 2.00000 & 2.00000&2.00000\\
 $\alpha$ & 2 & 2 & 2&2\\
 &\\
  &
\begin{tikzpicture}[scale=\scale, baseline=\bline]
    \foreach \i in {1,2,...,\nth} {
\coordinate (A\i) at ({-cos(\i*\an)},{sin(\i*\an)}) {};
\node[party, fill=white, scale=\zoom] (A\i) at (A\i) {};
}
\draw (A1)--(A3) (A1)--(A4) (A1)--(A6) (A1)--(A7) (A1)--(A8) (A1)--(A9) (A2)--(A4) (A2)--(A5) (A2)--(A6) (A2)--(A7) (A2)--(A8) (A2)--(A9) (A3)--(A5) (A3)--(A6) (A3)--(A8) (A3)--(A9) (A4)--(A6) (A4)--(A7) (A4)--(A8) (A4)--(A9) (A5)--(A7) (A5)--(A8) (A5)--(A9) (A6)--(A8) (A7)--(A9) (A8)--(A9);
 \end{tikzpicture}
 &
\begin{tikzpicture}[scale=\scale, baseline=\bline]
    \foreach \i in {1,2,...,\nth} {
\coordinate (A\i) at ({-cos(\i*\an)},{sin(\i*\an)}) {};
\node[party, fill=white, scale=\zoom] (A\i) at (A\i) {};
}
\draw (A1)--(A3) (A1)--(A4) (A1)--(A6) (A1)--(A7) (A1)--(A8) (A1)--(A9) (A2)--(A4) (A2)--(A5) (A2)--(A6) (A2)--(A7) (A2)--(A8) (A2)--(A9) (A3)--(A5) (A3)--(A6) (A3)--(A7) (A3)--(A8) (A3)--(A9) (A4)--(A6) (A4)--(A7) (A4)--(A8) (A4)--(A9) (A5)--(A7) (A5)--(A8) (A5)--(A9) (A6)--(A8) (A6)--(A9) (A7)--(A9);
 \end{tikzpicture}
 &
\begin{tikzpicture}[scale=\scale, baseline=\bline]
    \foreach \i in {1,2,...,8} {
\coordinate (A\i) at ({-cos(\i*\an)},{sin(\i*\an)}) {};
\node[party, fill=white, scale=\zoom] (A\i) at (A\i) {};
}
\coordinate (A9) at ({-cos(9*\an)},{sin(9*\an)}) {};
\node[draw, circle, fill=black, minimum width=1mm, scale=\zoom] (A9) at (A9) {};
\draw (A1)--(A3) (A1)--(A4) (A1)--(A6) (A1)--(A7) (A1)--(A8) (A1)--(A9) (A2)--(A4) (A2)--(A5) (A2)--(A6) (A2)--(A7) (A2)--(A8) (A2)--(A9) (A3)--(A5) (A3)--(A6) (A3)--(A7) (A3)--(A8) (A3)--(A9) (A4)--(A6) (A4)--(A7) (A4)--(A8) (A4)--(A9) (A5)--(A7) (A5)--(A8) (A5)--(A9) (A6)--(A8) (A6)--(A9) (A7)--(A9) (A8)--(A9);
 \end{tikzpicture}
 &
 \begin{tikzpicture}[scale=\scale, baseline=\bline]
    \foreach \i in {1,2,...,\nth} {
\coordinate (A\i) at ({-cos(\i*\an)},{sin(\i*\an)}) {};
\node[party, fill=white, scale=\zoom] (A\i) at (A\i) {};
}
\draw (A1)--(A3) (A1)--(A4) (A1)--(A5) (A1)--(A7) (A1)--(A8) (A1)--(A9) (A2)--(A4) (A2)--(A5) (A2)--(A6) (A2)--(A7) (A2)--(A8) (A2)--(A9) (A3)--(A5) (A3)--(A6) (A3)--(A7) (A3)--(A8) (A3)--(A9) (A4)--(A6) (A4)--(A7) (A4)--(A8) (A4)--(A9) (A5)--(A7) (A5)--(A8) (A6)--(A8) (A6)--(A9) (A7)--(A9);
 \end{tikzpicture}
 \\
 & \\
 $\omega$&$(1, 1, 1, 1, 1, 1, 1, 1, 1)$&$(1, 1, 1, 1, 1, 1, 1, 1, 1)$&$(1, 1, 1, 1, 1, 1, 1, 1, 2)$&$(1, 1, 1, 1, 1, 1, 1, 1, 1)$\\
 $\gamma_{\mathrm{ub}}$ & 2.00073 & 2.00062 & 2.00083&2.00094\\
 see-saw& 2.00000 & 2.00000 & 2.00000&2.00000\\
 $\alpha$ & 2 & 2 & 2&2\\
\end{tabular}
\end{center}
\caption{The top $12$ graphs with the largest gaps, where black vertices are of weights $2$. Here, $\gamma_{\mathrm{ub}}$ denotes the upper bound given by SDP relaxations and ``see-saw'' denotes the lower bound given by the see-saw method.}\label{tab:g9}
\end{table}

\subsection{Retrieving an approximately optimal state}\label{ssec:seesaw_state}
After solving the moment relaxation for \eqref{eq:spop}, it is possible to retrieve an approximately optimal state to \eqref{eq:spop} from the SDP solution. Suppose that the length of $S_1,\ldots,S_n$ is $m$ and $\{\langle x_{i}\rangle^2\}_{i=1}^n$ are extracted from the moment matrix. Then, for any binary vector $(a_i)_{i=1}^n\in\{-1,1\}^n$, we consider the following nonlinear SDP:
\begin{equation}\label{extract}
\begin{cases}
\min_{\rho\in\mathbb{S}^{2^m}} &\sum_{i=1}^n \left(\tr(\rho S_i)-a_i\sqrt{\langle x_{i}\rangle^2}\right)^2\\
\rm{s.t.}&\rho\succeq0,\\
&\tr(\rho)=1.\\
\end{cases}
\end{equation}
The objective of \eqref{extract} is a convex quadratic function in $\rho$ and could be solved by SDP solvers, e.g., {\tt COSMO} \cite{Garstka_2021}. In practice, we need to traverse all binary vectors $(a_i)_{i=1}^n$ and solve the corresponding SDPs \eqref{extract} until the optimal value of \eqref{extract} is approximately equal to zero. By doing so, we could retrieve an approximately optimal state to \eqref{eq:spop} as long as the relaxation is tight enough. Such an approximately optimal state can be then used as the starting point for the see-saw method.

We take the graph $G_9$ in Fig.~\ref{fig:HCRfeyz} for the illustration, which can be realized as a frustration graph of Pauli strings in $4$ qubits in order, i.e.,
\begin{equation}\label{eq:repg9}
    X\id\id\id, \id X\id\id, \id\id X\id, Z\id\id\id, \id Z\id\id, ZZZ\id, YZYX, YYXX, YXZZ.
\end{equation}
For the weight vector $w = (1,1,1,1,1,1,1,2,2)$, the weighted independence number $\alpha(G_9,w) = 3$ and the weighted beta number $\beta(G_9,w) = 3.044815$. By running the see-saw method $10$ times with the realization in Eq.~\eqref{eq:repg9} and random initial states, we only obtain the optimal lower bound for $\beta(G,w)$ which is trivially $3.0 = \alpha(G,w)$. Meanwhile, the SDP relaxation in the monomial basis $\{1,\,x_{i}\langle x_{i}\rangle\}_i$ results in an upper bound $3.236068$, for which the state $\rho = |s\rangle\langle s|$ can be extracted by the SDP in Eq.~\eqref{extract} with 
\begin{align}
\langle s| = (
& 0.065955i, -0.065955, -0.174514i, 0.174514, 0.204966i, -0.204966, 0.028492i, -0.028492,\nonumber\\ &  0.130338, -0.130338i, -0.166030, 0.166030i, -0.236494, 0.236494i, 0.567352, -0.567352i)
\end{align}
With such a state $\rho$ as an initial state, the initial value is $2.992458 < \alpha(G_9,w)$. However, the see-saw method with $40$ iterations already leads to the value $3.044815 = \beta(G_9,w)$. Moreover, the SDP relaxation in the monomial basis $\{1,\,x_{i}\langle x_{i}\rangle,x_{i}x_{j}\langle x_{i}\rangle\langle x_{j}\rangle,x_{i}\langle x_{j}\rangle\langle x_{i}x_{j}\rangle\}_{i,j}$ also results in the value $3.044815$.
\begin{figure}
\centering
\begin{tikzpicture}[scale=1, baseline=0]
\foreach \i in {2,3,...,8} {\pgfmathsetmacro{\angle}{450 - (\i - 1) * (360/9)}
    \node[party, fill=white, inner sep=2.5pt] (N\i) at (\angle:2cm) {};}
\foreach \i in {1,,9} {\pgfmathsetmacro{\angle}{450 - (\i - 1) * (360/9)}
    \node[party, fill=black, inner sep=2.5pt] (N\i) at (\angle:2cm) {};}
\draw[line width=0.5] (N1)--(N3) (N1)--(N4) (N1)--(N5) (N1)--(N7) (N1)--(N8) (N1)--(N9) (N2)--(N5) (N2)--(N6) (N2)--(N7) (N2)--(N8) (N2)--(N9) (N3)--(N6) (N3)--(N7) (N3)--(N8) (N4)--(N7) (N4)--(N9) (N5)--(N8) (N5)--(N9) (N6)--(N9) (N8)--(N9);
\end{tikzpicture}
\caption{The graph $G_{9}$ with $9$ vertices, where black vertices are of weights $2$.}
\label{fig:HCRfeyz}
\end{figure}
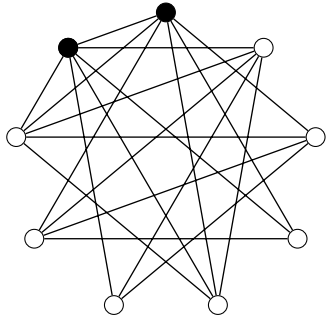

\subsection{Typicality of different types of graphs }\label{ssec:typical}
With the theoretical and numerical tools developed in the previous sections,
we investigate all graphs with up to $9$ vertices for perfectness, $h$-perfectness and $\hbar$-perfectness. 
We also include claw-free graphs~\cite{faudree1997} here, which are often used in solvable models~\cite{chapman2020characterization,chapman2023unified}. In comparison, $\hbar$-perfect graphs correspond to the models where an efficient estimation is possible as discussed in the next section of applications.
The results are presented in Table~\ref{tab:kinds_ghs2} and displayed in Fig.~\ref{fig:percentage}. 

\begin{table}[ht]
    \centering
        \begin{tabular}{r|l|l|l|l|l|l|l}
    \hline\hline
        $n$ & $3$ & $4$ & $5$ & $6$ & $7$ & $8$ & $9$ \\ \hline
        connected & $2$ & $6$ & $21$ & $112$ & $\bf{853}$ & $11117$ & $261080$ \\ \hline
        $\hbar$-perfect & $2$ & $6$ & $21$ & $\bf{112}$ & $852$ & $11099$ & $[259583,259661]$  \\ \hline
        $h$-perfect & $2$ & $6$ & $\bf{21}$ & $109$ & $780$ & $8689$ & $146375$ \\ \hline
        perfect & $2$ & $\bf{6}$ & $20$ & $105$ & $724$ & $7805$ & $126777$  \\ \hline
        claw-free & $\bf{2}$ & $5$ & $14$ & $50$ & $191$ & $881$ & $4494$  \\ \hline\hline
    \end{tabular}
    \caption{The numbers of different classes of graphs, where $n$ is the number of vertices. For the case of $9$ vertices, there are $78$  undetermined graphs by current analytical and numerical tools. See Sec.~\ref{sec:numtools} for more details.}\label{tab:kinds_ghs2}
\end{table}

\begin{figure}[ht]
    \centering
    \includegraphics[width=0.45\textwidth]{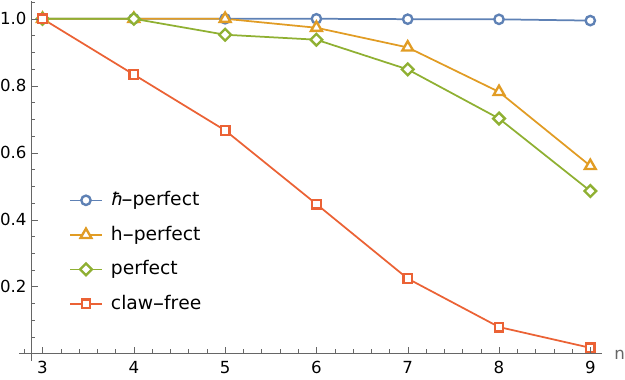}
    \caption{The percentages of claw-free graphs, perfect graphs, $h$-perfect graphs and $\hbar$-perfect graphs in all connected graphs for the fixed number $n$ of vertices.}
    \label{fig:percentage}
\end{figure}

The procedure for determining the number of $\hbar$-perfect graphs on $8$ and $9$ vertices, respectively, is as follows. Based on the fact that the graph $G_{15}$ introduced in Appendix~\ref{ssec:hper2} is $\hbar$-perfect, we can use Property~\ref{thm:subgraph2} to filter out all the $\hbar$-perfect graphs isomorphic to subgraphs of $G_{15}$.

For all the $h$-imperfect graphs with $8$ vertices, we first check whether $G$ has $\bar{C}_7$ as its induced subgraph. If this is the case, then $G$ is $\hbar$-imperfect. In total, there are $17$ graphs in this case. Otherwise, we claim that $G$ is $\hbar$-perfect in the following three cases:
\begin{enumerate}
    \item $G$ contains a pair of vertices $i$ and $j$ such that $j$ is either a copy or a split of $i$;
    \item $G$ is a fully connected join of the other two graphs;
    \item All the normal vectors of the facets of ${\rm STAB}(G)$ contain $0$.
\end{enumerate}
This claim holds for the following reasons.
For convenience, we call graphs in those three cases \textit{reducible}.
In the first case, the resulting graph by removing $j$ from $G$ is $\hbar$-perfect since it cannot be $\bar{C}_7$. Properties~\ref{thm:splitting2} and \ref{thm:copy2} ensure that $G$ is also $\hbar$-perfect.
In the second case, each component is $\hbar$-perfect since it cannot be $\bar{C}_7$. Property~\ref{thm:connected_union2} ensures that $G$ is also $\hbar$-perfect.
In the third case, all facet-induced subgraphs are $\hbar$-perfect since they cannot be $\bar{C}_7$, and hence all the corresponding weighted beta numbers are equal to the weighted independence numbers. Consequently, $G$ is $\hbar$-perfect.

There are only $102$ graphs beyond those $3$ classes, $2$ among which are induced subgraphs of $G_{15}$ and hence $\hbar$-perfect. Then there are $100$ undetermined graphs, for which we apply the numerical tools and find them to be $\hbar$-perfect (up to a $<1\times10^{-5}$ gap between the lower and upper bounds), except one graph as in Fig.~\ref{fig:otherghs}(a) which is named as $G_8$.
Hence, there are $18$ $\hbar$-imperfect graphs in total for the case of $8$ vertices.

\begin{figure}[ht]
\includegraphics[width=0.9\linewidth]{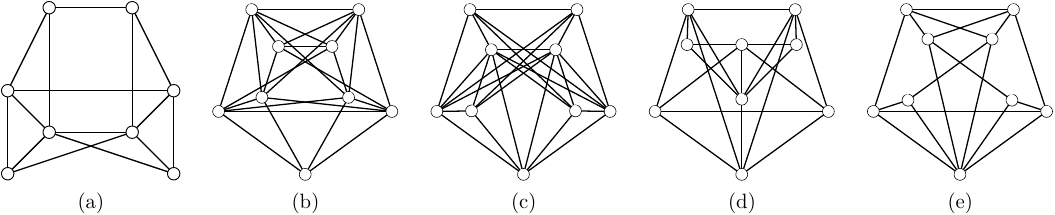}
\caption{The $\hbar$-imperfect graphs with no more than $9$ vertices which do not have other $\hbar$-imperfect graphs as induced subgraphs. Anticycles $\bar{C}_7$ and $\bar{C}_9$ are not listed here.}\label{fig:otherghs}
\end{figure}

Similarly, for all $h$-imperfect graphs with $9$ vertices, we first check whether $G$ has $\bar{C}_7$ or $G_8$ as its induced subgraphs. If this is the case, then $G$ is $\hbar$-imperfect. In total, there are $1414$ graphs in this case. Otherwise, $G$ cannot have other $\hbar$-imperfect graphs as induced subgraphs, since all $\hbar$-imperfect graphs with no more than $9$ vertices have either $\bar{C}_7$ or $G_8$ as an induced subgraph. Hence, if $G$ is reducible, then $G$ is $\hbar$-perfect.

There are only $2963$ graphs which are neither reducible nor induced subgraphs of $G_{15}$, for which we apply the numerical tools and find $2880$ ones to be $\hbar$-perfect (up to a $<1\times10^{-5}$ gap between the lower and upper bounds), $\bar{C}_9$ and other $4$ graphs as illustrated in Fig.~\ref{fig:otherghs}(b-e) to be $\hbar$-imperfect, and $78$ undetermined ones. For those $78$ undetermined ones, the gap between the upper bound of the weighted beta number and the corresponding weighted independence number is less than $3\times10^{-3}$.

\subsection{Further properties of the hierarchy}
\label{ssec:hierprop}

\begin{lemma}\label{lm:repdown}
For a given set of observables $\{S_i\}_{i=1}^n$ with a frustration graph $G$, let $S'_i = S_i\otimes M_i$ with $M_i = \otimes_{j=1}^n M_{ji}$ and $M_{ji}=Z$ if $j=i$, otherwise $M_{ji} = \id$. Then $\downarrow\hspace{-0.3em}{\cal Q}(\{S_i\}) \subseteq {\cal Q}(\{S'_i\}) \subseteq \be(G)$.
\end{lemma}
\begin{proof}
	By definition, for any point $p$ in ${\cal Q}(\{S_i\})$, it corresponds to a state $\rho$. More explicitly, the point can be represented as 
	$p = (\langle S_1\rangle_{\rho}^2, \ldots, \langle S_n\rangle_{\rho}^2)$.

	Let $\sigma = \rho \otimes (|0\rangle\langle 0|)^{\otimes n}$. The direct calculation shows that $\langle S'_i\rangle_{\sigma} = \langle S_i\rangle_{\rho}$. 
Hence,
	$p = (\langle S'_1\rangle_{\sigma}^2, \ldots, \langle S'_n\rangle_{\sigma}^2) \in {\cal Q}(\{S'_i\})$.
	Furthermore, let $\tau_x = x |0\rangle\langle 0| + (1-x)\id/2$ with $x\in [0,1]$, and $\sigma_{\vec{x}} = \rho\otimes(\otimes_{i=1}^n \tau_{x_i})$. Then we have $\langle S'_i\rangle_{\sigma_{\vec{x}}} = x_i \langle S_i\rangle_{\rho}$ which results in the point $p_{\vec{x}} = (x_1^2 \langle S_1\rangle_{\rho}^2, \ldots, x_n^2 \langle S_n\rangle_{\rho}^2)$. By varying the coefficient vector $\vec{x}$, we can obtain any point elementwise no more than $p$. That is, all those points are in ${\cal Q}(\{S'_i\})$. Consequently, we conclude that $\downarrow\hspace{-0.3em}{\cal Q}(\{S_i\}) \subseteq {\cal Q}(\{S'_i\})$.

 Finally, the two sets $\{S_i\}$ and $\{S'_i\}$ share the same frustration graph $G$, due to the fact that all matrices $M_i$ commute with each other. Consequently, ${\cal Q}(\{S'_i\}) \subseteq \be(G)$.
\end{proof}

\begin{theorem}
Given a graph $G$ with $n$ vertices, let $T_r$ be the projection of the feasible set of \eqref{mom} on the coordinates $\left(\langle x_i\rangle^2\right)_{i=1}^n$, that is,
\begin{equation}
    T_r\coloneqq\left\{\left(\langle x_i\rangle^2\right)_{i=1}^n\,\middle|\, M_r\succeq0, [M_r]_{1,1}=1, [M_r]_{u,v}=[M_r]_{a,b}\ (\forall\langle u^*v\rangle=\langle a^*b\rangle)\right\}.
\end{equation}
Then it holds that $\be(G) \subseteq T_r$. Besides, as $r$ goes to infinity, $T_r$ converges to $\be(G)$ in all nonnegative directions.
\end{theorem}
\begin{proof}
Let $\{S_i\}$ be any realization of the graph $G$. Then, the corresponding expectations with respect to any quantum state $\rho$ fulfill the constraints of \eqref{mom} and hence, ${\cal Q}(\{S_i\}) \subseteq T_r$.
Due to the arbitrariness of $\{S_i\}$, Lemma~\ref{lm:repdown} implies that $\downarrow\hspace{-0.3em} {\cal Q}(\{S_i\}) \subseteq T_r$. Then the convexity of $T_r$ leads to $\conv(\downarrow\hspace{-0.3em}{\cal Q}(\{S_i\})) \subseteq T_r$.
That is, $\be(G) \subseteq T_r$.

On the other hand, for each direction $w \ge 0$, the optimal value $\lambda_r(G,w)$ converges to $\beta(G,w)$. In other words, $T_r$ converges to $\be(G)$ in each nonnegative direction.
\end{proof}

\begin{theorem}
	The following program with $\lambda_r(G,w)$ defined in Eq.~\eqref{mom} is equivalent to an SDP:
\begin{align}\label{eq:deltaapp_inappendix}
	\omega_r \coloneqq\ &\max \sum\nolimits_i w_i\nonumber\\
	{\rm s.t.}\ & \lambda_r(G,w) \le 1,\nonumber\\
	     & w\ge 0.
\end{align}
\end{theorem}
\begin{proof}
Notice that the SDP in Eq.~\eqref{mom} can be reformulated as
\begin{align}\label{eq:original}
	\lambda_r(G,w) =\ &\sup \sum\nolimits_{i=1}^n w_i y_i\nonumber\\
	{\rm s.t.}\ & M_r = A_0 + \sum\nolimits_{j=1}^k y_j A_j \succeq 0,
\end{align}
where $A_j$'s are appropriate constant matrices encoding the equivalence relations and the Hermitian property of $M_r$ in Eq.~\eqref{mom}, and $A_0$ is the matrix with $[A_0]_{1,1}=1$ and all other elements being $0$. 
Then the dual SDP of Eq.~\eqref{eq:original} is
\begin{align}\label{eq:dual}
	\tilde{\lambda}_r(G,w)	\coloneqq\ &\inf \tr(Z A_0)\nonumber\\
	{\rm s.t.}\ & \tr(ZA_i) = -w_i, \ i=1, \ldots, n, \nonumber\\
		    & \tr(ZA_i) = 0, \ i=n+1, \ldots, k,\nonumber\\
            &\ Z\succeq0.
\end{align}
With tricks similar to those in \cite{josz2016strong}, one can show that there is no duality gap between Eq.~\eqref{eq:original} and Eq.~\eqref{eq:dual}, i.e. $\tilde{\lambda}_r(G,w) = \lambda_r(G,w)$. Then, for any feasible solution $w$ to Eq.~\eqref{eq:deltaapp_inappendix}, we have $\tilde{\lambda}_r(G,w)\le1$ and hence there exists a matrix $Z$ satisfying $\tr(ZA_0) \le 1$ as well as the conditions in Eq.~\eqref{eq:dual}.
Consequently, such a pair $(w,Z)$ is feasible to the following SDP:
\begin{align}\label{eq:omega}
	\tilde{\omega}_r\coloneqq\ &\sup \sum\nolimits_{i=1}^n w_i\nonumber\\
	{\rm s.t.}\ &\tr(ZA_i) = -w_i,\ i=1,\ldots,n,\nonumber\\
		    &\tr(ZA_i) = 0,\ i=n+1,\ldots,k,\nonumber\\
		    &\tr(ZA_0) \le 1,\nonumber\\
		    &\ w_i \ge 0,\quad Z\succeq0.
\end{align}
Conversely, any feasible solution $(w,Z)$ to Eq.~\eqref{eq:omega} is also a feasible solution to Eq.~\eqref{eq:dual}, implying $\lambda_r(G,w) =\tilde{\lambda}_r(G,w)  \le 1$. Therefore, such $w$ is feasible to Eq.~\eqref{eq:deltaapp_inappendix}. We thus conclude that $\omega_r = \tilde{\omega}_r$ and finish the proof.
\end{proof}
\subsection{Asymptotic prevalence of $\hbar$-imperfectness}
\label{ssec:asym}
\begin{theorem}
	The number of $\hbar$-perfect graphs with $n\gg 1$ vertices is at most $2^{c n(n-1)/2}$ for some constant $c<1$.
\end{theorem}
\begin{proof}
	Here we utilize the Erd\"{o}s-R\'{e}nyi random graph model, $G(n,p)$. In this model, there are $N = n(n-1)/2$ potential edges in $G(n,p)$, each of which is included with an independent and identical probability $p$. 

	We define a function $f(X_1, \ldots, X_N)$ as the number of induced subgraphs of $G(n,p)$ that are isomorphic to a specific $\hbar$-imperfect graph, $H_k$, on $k$ vertices, e.g., the anti-cycle graph $\bar{C}_7$ on $7$ vertices. 
	Notice that there are $C(n,k)$ ways to choose a subgraph on $k$ vertices from the random graph $G(n,p)$ on $n$ vertices, each of which is again a random graph $G(k,p)$. Thus, the expectation of $f(X_1, \ldots, X_N)$ is given by:
\begin{equation}
	\langle f(X_1, \ldots, X_N) \rangle = p_0\, C(n,k),
\end{equation}
where $p_0$ is the probability that the random graph $G(k,p)$ with $k$ vertices is isomorphic to $H_k$.

For an instance of the random graph to be $\hbar$-perfect, it must not contains any induced subgraph isomorphic to $H_k$, which means that the corresponding value of $f(X_1,\ldots,X_N)$ is $0$. Besides, the change in $f(X_1,\ldots,X_N)$ when the adjacent relation between two vertices switches is at most $c_e = C(n-2,k-2)$, since this operation can affect at most $c_e$ induced subgraphs containing these two vertices. Applying McDiarmid's inequality, we obtain the following upper bound on the probability of $f(X_1,\ldots,X_N) = 0$:
\begin{equation}\label{eq:mcd_ineq}
\begin{aligned}
	\text{Pr}\{f=0\} \le& \text{Pr}\{f\le \langle  f\rangle - p_0 C(n,k)\}\\
	\le& \exp({-2p_0^2 C(n,k)^2/[N c_e^2]})\\
	=& \exp({-2p_0^2 N/C(k,2)^2}),
\end{aligned}
\end{equation}
where the second inequality is due to McDiarmid's inequality ${\rm Pr}\{f\le \langle f\rangle - \epsilon\} \le \exp(-2\epsilon^2/[N c_e^2])$ with $\epsilon=p_0\, C(n,k)$ and $c_e$ being the upper bound of change in $f$ for the switch of each $X_i$.

For the case of $p = 1/2$, the expression in Eq.~\eqref{eq:mcd_ineq} represents the ratio of $\hbar$-perfect graphs to the total number of graphs, that is, $2^N$. The resulting number of $\hbar$-perfect graphs is upper bounded by $2^N  \exp(-{2p_0^2 N/C(k,2)^2} ) = 2^{c  N}$, with the constant $c$ given by
$1 - 2 p_0^2 \log e/C(k,2)^2$.
Since $p_0 > 0$, by setting $k$ to any number not less than $7$, it is clear that $c < 1$.
\end{proof}

\section{Applications}
\label{sec:app2}
\subsection{Entanglement detection and estimation}
\label{ssec:entanglement}
Here we use the Euclidean distance $d(\rho)$ from the point $\{\langle O_i\rangle\}$ to {$\be(G)$, with $G$ the frustration graph of Pauli strings $\{O_i\}$}, to provide a lower bound of ${\cal E}_{\rm HS}(\rho)$, i.e., the entanglement measured by the Hilbert-Schmidt distance~\cite{witte1999new}:
\begin{equation}
    {\cal E}_{\rm HS}(\rho) \coloneqq \min_{\sigma \in {\rm SEP}} (\tr[(\rho-\sigma)^2])^{1/2},
\end{equation}
with ${\rm SEP}$ stands for the set of separable states.

Since $\be(G)$ is a superset of the set of all points generated from separable states with the same observables $\{O_i\}$, for any separable state $\sigma$, it holds that
\begin{align}
	d^2(\rho) &= \min_{p\in \be(G)} \sum_i ||\langle O_i\rangle_{\rho}| - p_i |^2 \nonumber\\
    &\le \sum\nolimits_i ||\langle O_i\rangle_{\rho}| - |\langle O_i\rangle_{\sigma}| |^2 \nonumber\\
            &\le \sum\nolimits_i |\langle O_i\rangle_{\rho} - \langle O_i\rangle_{\sigma} |^2 \nonumber\\
	    &= \sum\nolimits_i (\tr[O_i (\rho - \sigma)])^2 \nonumber\\
	    &= \tr\big[ \big(\sum\nolimits_i O_i\otimes O_i\big) (\rho-\sigma)\otimes(\rho-\sigma)  \big]\nonumber\\
	    &\le \Big(\tr\big[ \big(\sum\nolimits_i\! O_i\otimes O_i\big)^2\big] \tr\big[ (\rho-\sigma)^2\otimes(\rho-\sigma)^2  \big] \Big)^{\frac{1}{2}}\nonumber\\
	    &\le \Big(\tr\big[ \big(\sum\nolimits_i O_i\otimes O_i\big)^2\big] \Big)^{\frac12}\tr\big[ (\rho-\sigma)^2\big] ,
\end{align}
which implies that
\begin{equation}
	{\cal E}_{\rm HS}(\rho) \ge d(\rho) / \Big(\tr\big[ \big(\sum\nolimits_i O_i\otimes O_i\big)^2\big] \Big)^{\frac{1}{4}}.
\end{equation}

For convenience, we relabel the four $2$-dimensional Bell states as
\begin{align}
	|\Psi_1\rangle = (|00\rangle+|11\rangle)/\sqrt{2},\,
	|\Psi_2\rangle = (|00\rangle-|11\rangle)/\sqrt{2}, \nonumber\\
	|\Psi_3\rangle = (|01\rangle+|10\rangle)/\sqrt{2},\,
	|\Psi_4\rangle = (|01\rangle-|10\rangle)/\sqrt{2}.
\end{align}
Then, by representing the $2^k$-dimensional space by $k$ qubits, we obtain the $4^k$ maximally entangled states represented by 
\begin{equation}
	|\Psi_{\vec{t}}\rangle_{1,\ldots,2i-1;2,\ldots,2i} \coloneqq \otimes_{i=1}^k |\Psi_{t_i}\rangle_{2i-1,2i},
\end{equation}
where $t_i \in \{1, 2, 3, 4\}$, $|\Psi_{t_i}\rangle_{2i-1,2i}$ is the maximally entangled state $|\Psi_{t_i}\rangle$ between the particles $2i-1$ and $2i$, the particles with odd labels belong to the first party, and those with even labels belong to the second party.

Denote by $\{\sigma_j\}_{j=0}^3$ the Pauli matrices $I, X, Y, Z$ in order.
Notice that $|\Psi_t\rangle$ is the common eigenstate of $\{\sigma_j\otimes\sigma_j\}_{j=0}^3$ with the eigenvalue being $\pm 1$ for any $t\in \{1,2,3,4\}$.
This implies that the state $|\Psi_{\vec{t}}\rangle$ is the common eigenstate of $\{S_l\otimes S_l\}_{l=0}^{4^k-1}$ with eigenvalues being $\pm 1$ for any $\vec{t}$, where $S_l$ is the Pauli string $\otimes_{i=1}^k \sigma_{l_i}$ and $l_i$ is the $i$-th digit of $l$ in the base of $4$.
Hence, $|\langle \Psi_{\vec{t}}| S_l\otimes S_l|\Psi_{\vec{t}}\rangle| = 1$, $\forall l$, which is impossible for separable states. Here we take $k=2$ as an example. Denote by $G$ the frustration graph of the $15$ operators $\{S_l\}_{l=1}^{15}$ by omitting the identity $S_0$, then $\beta(G)=3$. Thus, for any separable state $\sigma$, $\sum_{l=1}^{15} |\langle S_l\otimes S_l\rangle_{\sigma}| \le 3$.
\begin{theorem}
	For any Bell diagonal state $\rho = \sum_{i=1}^4 p_i |\Psi_i\rangle\langle \Psi_i|$ with $\sum_{i=1}^4 p_i = 1$ and $p_i \ge 0$, it is entangled if and only if $p({\cal S}, {\cal S}', \rho)$ is outside the set $\{(x,y,z)|x+y+z\le 1\} \cap \mathbb{R}^3_+$, where ${\cal S} = {\cal S'} = \{X, Y, Z\}$.
\end{theorem}
\begin{proof}
An known necessary and sufficient condition for $\rho$ to be entangled is $\max_i p_i > 1/2$. Without loss of generality, we assume that $p_1 \ge p_2 \ge p_3 \ge p_4$.
Notice that
\begin{align}
	|\Psi_1\rangle\langle \Psi_1| = (\id + XX - YY + Z Z)/4, \ 
	|\Psi_2\rangle\langle \Psi_2| = (\id - XX + YY + Z Z)/4, \nonumber\\
	|\Psi_3\rangle\langle \Psi_3| = (\id + XX + YY - Z Z)/4, \ 
	|\Psi_4\rangle\langle \Psi_4| = (\id - XX - YY - Z Z)/4.
\end{align}
The state $\rho$ can be reformulated as
\begin{equation}
	\rho = [\id + (p_1+p_3-p_2-p_4)XX + (p_2+p_3-p_1-p_4)YY + (p_1+p_2-p_3-p_4)Z Z]/4.
\end{equation}
Then the condition that $p({\cal S}, {\cal S}', \rho)$ is in the set $\{(x,y,z)|x+y+z\le 1\} \cap \mathbb{R}^3_+$  is equivalent to
\begin{equation}
	|p_1+p_3-p_2-p_4| + |p_2+p_3-p_1-p_4| + |p_1+p_2-p_3-p_4| \le 1.
\end{equation}
Under the assumption that $p_1\ge p_2 \ge p_3 \ge p_4$, it can be simplified to
\begin{align}
	1 &\ge p_1+p_3-p_2-p_4 + |p_1+p_4 - p_2 -p_3| + p_1 + p_2 - p_3 -p_4\nonumber\\
	  &= p_1+p_3-p_2-p_4 + |p_1+p_4 - p_2 -p_3| + p_1 + p_2 - p_3 -p_4\nonumber\\
	  &= 2p_1-2p_4 + |p_1+p_4-p_2-p_3|\nonumber\\
	  &= \max \{4p_1-1,1-4p_4\},
\end{align}
which is equivalent to $1/2 \ge p_1 = \max_i p_i$, i.e., $\rho$ is not entangled.
\end{proof}
We continue to consider the multipartite entanglement exemplified by GHZ states.
Let
\begin{align}
|\text{GHZ}_1\rangle = (|000\rangle + |111\rangle)/\sqrt{2},
|\text{GHZ}_2\rangle = (|000\rangle - |111\rangle)/\sqrt{2},\nonumber\\
|\text{GHZ}_3\rangle = (|100\rangle + |011\rangle)/\sqrt{2},
|\text{GHZ}_4\rangle = (|100\rangle - |011\rangle)/\sqrt{2},\nonumber\\
|\text{GHZ}_5\rangle = (|010\rangle + |101\rangle)/\sqrt{2},
|\text{GHZ}_6\rangle = (|010\rangle - |101\rangle)/\sqrt{2},\nonumber\\
|\text{GHZ}_7\rangle = (|110\rangle + |001\rangle)/\sqrt{2},
|\text{GHZ}_8\rangle = (|110\rangle - |001\rangle)/\sqrt{2}.
\end{align}
\begin{theorem}
	For any tripartite GHZ diagonal state $\rho = \sum_{i=1}^8 p_i |{\rm GHZ}_i\rangle\langle {\rm GHZ}_i|$ with $\sum_{i=1}^8 p_i = 1$ and $p_i\ge 0$, it is genuinely tripartite entangled if and only if  $\sum_{i=1}^7 |\langle O_i\rangle_{\rho}| > 3$, where $O_i$ are  
\begin{equation}\label{eq:stabilizer_inappendix}
	Z Z\id, Z\id Z, \id Z Z, XXX,-YYX, -YXY, -XYY.
\end{equation}
\end{theorem}
\begin{proof}
	It is known that $\rho$ is genuinely tripartite entangled if and only if $\max_i p_i > 1/2$. We first prove that the fact that $\rho$ is genuinely tripartite entangled implies $\sum_{i=1}^7 |\langle O_i\rangle_{\rho} | > 3$.

	Without loss of generality, we assume that $p_1 = \max_i p_i$, since all the $8$ GHZ states can be changed to each other via local unitaries, which does not affect entanglement.
	Notice that 
	\begin{equation}
		|{\rm GHZ}_1\rangle\langle {\rm GHZ}_1| = \big(\id_8 + \sum\nolimits_{i=1}^7 O_i\big)/8.
	\end{equation}
	Hence,
	\begin{equation}
		p_1 = \langle {\rm GHZ}_1|\rho|{\rm GHZ}_1\rangle = \big(1 + \sum\nolimits_{i=1}^7 \langle O_i\rangle_{\rho}\big)/8 > 1/2
	\end{equation}
	is equivalent to $\sum\nolimits_{i=1}^7 \langle O_i\rangle_{\rho} > 3$ and this implies that $\sum\nolimits_{i=1}^7 |\langle O_i\rangle_{\rho}| > 3$.

	Besides, whenever $\rho$ is not genuinely tripartite entangled, we know that $\sum\nolimits_{i=1}^7 |\langle O_i\rangle_{\rho}| \le 3$.
	Consequently, $\rho$ is genuinely tripartite entangled if and only if $\sum\nolimits_{i=1}^7 |\langle O_i\rangle_{\rho}| > 3$.
\end{proof}
As for the high-dimensional case, we use the following embedding technique when the dimension is not a power of $2$. Here we take the $3$-dimensional case to illustrate the idea.
Denote by $X^{(i)}, Y^{(i)}, Z^{(i)}$ the operators $X, Y, Z$ in the two-dimensional subspace without the $i$-th dimension. Then it holds for any separable state $\rho$  that
\begin{equation}\label{eq:subineq}
	\sum_{i,j=1}^3 \big[ |\langle X^{(i)} X^{(j)}\rangle_{\rho}| + |\langle Y^{(i)} Y^{(j)}\rangle_{\rho}| + |\langle Z^{(i)} Z^{(j)}\rangle_{\rho}|\big] \le 4.
\end{equation}
\begin{proof}
	Let $P_i = \id_3 - |i\rangle\langle i|$, $\rho_{ij} = (P_i\otimes P_j) \rho (P_i\otimes P_j) /\tr(P_i\otimes P_j \rho)$ if $\tr(P_i\otimes P_j \rho) \neq 0$. Otherwise, let $\rho_{ij} = \id_4/4$ the $4$-dimensional maximally mixed state. 
	By definition, $\rho_{ij}$ is a bipartite state with local dimensions being $2$. If $\rho$ is a separable state with the decomposition
	\begin{equation}
		\rho = \sum_k p_k \rho_{A,k} \otimes \rho_{B,k}.
	\end{equation}
	Then we have that
	\begin{equation}
		\rho_{ij} = \sum_k p_k (P_i\rho_{A,k}P_i) \otimes (P_j\rho_{B,k}P_j)/\tr(P_i\otimes P_j \rho)
	\end{equation}
	separable, too.
	This leads to
	\begin{equation}
		|\langle XX\rangle_{\rho_{ij}}|
		+|\langle YY\rangle_{\rho_{ij}}|
		+|\langle ZZ\rangle_{\rho_{ij}}| \le 1,
	\end{equation}
	which is equivalent to
\begin{equation}
	|\langle X^{(i)} X^{(j)}\rangle_{\rho}| + |\langle Y^{(i)} Y^{(j)}\rangle_{\rho}| + |\langle Z^{(i)} Z^{(j)}\rangle_{\rho}| \le \tr(P_i\otimes P_j \rho).
\end{equation}
Consequently,
\begin{equation}
	\sum_{i,j=1}^3 |\langle X^{(i)} X^{(j)}\rangle_{\rho}| + |\langle Y^{(i)} Y^{(j)}\rangle_{\rho}| + |\langle Z^{(i)} Z^{(j)}\rangle_{\rho}| \le \sum_{i,j=1}^3 \tr(P_i\otimes P_j \rho) = 4.
\end{equation}
\end{proof}
We first consider the state
\begin{equation}
    \rho(p) = (1-p) |\Psi_3\rangle\langle \Psi_3| + p |\Phi'\rangle\langle \Phi'|,
\end{equation}
where $|\Psi_3\rangle = (\sum_{i=1}^3 |ii\rangle)/\sqrt{3}$ and $|\Phi'\rangle = (|12\rangle + |21\rangle)/\sqrt{2}$. 
Direct calculation shows that the value of the left-hand side of Eq.~\eqref{eq:subineq} is
\begin{equation}
\frac{2}{3}(2 | 2-5 p| -5 p+8),
\end{equation}
whose minimum is $4$ and can only be achieved with $p=2/5$.

Then we consider a state that is entangled but unfaithful~\cite{Weilenmann2019EntanglementDB},
\begin{equation}
	\rho = 0.999 (0.50179 |\phi_1\rangle\langle \phi_1| + 0.49821|\phi_1\rangle\langle \phi_1|) + 0.001 \id/9,
\end{equation}
where
\begin{align}
	|\phi_1\rangle &= 0.628|11\rangle - 0.778 |22\rangle,\nonumber\\
	|\phi_2\rangle &= 0.807|01\rangle - 0.185|02\rangle  - 0.102 |10\rangle - 0.027|11\rangle + 0.011|12\rangle + 0.551|20\rangle - 0.024|21\rangle - 0.022|22\rangle.
\end{align}
Direct calculation shows that the value of the left-hand side of Eq.~\eqref{eq:subineq} is $4.8514 > 4$, i.e., a violation of the bound for separable states.
\subsection{Sample complexity in shadow tomography}\label{ssec:tomography}
\begin{theorem}
  Let ${\cal D}$ be the set of probability distributions and let $\alpha^*(G)$ and $\vartheta(G)$ be the fractional packing number~\cite{schrijver1979fractional} and the Lov\'{a}sz number~\cite{lovasz1979shannon} of graph $G$, respectively. It holds
  \begin{equation}
    \min_{w\in {\cal D}} \beta(G,w) \in [1/\alpha^*(\bar{G}), 1/\vartheta(\bar{G})],
  \end{equation}
where the lower bound is achieved whenever $G$ is $\hbar$-perfect.
\end{theorem}
\begin{proof}
  It is known that
  \begin{equation}
    \beta(G,w) \in [\alpha(G,w), \vartheta(G,w)],\, \forall w \in {\cal D}.
  \end{equation}
  We continue to prove that
  \begin{align}
    \min_{w\in{\cal D}} \alpha(G,w) &= 1/\alpha^*(\bar{G}),\label{eq:alphastar}\\
    \min_{w\in{\cal D}} \vartheta(G,w) &= 1/\vartheta(\bar{G}).
  \end{align}
  For $G$ to be $\hbar$-perfect, we have $\alpha(G,w) = \beta(G,w), \forall w\ge 0$, which implies that the lower bound is exact together with Eq.~\eqref{eq:alphastar}.
  Notice that
  \begin{align}
    \min_{w\in {\cal D}} \alpha(G,w) = \min_{w\ge 0} \frac{\alpha(G,w)}{\sum_i w_i}
                                     = \min_{w\ge 0, \alpha(G,w)=1} \frac{1}{\sum_i w_i}
                                     = \frac{1}{\alpha^*(\bar{G})},
  \end{align}
  where the last line is due to the definition of $\alpha^*(\bar{G})$~\cite{schrijver1979fractional}, i.e.,
  \begin{align}
    \alpha^*(G) \coloneqq \max &\quad\sum_i w_i\nonumber\\
    \mathrm{s.t.}&\quad w\ge 0,\nonumber\\
         &\quad\sum_{i\in C} w_i \le 1, \,\forall C \in {\rm Clis },
  \end{align}
  with ${\rm Clis}$ being the set of cliques of $\bar{G}$, or equivalently, the set of independence sets of $G$.

  Similarly,
  \begin{align}
    \min_{w\in {\cal D}} \vartheta(G,w) = \min_{w\ge 0} \frac{\vartheta(G,w)}{\sum_i w_i} = \min_{w\ge 0, \vartheta(G,w)=1} \frac{1}{\sum_i w_i} = \frac{1}{\vartheta(\bar{G})},
  \end{align}
  where the last line is due to the fact that~\cite{knuth1994sandwich}
  \begin{align}
                                     & {\max_{w\ge 0, \vartheta(G,w)=1} \sum_i w_i}\nonumber\\
  =& {\max_{w\ge 0, \sum_i w_i v_i\le 1, \forall v\in {\rm TH}(G)} \sum_i w_i}\nonumber\\
                                     =& {\max_{w\in {\rm TH}(\bar{G})} \sum_i w_i}\nonumber\\
                                     =& \,\vartheta(\bar{G}),
  \end{align}
  since
  \begin{equation}
  {\rm TH}(\bar{G}) = \{w\mid w\ge 0, \sum_i w_i v_i \le 1, \forall v \in {\rm TH}(G)\}.
  \end{equation}
\end{proof}

\begin{lemma}
  $\beta(G,w)$ is a convex function of $w \ge 0$.
\end{lemma}
\begin{proof}
  By definition, for any $w\ge 0, \tilde{w}\ge 0$ and $x\in [0,1]$, we have
  \begin{align}
   & \,\beta(G,x w + (1-x) \tilde{w})\nonumber\\ 
    =& \max_{\rho}[x \sum_i w_{i} \langle S_i\rangle_{\rho}^2 + (1-x) \sum_i \tilde{w}_{i} \langle S_i\rangle_{\rho}^2]\nonumber\\
    \le& \,x\max_{\rho} \sum_i w_{i} \langle S_i\rangle_{\rho}^2 + (1-x) \max_{\rho}\sum_i \tilde{w}_{i} \langle S_i\rangle_{\rho}^2\nonumber\\
    =& \,x \beta(G,w) + (1-x)\beta(G,\tilde{w}).
  \end{align}
\end{proof}
\begin{theorem}
Denote ${\cal D}$ the set of all probability distributions and ${\cal D}_s$ the set of the ones which has the symmetry induced by graph $G$, then
  \begin{equation}
    \min_{w\in {\cal D}} \beta(G,w) = \min_{w\in {\cal D}_s} \beta(G,w),
  \end{equation}
  Especially, when $G$ is vertex-transitive, $\min_{w\in {\cal D}} \beta(G,w) = \beta(G)/n$ with $n$ being the number of vertices of $G$.
\end{theorem}
For example, all cycle graphs and their complements are vertex-transitive. For a given cycle graph $C_n$ with $n$ vertices, $\beta(C_n) = \alpha(C_n) = \lfloor n/2 \rfloor$~\cite{xu2023bounding}. Thus, the sample complexity parameter $\delta$ is $\beta(C_n)/n = \lfloor n/2 \rfloor/n \approx 1/2$ when $n$ is large enough. 
For a given anti-cycle graph $\bar{C}_n$ with $n$ vertices, $2= \alpha(\bar{C}_n) < \beta(\bar{C}_n) \le \vartheta(\bar{C}_n)$, where $\vartheta(\bar{C}_n) = 1+1/\cos(\pi/n)$ when $n$ is odd and $\vartheta(\bar{C}_n) = 2$ when $n$ is even~\cite{lovasz1979shannon}. Thus, $\beta(\bar{C}_n)\approx 2$ and the sample complexity parameter $\delta$ is $\beta(\bar{C}_n)/n \approx 2/n$ when $n$ is large enough. Since the sample complexity is in $[\Omega(1/[\epsilon^2 \delta]), O(\log n/[\epsilon^2 \delta])]$ for the precision $\epsilon$, the sample complexity for the anti-cycle $\bar{C}_n$ is higher than the one for the cycle $C_n$. More explicitly, the ratio is at least $\Omega(n/\log n)$ for the same precision.

For a random graph $G(n,p)$ with $n$ vertices and $p$ is the probability of each edge~\cite{coja2005lovasz}, its complement graph is $G(n,1-p)$. Then the sample complexity for $G=G(n,p)$ is lower bounded by 
\begin{equation}
    \Omega(1/[\epsilon^2 \delta]) = \Omega(\vartheta(\bar{G})/\epsilon^2)  = \Omega(\vartheta(G(n,1-p))/\epsilon^2) = \Omega(\sqrt{n/(1-p)}/\epsilon^2),
\end{equation}
where the last equality is from Theorem~4 in Ref.~\cite{coja2005lovasz}.

\subsection{Using classical graph algorithms for bounding ground state energies}\label{ssec:ground}
For a sequence of real-valued coefficients $a_i$, let $H=\sum_i a_i A_i$ be a Hamiltonian. 
It is a central, but at the same time very difficult, task to find or estimate the smallest or the highest energy of $H$ generally. Given $\beta(G,w)$ such an estimate from the outside, this is the direction not covered by variation methods, can be established. 
We have \cite{xu2023bounding}
\begin{align}
    \label{eq:approx2_inappendix}
   \sup_\rho \langle H \rangle_\rho  &\leq \sqrt{\inf_w  \Big( \sum_i a_i^2/w_i\Big)\beta(G,w)},
\end{align}
which is equivalent to 
\begin{align}
    \label{eq:approx2__inappendix2}
   \sup_\rho \langle H \rangle_\rho  &\leq \sqrt{\inf_{w, \beta(G,w)= 1}   \sum_i a_i^2/w_i}\nonumber\\
   &= \sqrt{\inf_{w, \beta(G,w)\le 1}   \sum_i a_i^2/w_i},
\end{align}
where the last equality is from the fact that the infimum is obtained when $\beta(G,w)=1$ in the region that $\beta(G,w)\le 1$.

Here, having a $\hbar$-perfect graph again simplifies the optimization over $w$ in the above to an optimization over the stable set polytope.
To find the upper bound in Eq.~\eqref{eq:approx2_inappendix}, it is equivalent to considering the following optimization problem:
\begin{align}
    \label{eq:approx3_inappendix}
    \inf_w   &\quad\sum_i a_i^2/w_i\nonumber\\
    \mathrm{s.t.} &\quad \sum_{i\in I} w_i \le 1,\, \forall I \in {\cal I},\nonumber\\
    & \quad w_i \ge 0,
\end{align}
where ${\cal I}$ contains all independence sets of $G$.
It is clear to see that all the constraints in Eq.~\eqref{eq:approx3_inappendix} are linear. Besides, $\sum_i a_i^2/w_i$ is a convex function of $w$ since $1/w_i$ is a convex function of $w_i$ for $w_i>0$. Consequently, the programming in Eq.~\eqref{eq:approx3_inappendix} is a convex program with linear constraints, which in principle can be solved efficiently.

Here, we consider the following model as an example:
\begin{equation}
H_n =	a X_1+b Z_1 + c Y_{n} + \sum_{i=1}^{n-1}(d_i X_iX_{i+1} + e_i Z_iZ_{i+1}),
\end{equation}
with coefficients $u = (a,b,c,d_1,\ldots,d_{n-1},e_1,\ldots,e_{n-1})$.

The frustration graph of the Pauli strings in $H_n$ is the cycle graph $C_{2n+1}$, which is $h$-perfect and hence $\hbar$-perfect. Then we have
\begin{align}
	\langle H_n \rangle^2 \le \inf_{w\ge 0}\sum_i u_i^2 /w_i \alpha(C_{2n+1},w) \le n\sum_i u_i^2,
\end{align}
where the last inequality holds since we simply set $w_i=1$.
In the case that $u_i=1$, that is, all coefficients are $1$, $\langle H_n\rangle \ge -\sqrt{n(2n+1)}$. 
For $n\in \{2,3,4,5,6\}$, we calculate numerically the ground state energy and compare it with the estimation. The ratio of the estimation over the exact value is
\begin{equation}
1.02749, 1.02943, 1.01672, 1.00589, 1.00416, 1.00121
.\end{equation}
The difference between the estimation and the exact value is
\begin{equation}
0.0845941, 0.13099, 0.0986868, 0.0434194, 0.036549, 0.0123609
.\end{equation}
As $n$ increases, the ratio tends to $1$ and the difference tends to $0$.

Another model is 
\begin{equation}
	\tilde{H}_n =	a X_1+b Z_1 + c Y_{n} + \sum_{i=1}^{n-1}(d_i X_iX_{i+1} + e_i Y_iY_{i+1} + f_i Z_iZ_{i+1}),
\end{equation}
with coefficients $u = (a,b,c,d_1,\ldots,d_{n-1},e_1,\ldots,e_{n-1},f_1,\ldots,f_{n-1})$.
As confirmed by direct calculation, the corresponding frustration graph is also $h$-perfect for $n\le 5$. Hence, there is also a convex program for the estimation of its ground state energy. We use the case of $n=2$ to illustrate the idea of obtaining the weight vector $w$ for a tighter bound. The idea is to increase the components of $w$ without increasing the weighted independence number. 
Let
\begin{equation}
	\tilde{H}_2 = X_1 + Z_1 + Y_2 + \left( X_1X_2 + Y_1Y_2 + Z_1Z_2 \right).
\end{equation}
Let $G$ be the frustration graph of the Pauli strings in $\tilde{H}_2$ in order.
Then
\begin{align}
	\langle \tilde{H}_2\rangle^2 \le \inf_{w\ge 0} \sum_{i=1}^6 1/w_i \beta(G,w) \le 6 \beta(G) = 18.
    \end{align}
However, by taking $w=(2,2,1,1,1,1)$, $\beta(G,w) = \beta(G)=3$ while $\sum_{i=1}^6 1 /w_i = 5$. This provides a tighter bound $\langle \tilde{H}_2\rangle^2 \le 15$. In comparison, the exact ground state energy is $-3.722935$, while the estimated lower bound is $-\sqrt{15} \approx -3.872983$, which is better than the one $-\sqrt{18} \approx -4.242641$. 

As far as we know, there is also a numerical method based on SDP relaxations to provide lower bounds on the ground state energy~\cite{wang2024certifying}. However, such a method is feasible for middle- and large-scale models only if there are lots of symmetries in the model, which ensures that a good enough estimation can be obtained within reasonable computational resources. In comparison, the graph-theoretic method still works even if there is not enough symmetry in the model at the price of precision, due to the relaxation from the Cauchy-Schwarz inequality. Notice that the symmetry in the model can also be employed to reduce the complexity in the calculation of weighted independence numbers, and hence speed up the graph-theoretic estimation.

\subsection{Quantum encoding of the independence number}\label{ssec:alpha}
\begin{theorem}
    Let a graph with $n$ vertices be encoded in $l=1/2 \log_2(n/c)$ qubits. Assume weights normalized by $\Vert w \Vert_1=n$. Then, the maximal eigenvalue of $\tilde{H}_S^{(m)}$ approximates $\beta(G,w)$ by 
    \begin{align}\label{eq:definetti-without-proj}
        \beta(G,w) \leq \lambda_{\max}\left( \tilde{H}_S^{(m)} \right) \leq \beta(G,w) + n\sqrt{ \frac{\log(n)-\log(c)}{m}}.
    \end{align}
 That is, we obtain a linear Hamiltonian encoding of the number $\beta(G,w)$ with an additive regularized error 
    \begin{align}
        \varepsilon =\left|\lambda_{\max}(\tilde{H}_S^{(m)})/n -\beta(G,w)/n\right|
    \end{align}
    by using $L$ qubits with
    \begin{align}
        L=(m+2)l\sim  O\left(\frac{\log(n/c)^2 }{\varepsilon^2}\right).
    \end{align}
\end{theorem}
\begin{proof}
Our system $\mathcal{H}_L$ under consideration consists of $L=(m+2)l$ qubits which we will consistently regard as a tensor product $\mathcal{H}_l^{(m+2)}$ of $(m+2)$ copies of a system $\mathcal{H}_l = \mathbb{C}^{2^l}$. We will refer to each factor of type $\mathcal{H}_l$ as a local system.  

Let $\rho_{\max}$ be a state on such $L$ qubits that attains the maximal eigenvalue of  $\tilde{H}_S^{(m)}$. 
W.l.O.g. we assume that $\rho_{\max}$ inherits the permutation invariance from  $\tilde{H}_S^{(m)}$. 
We start by decomposing the maximal eigenvalue as
\begin{equation}
\lambda_{\max}\bigl(\tilde{H}_S^{(m)}\bigr)
= \tr\Bigl(\tilde{H}_S^{(m)}\rho_{\max}\Bigr)
= \tr\Bigl(\tilde{H}_S^{(m)}\bigl(\rho_{\max}-\!\!\int \sigma^{\otimes (m+2)}\,d\mu(\sigma)\bigr)\Bigr)
  + \tr\Bigl(\tilde{H}_S^{(m)}\!\int \sigma^{\otimes (m+2)}\,d\mu(\sigma)\Bigr).
\end{equation}
Since $\tilde{H}_S^{(m)}$ is a composition of ${m+2 \choose 2}n$ Pauli strings acting only on two local factors and each Pauli string $S$ corresponds to a measurement with two outcomes, it follows that
\begin{equation}
\tr\bigl(S(\rho_{\max}-\!\!\int \sigma^{\otimes (m+2)}\,d\mu(\sigma))\bigr)
= \tr\bigl(S^{(2)}(\rho^{(2)}_{\max}-\!\!\int \sigma^{\otimes 2}\,d\mu(\sigma))\bigr)
\le \bigl\lVert\rho^{(2)}_{\max}-\!\!\int \sigma^{\otimes 2}\,d\mu(\sigma)\bigr\rVert_{\mathrm{LOCC}_1},
\end{equation}
where $\bigl\lVert \tau\bigr\rVert_{\mathrm{LOCC}_1} = \max_{\{M_i\}\in \mathrm{LOCC}_1} \bigl\lVert \sum_i \tr(M_{i}\tau)|i\rangle\langle i| \bigl\lVert_1$ with $\mathrm{LOCC}_1$ being the set of one-way LOCC measurement~\cite{li2015quantum}, 
$S^{(2)}$ denotes $S$ on its bipartite support, $\rho^{(2)}_{\max}$ denotes the bipartite reduced state on the same support, which is well-defined since $\rho$ is permutationally invariant.

Summing over the $n$ non-equivalent contributions gives
\begin{equation}
\tr\Bigl(\tilde{H}_S^{(m)}\bigl(\rho_{\max}-\!\!\int \sigma^{\otimes (m+2)}\,d\mu(\sigma)\bigr)\Bigr)
\le n\,\bigl\lVert\rho^{(2)}_{\max}-\!\!\int \sigma^{\otimes 2}\,d\mu(\sigma)\bigr\rVert_{\mathrm{LOCC}_1}.
\end{equation}
For the second term, one has
\begin{align}
\tr\Bigl(\tilde{H}_S^{(m)}\!\int \sigma^{\otimes (m+2)}\,d\mu(\sigma)\Bigr)
&= \int \tr\bigl(\tilde{H}_S^{(m)}\sigma^{\otimes (m+2)}\bigr)\,d\mu(\sigma)
= \int \tr\bigl(\tilde{H}_S\sigma^{\otimes 2}\bigr)\,d\mu(\sigma)\nonumber\\
&\leq \max_{\sigma} \tr\bigl(\tilde{H}_S\sigma^{\otimes 2}\bigr)
=\beta(G,w).
\end{align}
Combining the two inequalities, we obtain
\begin{equation}
\lambda_{\max}\bigl(\tilde{H}_S^{(m)}\bigr)
\le n\,\bigl\lVert\rho^{(2)}_{\max}-\!\!\int \sigma^{\otimes 2}\,d\mu(\sigma)\bigr\rVert_{\mathrm{LOCC}_1}
+ \beta(G,w).
\end{equation}
By the finite de Finetti theorem for one-way LOCC measurements~\cite[Thm.~1]{li2015quantum}, the $\mathrm{LOCC}_1$ distance satisfies
\begin{equation}
\bigl\lVert\rho^{(2)}_{\max}-\!\!\int \sigma^{\otimes 2}\,d\mu(\sigma)\bigr\rVert_{\mathrm{LOCC}_1}
\le \sqrt{\frac{2\ln d}{m}}
= \sqrt{\frac{2\ln(\sqrt{n/c})}{m}}
= \sqrt{\frac{\ln(n/c)}{m}},
\end{equation}
and hence
\begin{equation}
\lambda_{\max}\bigl(\tilde{H}_S^{(m)}\bigr)
\le \beta(G,w) + n\sqrt{\frac{\ln(n/c)}{m}}.
\end{equation}

On the other hand, we have 
\begin{align}
   \beta(G,w) = \max_{\sigma} \tr\bigl(\tilde{H}_S\sigma^{\otimes 2}\bigr)= \max{\sigma} \tr\bigl(\tilde{H}_S^{(m)} \sigma^{\otimes (m+2) }\bigr)
&\leq \lambda_{\max} \left( \tilde{H}_S^{(m)} \right) 
\end{align}
which completes the proof of \eqref{eq:definetti-without-proj}.
Setting 
\begin{equation}
m \le \frac{\ln(n/c)}{\epsilon^2}.
\end{equation}
in \eqref{eq:definetti-without-proj} yields 
\begin{align}
    \lambda_{\max}\bigl(\tilde{H}_S^{(m)}\bigr) - \beta(G,w) \leq n \epsilon.
\end{align}

Since $l = \log d = \log(n/c)$, the total length $L = (m+2)l$ satisfies
\begin{equation}
L = (m+2)l = O\!\left(\frac{\log^2(n/c)}{\epsilon^2}\right),
\end{equation}
which completes the proof.
\end{proof}